\documentclass[12pt]{article}

\usepackage{amsmath}
\usepackage{graphicx,psfrag,epsf, enumerate, natbib, url} 
 \usepackage[usenames]{color}
\usepackage{geometry,float,sectsty,amstext,relsize,floatpag,wrapfig,listings,subfigure,amsfonts,amsthm,mathabx,enumitem,hyperref,grffile,colortbl,dsfont,bbm,setspace, overpic}
\usepackage[bottom]{footmisc}

\usepackage{tabularx}

\usepackage{multibib} 
\newcites{si}{Additional References for the Supplementary Material}

\usepackage{titlesec}
\titlelabel{\thetitle.\quad}
\renewcommand\thesection{\arabic{section}}
\renewcommand{\thesubsection}{\thesection.\arabic{subsection}}
\sectionfont{\centering\sf\normalsize\bfseries}
\subsectionfont{\sf\normalsize}

\newcommand{\single}{\renewcommand{\baselinestretch}{1.2}\normalsize}
\newcommand{\double}{\renewcommand{\baselinestretch}{1.63}\normalsize}
\newcommand{\bc}{\begin{center}}
\newcommand{\ec}{\end{center}}
\newcommand{\no}{\noindent}

\DeclareMathOperator*{\argmin}{argmin}

\definecolor{LightGray}{gray}{0.9}
\newcommand\numberthis{\addtocounter{equation}{1}\tag{\theequation}}
\newcommand{\papertitle}{Modeling Persistent Trends in Distributions}

\newtheorem{thm}{Theorem}
\newtheorem{lem}{Lemma}
\newcounter{factnum}
\newcounter{claimnum}
 \newcounter{defnum}
\newtheorem{definition}[defnum]{Definition}
\newcounter{assumptioncount}

\newcommand{\beginsupplement}{ 
        \setcounter{section}{0}
        \renewcommand{\thesection}{S\arabic{section}} %
         \renewcommand{\thesubsection}{\thesection.\arabic{subsection}}
        \setcounter{table}{0}
        \renewcommand{\thetable}{S\arabic{table}} %
        \setcounter{figure}{0}
        \renewcommand{\thefigure}{S\arabic{figure}} %
     }

  \let\oldthebibliography=\thebibliography
  \let\endoldthebibliography=\endthebibliography
  \renewenvironment{thebibliography}[1]{
    \begin{oldthebibliography}{#1} 
      \setlength{\parskip}{0ex}
      \setlength{\itemsep}{0ex}
  }
  {
    \end{oldthebibliography}
  }

\begin{document}
\thispagestyle{empty} \single \bc {\bf \sc \Large \papertitle}\vspace{0.15in}\\
Jonas Mueller, Tommi Jaakkola, and David Gifford \\
MIT Computer Science \& Artificial Intelligence Laboratory \\
Cambridge, MA 02139
 
 \ec 
 
 \vspace{0.4in} \thispagestyle{empty}
  \bc{\sf  \textbf{Abstract}} \ec \vspace{-.1in} \no
We present a nonparametric framework to model a short sequence of probability distributions that vary both due to underlying effects of sequential progression and confounding noise.  To distinguish between these two types of variation and estimate the sequential-progression effects, our approach leverages an assumption that these effects follow a persistent trend.  This work is motivated by the recent rise of single-cell RNA-sequencing experiments over a brief time course, which aim to identify genes relevant to the progression of a particular biological process across diverse cell populations.  While classical statistical tools focus on scalar-response regression or order-agnostic differences between distributions, it is desirable in this setting to consider both the full distributions as well as the structure imposed by their ordering.  We introduce a new regression model for ordinal covariates where responses are univariate distributions and the underlying relationship reflects consistent changes in the distributions over increasing levels of the covariate.   This concept is formalized as a \emph{trend} in distributions, which we define as an evolution that is linear under the Wasserstein metric.  Implemented via a fast alternating projections algorithm, our method exhibits numerous strengths in simulations and analyses of single-cell gene expression data.

 \vspace{0.7in}
\no {KEY WORDS:\quad Wasserstein distance, batch effect, quantile regression, pool adjacent violators algorithm, single cell RNA-seq}.
\thispagestyle{empty} \vfill

 

\newpage
\pagenumbering{arabic} \setcounter{page}{1} \double

\section{Introduction}
\label{sec:intro}

A common type of data in scientific and survey settings consists of real-valued observations sampled in batches, where each batch shares a common label (this numerical/ordinal value is the \emph{covariate}) whose effects on the observations are the item of interest.   When each batch consists of a large number of i.i.d.\ observations, the empirical distribution of the batch may be a good approximation of the underlying population distribution conditioned on the value of the covariate.  A natural goal in this setting is to quantify the covariate's effect on these conditional distributions, considering changes across all segments of the population.  In the case of high-dimensional observations, one can measure this effect separately for each variable to identify which are the most interesting.  However, it may often occur that, in addition to random sampling variability, there exist unmeasured confounding variables (unrelated to the covariate) that affect the observations in a possibly dependent manner within the same batch (cf.\ \emph{batch effects} in \citeauthor{Risso2014} \citeyear{Risso2014}).  

The primary focus of this paper is the introduction of the TRENDS (Temporally Regulated Effects on Distribution Sequences) regression model, which infers the magnitude of these covariate-effects across entire distributions.  TRENDS is an extension of classic regression with a single covariate (typically of fixed-design), where one realization of our dependent variable is a batch's entire empirical distribution (rather than a scalar) and the condition that fitted-values are smooth/linear in the covariate is replaced by the condition that fitted distributions follow a \emph{trend}.  Formally defined in \S\ref{sec:trenddef}, a trend describes a sequence of distributions where the $p^{\text{th}}$ quantile evolves monotonically for all $p \in (0,1)$, though not necessarily in the same direction for different $p$, and there are at most two partitions of the quantiles that move in opposite directions.  Thus, TRENDS extends scalar-valued regression to full distributions while retaining the ability to distinguish effects of interest from extraneous noise.  

Despite the generality of our ideas, we motivate TRENDS with a concrete scientific application: the analysis of single-cell RNA-sequencing time course data (see \S\ref{sec:acs} for a different application to income data; references preceded by `S' are in the Supplementary Material).

The recent introduction of single-cell RNA-seq (scRNA-seq) techniques to obtain transcriptome-wide gene expression profiles from individual cells has drawn great interest \citep{GeilerSamerotte2013}.  Previously only measurable in aggregate over a whole tissue-sample/culture consisting of thousands of cells, gene-expression at the single-cell level offers insight into biological phenomena at a much finer-grained resolution, and is important to quantify as even cells of the same supposed type exhibit dramatic variation in morphology and function.  One promising experimental design made feasible by the advent of this technology involves sampling groups of cells at various times from  tissues / cell-cultures undergoing development and applying scRNA-seq to each sampled cell \citep{Trapnell2014, Buettner2015}.  It is hoped that these data can reveal which  \emph{developmental} genes regulate/mark the emergence of new cell types over the course of development.  

Current scRNA-seq cost/labor constraints prevent dense sampling of cells continuously across the entire time-continuum.  Instead, researchers target a few time-points, simultaneously isolating sets of cells at each time and subsequently generating RNA-seq transcriptome profiles for each individual cell that has been sampled.  More concretely, from a cell population undergoing some biological process like development, one samples $N_\ell \ge 1$ batches of cells from the population at time $t_\ell$ where $\ell = 1,2,\dots, L$ indexes the time-points in the experiment and $i=1,\dots, N = \sum_{\ell=1}^L N_\ell$ indexes the batches.  Each batch consists of $n_i$ cells sampled and sequenced together. We denote by $x^{(g)}_{i, s} \in \mathbb{R}$ the measured expression of gene $g$ in the $s$th cell of the $i$th batch ($1 \le s \le n_i$), sampled at time $t_{\ell_i}$.  

Because expression profiles are restricted to a sparse set of time points in current scRNA-seq experiments, the underlying rate of biological progression can drastically differ between equidistant times.  Thus, changes in the expression of genes regulating different parts of this process may be highly nonuniform over time, invalidating assumptions like linearity or smoothness.  One common solution in standard tissue-level RNA-seq time course analysis is time-warping \citep{Bar-Joseph2003}.  Since our interest lies not in predicting gene-expression at new time-points, we instead aim for a procedure that respects the sequence of times without being  sensitive to their precise values.  In fact, researchers commonly disregard the wall-clock time at which sequencing is done, instead recording the experimental chronology as a sequence of stages corresponding overall qualitative states of the biological sample.  For example, in \citet{Deng2014}: Stage 1 is the oocyte, Stage 2 the zygote, \dots, Stage 11 the late blastocyst.  Attempting to impose a common scale on the stage numbering is difficult because the similarity in expression expected across different pairs of adjacent stages might be highly diverse for different genes.  In this work, we circumvent this issue by disregarding the time-scale and $t_\ell$ values, instead working only with the ordinal levels $\ell$ (so the only information retained about the times is their order $t_1 < t_2 < \dots < t_L$ ), as done by \citet{Bijleveld1998} (Section 2.3.2). 

\begin{figure}[h!] \centering
\includegraphics[width= 0.6\textwidth]{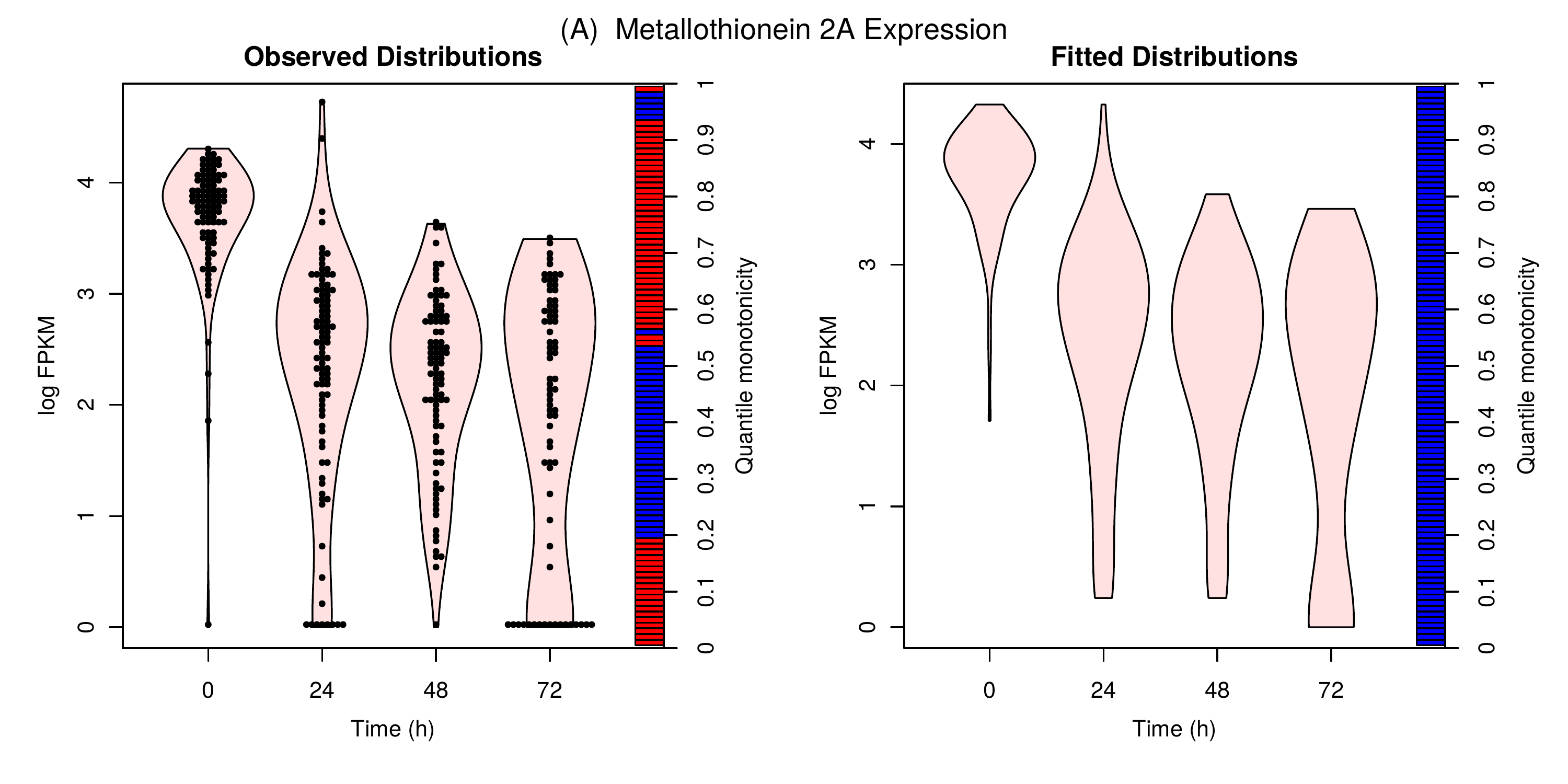} \\
\vspace*{-3mm}
\includegraphics[width= 0.6\textwidth]{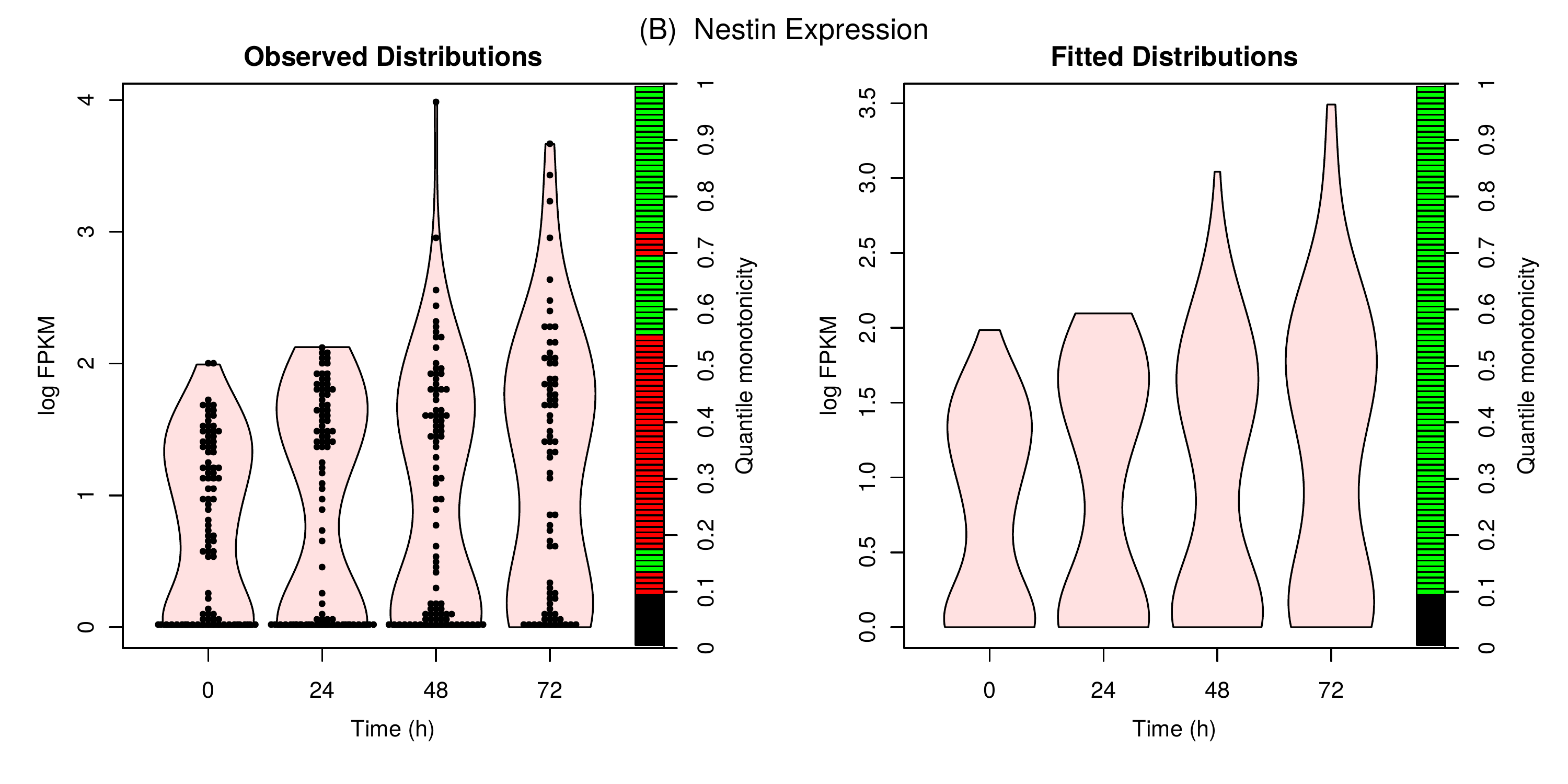} \\
\vspace*{-2mm}
\includegraphics[width=0.28 \textwidth]{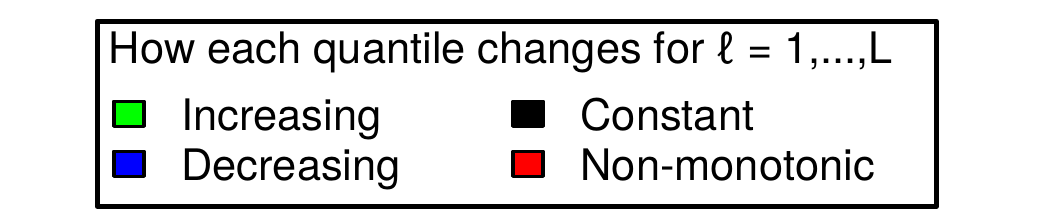}
\caption{Violin plots (kernel density estimates) depicting the empirical distribution of known developmental genes' expression measured in myoblast cells (on left), and the corresponding TRENDS fitted distributions (on right). Each point shows a sampled cell.}
\label{geneexamples}
\end{figure}

Depictions of such data from two genes (where $N_\ell = 1$ for each $\ell$) are shown in the lefthand panels of Figure \ref{geneexamples}.  Lacking longitudinal measurements, these data differ from those studied in time series analysis:  at each time point, one observes a different group of numerous exchangeable samples (no cell is profiled in two time points), and also the number of time points is small (generally $L < 10$).  As a result of falling RNA-seq costs, multiple cell-capture plates (each producing a batch of sampled cells, i.e.\  $N_\ell > 1$) are being used at each time point to observe larger fractions of the cell population \citep{Zeisel2015}.   Because the cells in a batch are simultaneously collected and sequenced (independently of other batches), the measured gene-expression values are often biased by \emph{batch effects}: technical artifacts that perturb observed values in a possibly correlated fashion between cells of the same batch  \citep{Risso2014, Kharchenko2014}.   Rather than treating the cells from a single time point identically, it is desirable to retain  batch information and account for this nuisance variation.  Batch effects are also prevalent in other applications including temporal studies of demographic statistics, where a simultaneously-collected group of survey results may be biased by latent factors like location.

Furthermore, cell populations can exhibit enormous  heterogeneity, particularly in developmental or in vivo settings \citep{Trapnell2014, Buettner2015}.  A few high-expression cells often bias a population's average expression, and transcript levels can vary 1,000-fold between seemingly equivalent cells \citep{GeilerSamerotte2013}\footnote{{\citeauthor{GeilerSamerotte2013} lament: ``analyzing gene expression in a tissue sample is a lot like measuring the average personal income throughout Europe -- many interesting and important phenomena are simply invisible at the aggregate level.  Even when phenotypic measurements have been meticulously obtained from single cells or individual organisms, countless studies ignore the rich information in these distributions, studying the averages alone''.}}.  By fitting a TRENDS model (which accounts for both batch effects and the full distribution of expression across cells) to each gene's expression values, researchers can rank genes based on their presumed developmental relevance or employ hypothesis testing to determine whether observed temporal variation in expression is biologically relevant.  

\section{Related Work}
\label{sec:relatedwork}

To better motivate the ideas subsequently presented in this paper, we first describe why existing methods are not suited for scRNA-seq time course experiments and similar ordered-batched data lacking longitudinal  measurements.  As an alternative to time-series techniques, regression models might be applied in this setting, such as the Tobit generalized linear model of \citet{Trapnell2014}.  However, these models rely on  linearity/smoothness assumptions, which can be inappropriate for  sporadic processes such as development.  More importantly, classic regression models scalar values such as conditional expectations, for which results must be interpreted as the effects in a hypothetical ``average cell''.  

Rather than focusing only on (conditional) expectations or a few quantiles, it is often more appropriate to model the full (conditional) distribution of values in a heterogeneous population \citep{GeilerSamerotte2013, Buettner2015}.  Let $P_{\ell}$ denote the underlying distribution of the observations from covariate-level $\ell$.   An omnibus test for distribution-equality ($H_0: P_{1} = \dots = P_{L} $ vs.\ the alternative that they are not all identical, cf.\ the Komogorov-Smirnov method described in \S\ref{sec:othermthds}) can capture arbitrary changes, but fails to  reflect sequential dynamics.  Significance tests also do not quantify the size of effects, only the evidence for their existence.  \citet{Krishnaswamy2014} have proposed  a mutual-information based measure (DREMI) to quantify effects, which could be applied to our setting.  However, under systematic noise caused by batch effects, measures of general statistical dependence between the batch-values and label $\ell$ (e.g.\  mutual information or hypothesis testing) become highly susceptible to the spurious variation present in the observed distributions (resulting in false positives).  We thus prefer borrowing strength  in the sense that a consistent change in distribution should ideally be observed across multiple time points for an effect to be deemed significant.

Instead of these general approaches, we model the $P_\ell$ as conditional distributions $\Pr(X \mid \ell)$ which follow some assumed structure as $\ell$ increases.  Work in this vein has focused on modeling only a few particular quantiles of interest \citep{Bondell2010} or accurate estimation of the conditional distributions using smooth nonparametric regression techniques \citep{Fan1996, Hall1999}.  While such estimators possess nice theoretical properties and good predictive-power, the relationships they describe may be opaque and it is unclear how to quantify the covariate's effect on the entire distribution.  Note that in the case of classic regression, interpretable linear methods remain favored for measuring effects throughout the sciences, despite the availability of flexible nonlinear function families.  Our TRENDS framework retains this interpretability while modeling effects across full distributions. 

Change-point analysis can also be applied to sequences of distributions, but is designed for detecting the precise locations of change-points over long intervals.  However, scRNA-seq experiments only span a brief time-course (typically $L \le 10$), and the primary analytic goal is rather to quantify how much a gene's expression has changed in a biologically interesting manner.  Many change-point methods require explicit parameterization of the types of distributions, an undesirable necessity given the irregular nature of scRNA-seq expression measurements \citep{Kharchenko2014}.  Moreover, some development-related genes exhibit gradual rather than abrupt temporal temporal changes in expression.   Requiring few statistical assumptions, TRENDS models changes ordinally rather than only considering effects that are either smooth or instantaneous, and this method can therefore accurately quantify both abrupt or gradual effects.

\section{Methods}
\label{sec:methods}
Formally, TRENDS fits a regression model to an ordered sequence of distributions, or more broadly, sample pairs $\{(\ell_i, \widehat{P}_i)\}_{i=1}^N$ where each $\ell_i \in \{1,\dots, L \}$ is an ordinal-valued label associated with the $i$th batch, for which we have univariate empirical distribution  $\widehat{P}_i$.  Here, it is supposed that for each batch $i$: a (empirical) quantile function $\widehat{F}^{-1}_i$ is estimated from $n_i$ scalar observations $\{X_{i,s} \}_{s=1}^{n_i} \sim P_i$  sampled from underlying distribution $P_i = \Pr(X \mid \ell_i)$, which may be contaminated by different batch effects for each $i$.  We assume a fixed-design where each level of the covariate $1,\dots, L$ is associated with at least one batch.  In scRNA-seq data, $\widehat{P}_{i}$ is the empirical distribution of one  gene's measured expression values over the cells captured in the same batch and $\ell_i$ indicates the index of the time point at which the batch was sampled from the population for sequencing.  

Unlike the supervised learning framework where one observes samples of $X$ measured at different $\ell$ and the goal is to infer some property of $P_\ell := \Pr(X | \ell)$, in our setting, we can easily obtain $\widehat{P}_{i}$ as an empirical estimate of $\Pr(X | \ell_i)$.   We thus neither seek to estimate the distributions $P_1,\dots, P_L$, nor test for inequality between them.  Rather, the primary goal of TRENDS analysis is to infer how much of the variation in $\Pr(X \mid \ell)$ across different $\ell$ may be attributed to changes in $\ell$ as opposed to the effects of other unmeasured confounding factors.  To quantify this variation, we introduce conditional effect-distributions $Q_\ell$ for which the sequence of transformations $Q_1 \rightarrow Q_2 \rightarrow \dots \rightarrow Q_L$ entirely captures the effects of $\ell$-progression on $\Pr(X \mid \ell)$, under the assumption that these underlying forces follow a \emph{trend} (defined in \S\ref{sec:trenddef}).  We emphasize that the $Q_\ell$ themselves are not our primary inferential interest, rather it is the variation in these conditional-effect distributions that we attribute to increasing-$\ell$ rather than batch effects.

Thus, the $Q_\ell$ are \emph{not} estimators of the sequence of $P_{\ell_i}$.  Rather, the $Q_\ell$ represent the distributions one would expect see in the absence of exogenous effects and random sampling variability, in the case where the underlying distributions \emph{only} change due to $\ell$-progression and we observe the entire population at each $\ell$.  Because we do not believe exogenous effects unrelated to $\ell$-progression are likely to follow a trend over $\ell$, we can identify the sequence of trending distributions which best models the variation in $\{\widehat{P}_{\ell_i}\}_{i=1}^N$ and reasonably conclude that changes in this sequence reflect the $\ell$-progression-related forces affecting $P_\ell$.  

\section{Wasserstein Distance}
\label{sec:wassdist}

TRENDS employs the Wasserstein distance to measure divergence between distributions.  Intuitively interpreted as the minimal amount of ``work'' that must be done to transform one distribution into the other, this metric has been successfully applied in many domains  \citep{Levina2001, Mueller2015}.  The Wasserstein distance is a natural dissimilarity measure of populations because it accounts for the proportion of individuals that are different as well as \emph{how} different these individuals are.  For univariate distributions, the $L_q$ Wasserstein distance is simply the $L_q$ distance between quantile functions given by:  \vspace*{-2mm}
\begin{equation}
d_{L_q}(P,Q) = \left( \int_0^1 \left| F^{-1}(p) - G^{-1}(p)\right|^q \ \mathrm{d}p \right)^{1/q}
\label{wasserstein}  \vspace*{0mm}
\end{equation}
where $F, G$ are the CDFs of $P, Q$ and $F^{-1}, G^{-1}$ are the corresponding \emph{quantile} functions.  Slightly abusing  notation, we use $d_{L_q}(\cdot,\cdot)$ to denote both Wasserstein distances between distributions or the corresponding quantile functions' $L_q$-distance (both $q = 1,2$ are used in this work).  In addition to being easy to compute (in 1-D), the $L_2$ Wasserstein metric is equipped with a natural space of quantile functions, in which the Fr\'echet mean takes the simple form stated in Lemma \ref{frechet}.   Calling this average the \emph{Wasserstein mean}, we note its implicit use in the popular quantile normalization technique \citep{Bolstad2003}. 
\begin{lem} Let $\mathcal{Q}$ denote the space of all quantile functions.  The Wasserstein mean is the Fr\'echet mean in $\mathcal{Q}$ under the $L_2$ norm:  \vspace*{-3mm}
\begin{equation} 
\overline{\mathbf{F}}^{-1} := \frac{1}{N} \sum_{i=1}^N F^{-1}_{i} = \argmin_{G^{-1} \in \mathcal{Q}} \bigg\{  \sum_{i=1}^N \int_0^1 \left( F^{-1}_{i}(p) -  G^{-1}(p) \right)^2 \mathrm{d}p \bigg\} 
 \label{frechetobjective}
\end{equation}  
  \label{frechet}
\end{lem} 
\vspace*{-10mm}

\section{Characterizing trends in distributions}
\label{sec:trenddef}

\begin{definition} \normalfont Let $F_\ell^{-1}(p)$ denote the $p$th quantile of distribution $P_\ell$ with CDF $F_\ell$.  A sequence of distributions $P_1,\dots, P_L$  follows a \textbf{\emph{trend}} if:
\begin{enumerate} \setlength\itemsep{0em}
\item For any $p \in (0,1)$, the sequence $[F_1^{-1}(p), \dots, F_L^{-1}(p)]$ is monotonic.
\item There exists $p^* \in [0,1)$ and two intervals $A, B$ that partition the unit-interval at $p^*$ (one of $A$ or $B$  equals $(0,p^*)$ and the other equals $[p^*, 1)$) such that:  for all $p \in A$, the sequences $[F_1^{-1}(p), \dots, F_L^{-1}(p)]$ are all nonincreasing, and for all $q \in B$, the sequences $[F_1^{-1}(q), \dots, F_L^{-1}(q)]$ are all nondecreasing.  Note that if $p^* = 0$, then all quantiles must change in the same direction as $\ell$ grows.  
\end{enumerate}
 \label{trenddef}
\end{definition}  

\vspace*{-0mm}

\begin{figure}[h!] 
\centering
\includegraphics[width=0.243\textwidth]{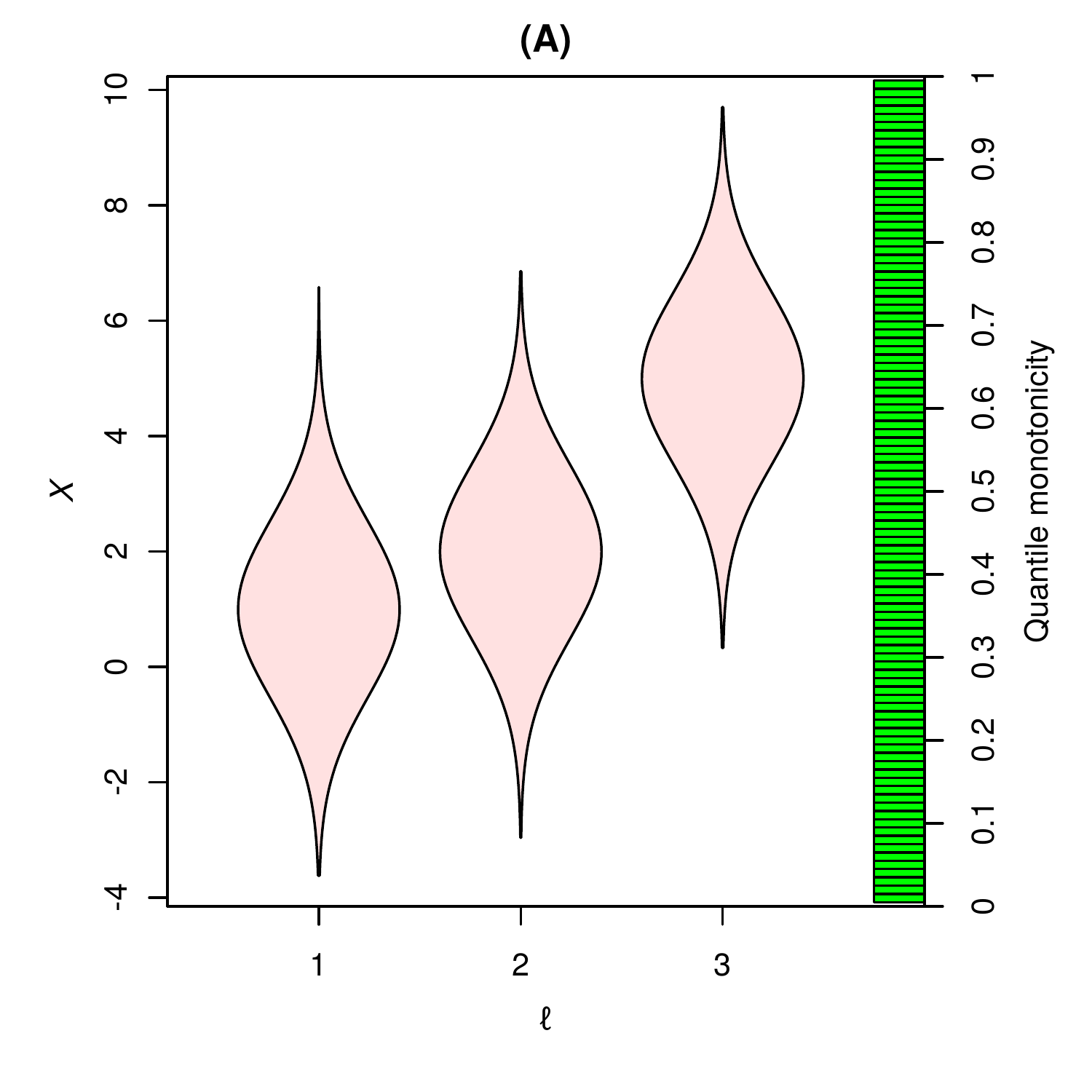}
\includegraphics[width=0.243\textwidth]{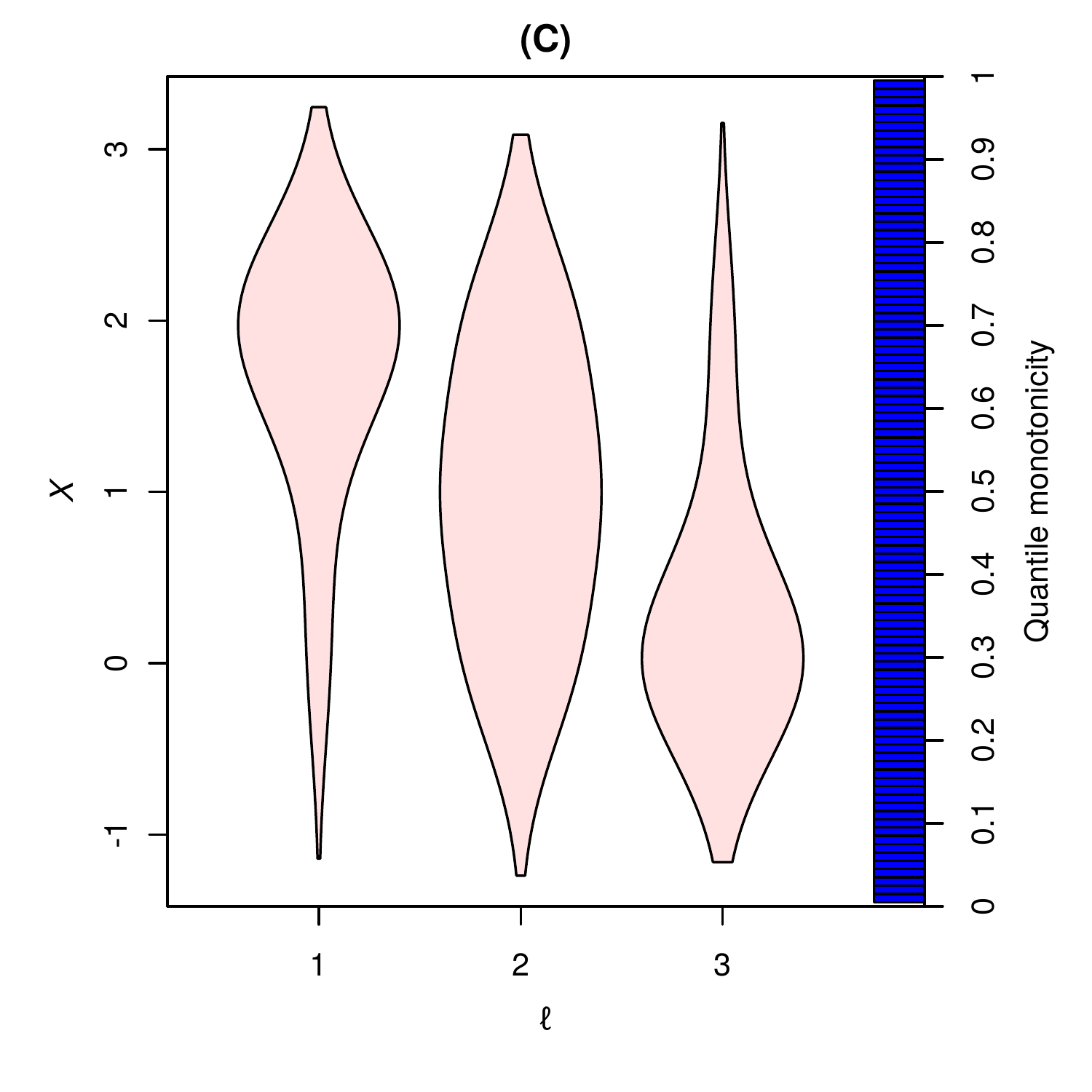}
\includegraphics[width= 0.243\textwidth]{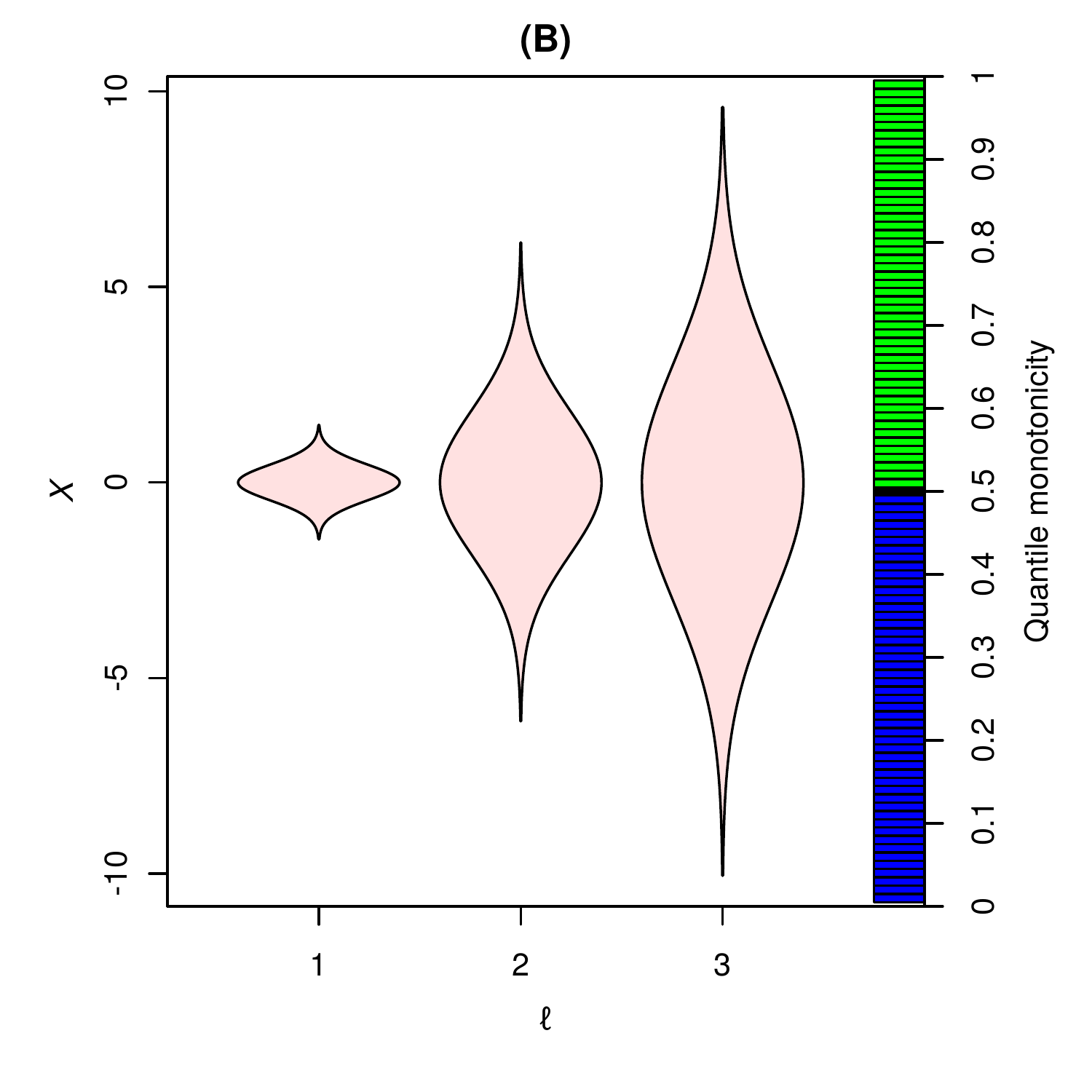}
\includegraphics[width=0.243\textwidth]{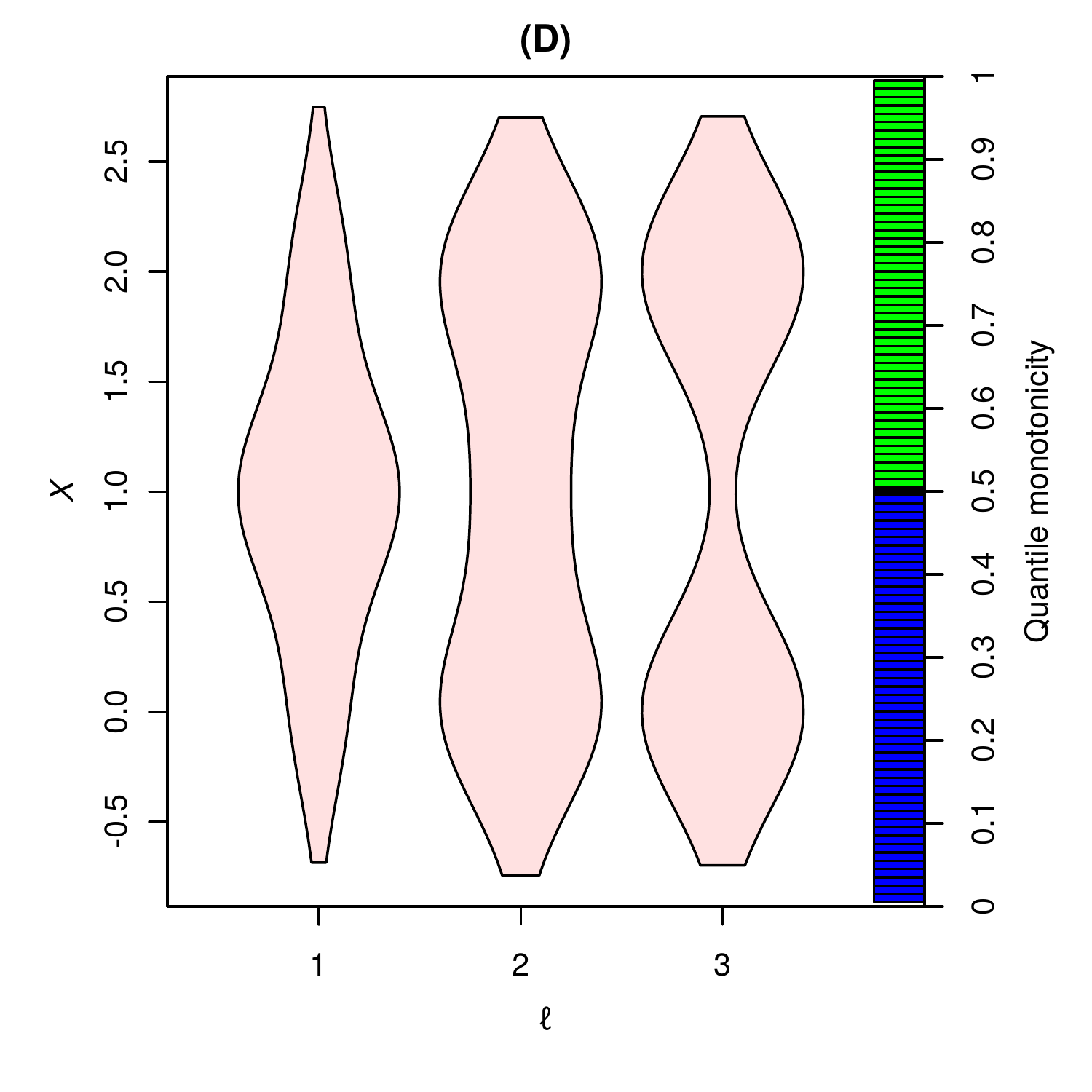} \\
\vspace*{-2.1mm}
\includegraphics[width=0.33 \textwidth]{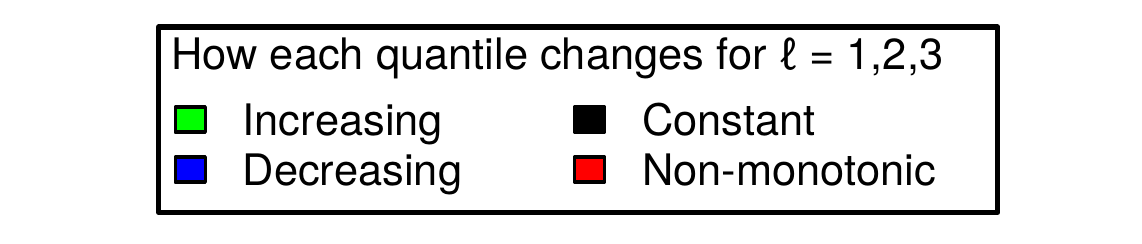}
\caption{Violin plots depicting four different sequences of distributions which follow a trend.  The $p$th rectangle in the color bar on the righthand side indicates the monotonicity of the $p$th quantile over the sequence of distributions (for $p = 0.01, 0.02, \dots, 0.99$).}
\label{trendgoodex}
\end{figure}

Our formal definition of a trend applies to distributions which evolve in a consistent fashion, ensuring that the temporal-forces that drive the transformation from $P_{1}$ to $P_{L}$ do so without reversing their effects or leading to wildly different distributions at intermediate $\ell$ values.  While the second condition of our definition technically subsumes the first, Condition 1 contains our key idea and is therefore separated from Condition 2, a subtler additional assumption that does not significantly alter results in practice.  Note that the trend definition employed in this paper is intended for relatively short sequences and does not include cyclic/seasonal patterns studied in time-series modeling.

\begin{lem} If distributions $P_1, \dots, P_L$ follow a trend, then  \vspace*{-3mm}
\begin{equation*} d_{L_1}(P_i, P_j) = \sum_{\ell = i+1}^j d_{L_1}(P_{\ell-1}, P_\ell)  \ \ \ \text{ for all } \ i < j \in \{1,\dots,L\}
\end{equation*}
\label{trendl1}
\end{lem} \vspace*{-10mm}
Measuring how much the distributions are perturbed between each pair of levels via the $L_1$ Wasserstein metric, Lemma \ref{trendl1} shows the trend criterion as an instance of Occam's razor, where the underlying effects of interest are assumed to transform the distribution sequence in the simplest possible manner (recall that the Wasserstein distance is interpreted as the minimal work required for a given transformation).  If one views the underlying effects of interest as a literal force acting in the space of distributions, Lemma \ref{trendl1} implies that this force points the same direction for every $\ell$ (i.e.\ $P_1,\dots,P_L$ lie along a line in the  $L_1$ Wasserstein metric space of distributions).  A trend is more flexible than a linear restriction in the standard sense, because the magnitude of the force (how far along the line the distributions move) can vary over $\ell$.  Thus, we have formally extended the colloquial definition of a trend (``a general direction in which something is developing or changing'') to probability distributions.

To further conceptualize the trend idea, one can view quantiles as different segments of a population whose values are distributed according to $\Pr(X \mid \ell)$ (e.g.\ for wealth-distributions, it has become popular to highlight the ``one percent'').  From this perspective, it is reasonable to assume that while the force of sequential progression may have different effects on the groups of individuals corresponding to different segments of the population, its effects on a single segment should be consistent over the sequence.   If some segment's values initially change in one way at lower levels of $\ell$ and subsequently revert in the opposite direction over larger $\ell$ (i.e.\ this quantile is non-monotone), it is natural to conclude there are actually multiple different progression-related forces affecting this homogeneous group of individuals.  It is therefore natural to assume a trend if we only wish to measure the effects of a single primary underlying force.  Often in settings such as scRNA-seq developmental experiments, the researcher has a priori interest in a specific effect (such as how each gene contributes to a specific stage of the developmental process).  Therefore, data are collected over a short $\ell$-range such that the primary effects of interest should follow a trend.

The second condition in the trend definition specifies that adjacent quantiles must move in the same direction over $\ell$ except at most a single $p^*$.  This restricts the number of population-segments which can increase over $\ell$ when a nearby segment of the population is decreasing.  Intuitively, Condition 2 forces us to borrow strength across adjacent quantiles when estimating effects that follow a trend.  The main effect of the additional restriction imposed by this condition prevents a trend from completely capturing extremely-segmented effects (such as the example depicted in Figure \ref{trendbadex}C).   However, applications involving such complex phenomena are uncommon (it is difficult to imagine a setting where the primary effects-of-interest push more than two adjacent segments of a population in different directions), and such nuanced changes can be reasonably attributed to spurious nuisance variation.  We note that a trend can still  roughly approximate the major overall effects even when the actual distribution-evolution violates Condition 2 (as seen in Figure \ref{trendbadex}C).  In practice, the results of our method are not significantly affected by this second restriction, but it provides nice theoretical properties ensuring our estimation procedure (presented in \S\ref{sec:fitting}) efficiently finds a globally optimal solution, as well as additional robustness against spurious quantile-variation in the data (possibly due to estimation-error given limited samples per batch).  

Figure \ref{trendgoodex} depicts simple examples of trending distribution-sequences.  In each example, it is visually intuitive that the evolution of the distributions proceeds in a single consistent fashion.  To highlight the broad spectrum of interesting effects TRENDS can detect, we present three conceptual examples in \S\ref{sec:trendexs} of distribution-sequences that follow a trend, which includes consistent changes in location/scale and the growth/disappearance of modes.  Despite imposing conditions on every quantile, the trend criterion does not require: explicit parameterization of the distributions, specification of a precise functional form of the $\ell$-effects, or reliance on a smooth or constant amount of change between different levels.  This generality is desirable for modeling developmental gene expression and other enigmatic phenomena where stronger assumptions may be untenable.

\begin{figure}[h!] 
\includegraphics[width= 0.49\textwidth]{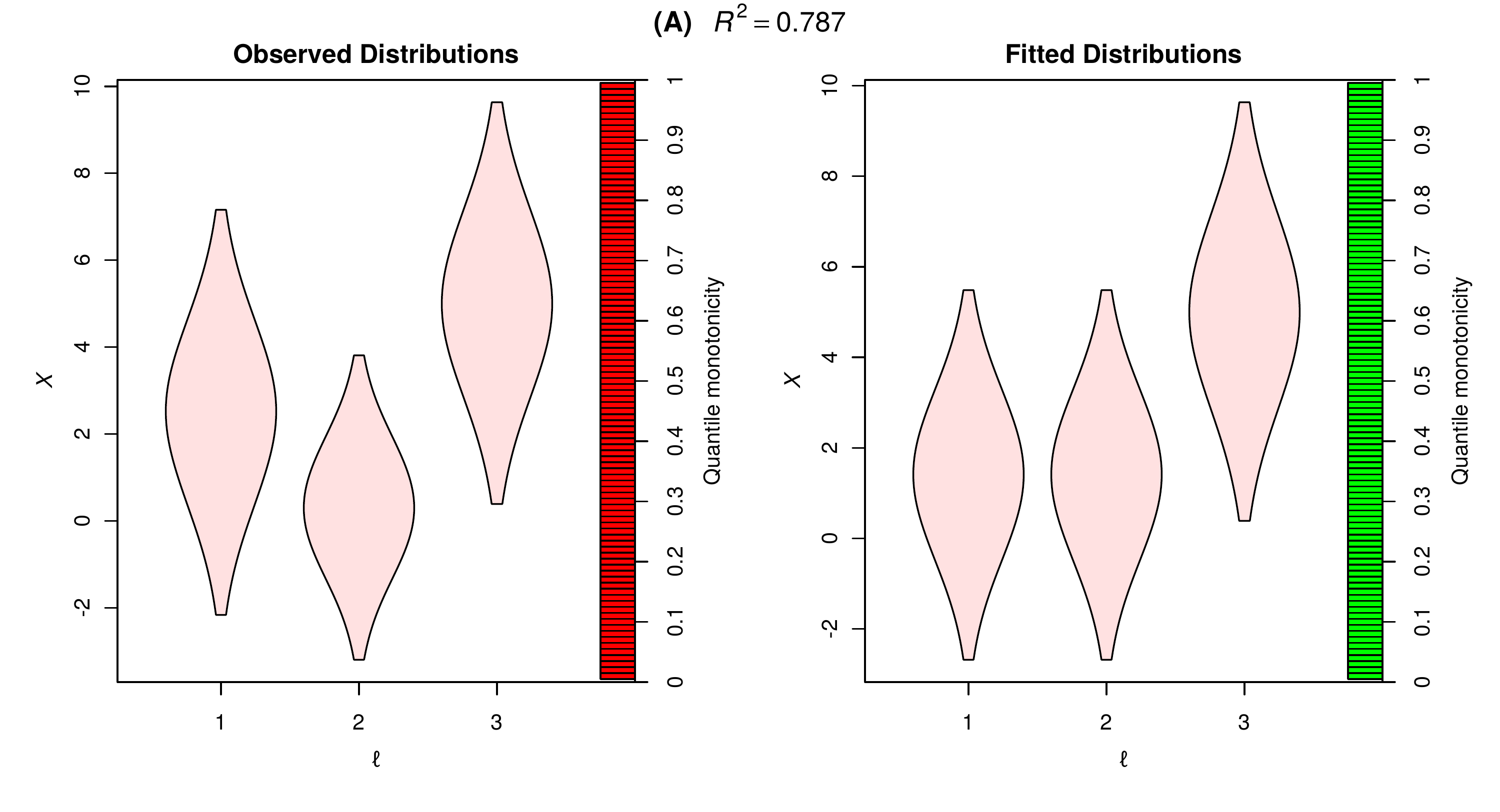} 
\includegraphics[width= 0.49\textwidth]{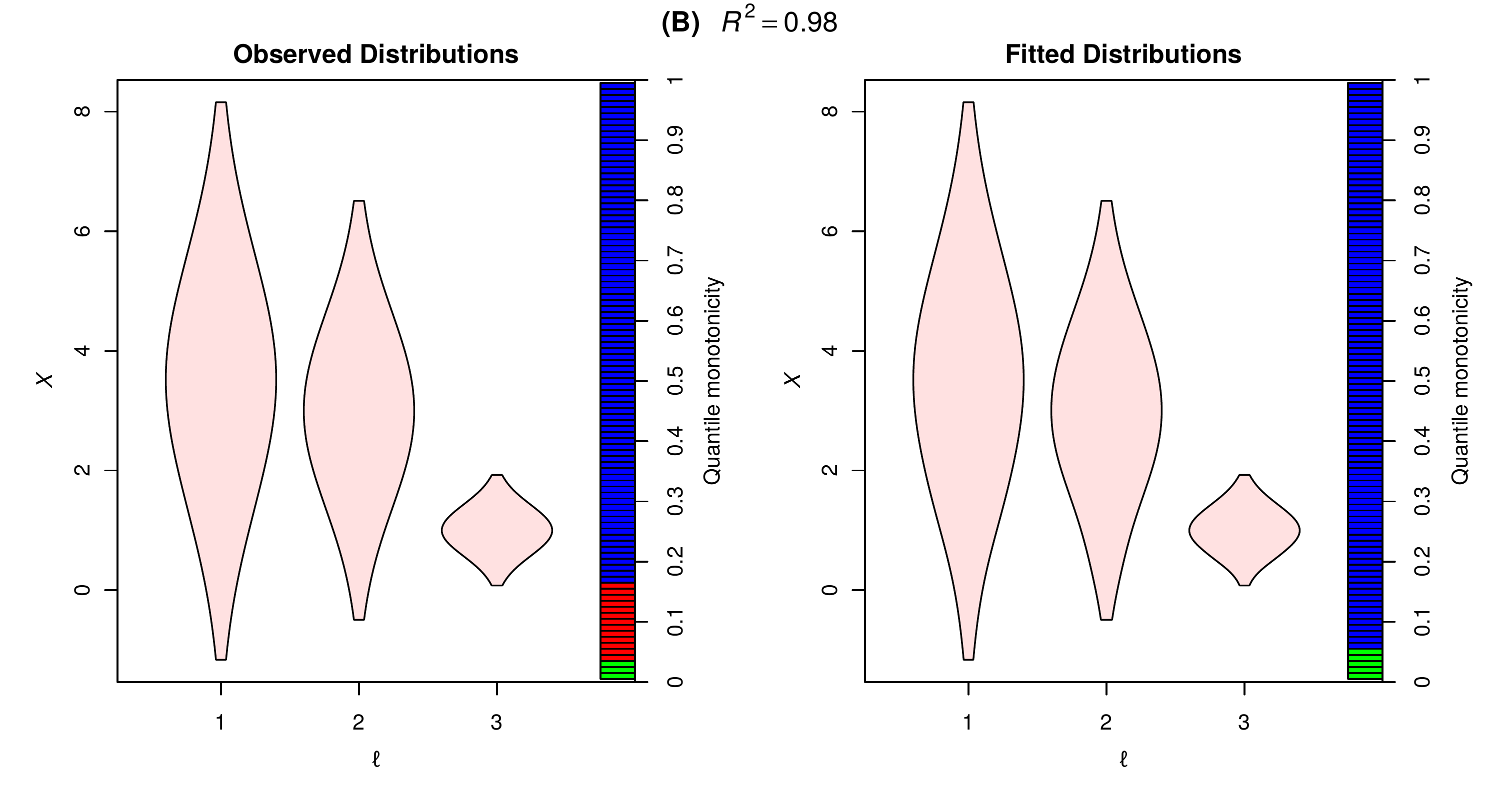}

\includegraphics[width=0.49 \textwidth]{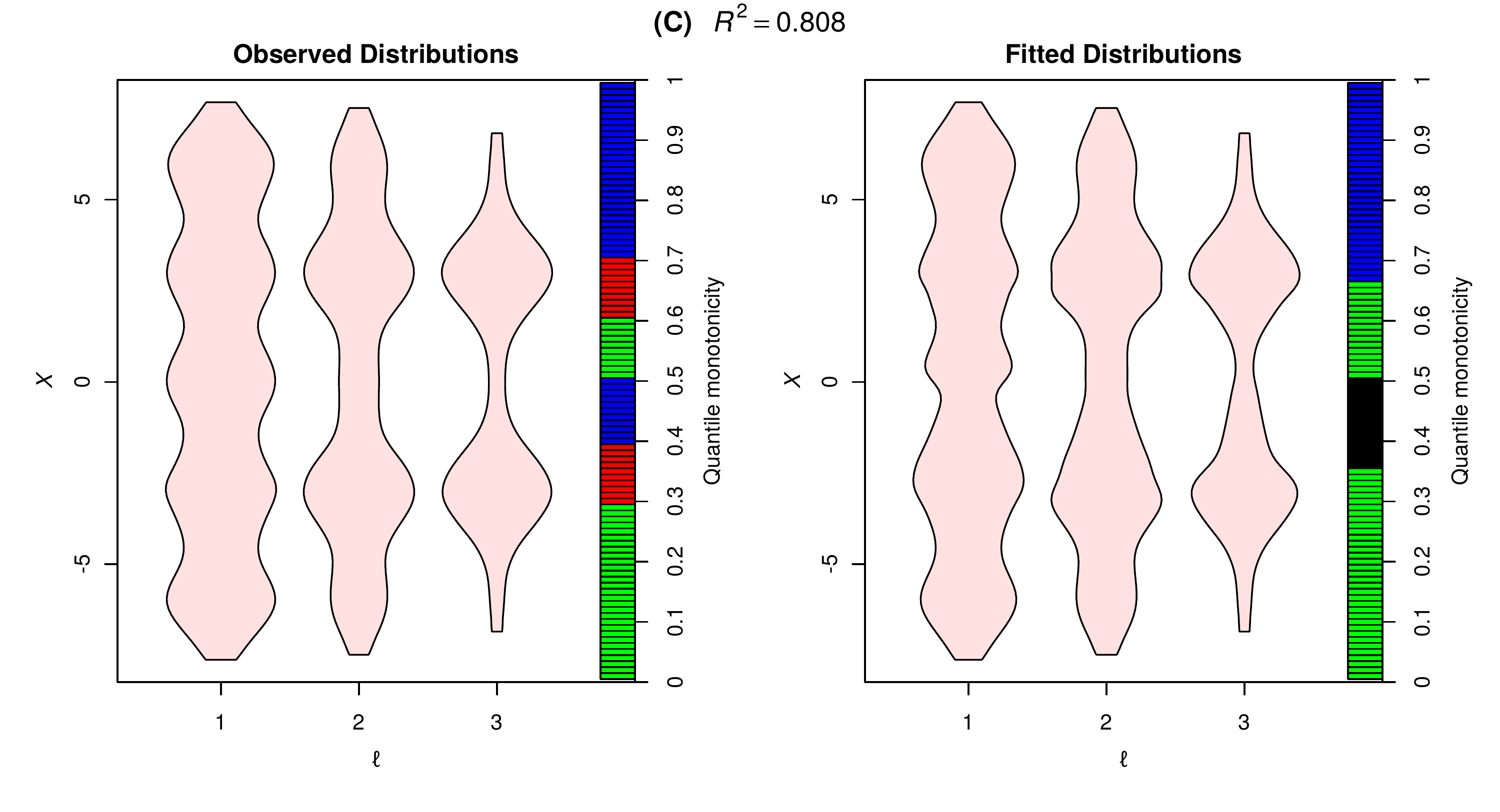}   \hspace*{3mm}
\raisebox{1.2\height}{\includegraphics[width=0.4 \textwidth]{goodexamples_colorbarLegend.pdf}}
\caption{Violin plots depicting sequences of distributions which do \emph{not} follow a trend (Observed Distributions in lefthand panels).  Shown to the right of each example are the corresponding fitted distributions estimated by TRENDS (with the TRENDS $R^2$ value).}
\label{trendbadex} \thisfloatpagestyle{empty}
\end{figure}

 The lefthand panels of Figure \ref{trendbadex} depict three examples of sequences which do not follow a trend for different reasons.  To the right of each example, we show the ``best-fitting'' sequence that does follow a trend (formally defined in (\ref{trends})), each distribution of which corresponds to our estimate of $Q_\ell$ (introduced in \S\ref{sec:methods}).  We reiterate that the $Q_\ell$ are not by themselves of interest, but are merely used to quantify the sequential-progression effects (as will be described in \S\ref{sec:r2}).   Nonetheless, the visual depiction of the trending $Q_\ell$ provides insight regarding what sort of changes a trend can accurately approximate.  Whereas the evolution of the (trending) fitted distributions in Figure \ref{trendbadex}A (on right) can intuitively be attributed to one consistent force, multiple are required to explain the variation in the original non-trending sequence of distributions on the left.  Identifying a single consistent effect responsible for the changes in the left panel of Figure \ref{trendbadex}B is far more plausible, and we note that these distributions in fact are much closer to following a trend (while hard to visually discern, the $0.04^{\text{th}} - 0.16^{\text{th}}$ quantiles of the observed distribution sequence increase between $\ell = 1$ to 2 and decrease slightly from $\ell = 2$ to 3, thus violating a trend).

During specific stages of development, changes in the observed cellular gene-expression distributions generally stem from the emergence/disappearance of different cell subtypes (plus batch and random sampling effects).  Clear subtype distinctions may not exist in early stages where cells remain undifferentiated, and thus not only are the relative proportions of different subtypes changing, but the subtypes themselves may transform as well.  Therefore, developmental genes' underlying expression patterns are likely described by  Examples 2 and 3 (of specific conceptual types of trends) in \S\ref{sec:trendexs}.  The trend criterion fits our a priori knowledge well, while remaining flexible with respect to the precise nature of expression changes.

\section{TRENDS regression model}
\label{sec:trendsregression}

Recall that in our setting, even the underlying batch distributions $P_i$ (from which the observations $X_{i,s}$ are sampled) may be contaminated by latent confounding effects.  We assume the quantile functions of each $P_i$ are generated from the  model below: \vspace*{-4mm}
\begin{align*}
& F^{-1}_i = G^{-1}_{\ell_i} + \mathcal{E}_i \ \ \text{such that } G^{-1}_1,\dots, G^{-1}_L \text{ follow a trend, and the following hold: } \ \ \ \numberthis \label{trendsmodel}  \\
\end{align*} \vspace*{-26mm}
\begin{itemize}  \setlength\itemsep{0em}
\refstepcounter{assumptioncount} \label{assumptionconstrainedei}
\item [(A.\arabic{assumptioncount})] $\mathcal{E}_i : (0,1) \rightarrow \mathbb{R}$ is constrained so that $G^{-1}_{\ell_i}$ and $F_i^{-1}$ are valid quantile functions.
\refstepcounter{assumptioncount} \label{assumptionsubgaussianei}
\item [(A.\arabic{assumptioncount})] For all $p \in (0,1)$ and $i$: $\mathcal{E}_i(p)$ follows a sub-Gaussian($\sigma$) distribution \citep{Honorio2014}, so $ \mathbb{E}[\mathcal{E}_i(p)] = 0$ and $\Pr(|\mathcal{E}_i(p) | > t) \le 2 \exp \left(-\frac{t^2}{2\sigma^2} \right)$ for any $t > 0$.  \refstepcounter{assumptioncount} \label{assumptionindependentei} 
\item [(A.\arabic{assumptioncount})] For all $p \in (0,1)$ and $i \neq j$:  $\mathcal{E}_i(p)$ is statistically independent of $\mathcal{E}_j(p)$.
\end{itemize}

In this model, $G^{-1}_\ell$ is the quantile function of the conditional effect-distribution $Q_\ell$, whose evolution captures the underlying effects of level-progression.  The random noise functions $\mathcal{E}_i : (0,1) \rightarrow \mathbb{R}$ can represent measurement-noise or the effects of other unobserved variables which contaminate a batch.  Note that the form of $\mathcal{E}_i$ is implicitly constrained to ensure all $F^{-1}_i, G^{-1}_{\ell_i}$ are valid quantile functions.  Because $\mathcal{E}_i(p_1)$ and $\mathcal{E}_i(p_2)$ are allowed to be dependent for $p_1 \neq p_2$, the effect of one $\mathcal{E}_i$ may manifest itself in multiple observations $X_{i,s}$, even if these observations are drawn i.i.d.\ from $P_i$ (for example, a batch effect can cause all of the observed values from a batch to be under-measured).   In fact, condition (A.\ref{assumptionconstrainedei}) encourages significant dependence between the noise at different quantiles for the same batch.  The assumption of sub-Gaussian noise is quite general, encompassing cases in which the $\mathcal{E}_i(p)$ are either: Gaussian, bounded, of strictly log-concave density, or any finite mixture of sub-Gaussian variables \citep{Honorio2014}.  Although condition (A.\ref{assumptionindependentei}) stringently ensures all dependence between observations from different $\ell$ arises due to the trend, similar independence assumptions are required in general regression settings where one cannot reasonably a priori specify a functional form of dependence in the noise.  Real batch effects are likely to satisfy (A.\ref{assumptionindependentei}) since they typically have the same chance of affecting any given batch in a certain manner (because the same experimental procedure is repeated across batches, as in the case of the cell-capture and library preparation in scRNA-seq).  Nonetheless, we note that assumption (A.\ref{assumptionsubgaussianei}) can be immediately generalized (with trivial changes to our proofs) in order to allow heteroscedasticity in the batch effects $\mathcal{E}_i$ (endowing each batch with a different $\sigma_i$ sub-Gaussian parameter), but we opt for simplicity in this theoretical exposition.

Model (\ref{trendsmodel}) is a distribution-valued analog of the usual regression model, which assumes  scalars $Y_i = f(X_i) + \epsilon_i$ where $\epsilon_i \sim$ sub-Gaussian($\sigma^2$) and $\epsilon_i$ is independent of $\epsilon_j$ for $i \neq j$.  In (\ref{trendsmodel}), an analogous $f$ maps each ordinal level \{1,\dots, L\} to a quantile function, $f(\ell_i) = G^{-1}_{\ell_i}$,  and the class of functions is restricted to those which follow a trend.  Our assumption of mean-zero $\mathcal{E}_i$ that are independent between batches is a straightforward extension of the scalar error-model to the batch-setting, and ensures that the exogenous noise is unrelated to $\ell$-progression under (\ref{trendsmodel}).  Just as the  $Y_1, \dots, Y_N$ are rarely expected to exactly lie on the curve $f(x)$ in the classic scalar-response model, we do not presume that the observed distributions $\widehat{P}_i$ will exactly follow a trend (even as $n_i \rightarrow \infty \ \forall i$ so that $\widehat{P}_i \rightarrow P_i$).  Rather our model simply encodes the assumption that the effects of level-progression on the distributions should be consistent over different $\ell$ (i.e.\ the effects follow a trend).  

For each $\ell$, TRENDS finds a fitted distribution $\widehat{Q}_\ell$ using the \emph{Wasserstein least-squares} fit which minimizes the following objective:
\begin{equation}
\widehat{Q}_1,\dots,\widehat{Q}_L = \argmin_{Q_1,\dots,Q_L} \bigg\{ \sum_{\ell = 1}^L \ \sum_{i \in I_\ell} d_{L_2}(Q_\ell, \widehat{P}_i)^2 \bigg\} \ \text{ where } Q_1, \dots, Q_L \text{ follow a trend}
\label{trends}
\end{equation}
where $I_\ell$ is the set of batch-indices $i$ such that $\ell_i = \ell$, and we require $N_\ell := | I_\ell | \ge 1$  for all $\ell \in \{1,\dots,L\}$.  Subsequently, one can inspect changes in the $\widehat{Q}_\ell$ which should reflect the  transformations in the underlying $P_\ell$ that are likely caused by increasing $\ell$.  Figure \ref{trendbadex} shows some examples of fitted distributions produced by TRENDS regression.  The objective in (\ref{trends}) bears great similarity to the usual least-squares loss used in scalar regression, the only differences being:  scalars have been replaced by distributions, squared Euclidean distances are now squared Wasserstein distances, and the class of regression functions is defined by a trend rather than linearity/smoothness criteria.  

Expression measurements in scRNA-seq are distorted by significant batch effects, so the $\mathcal{E}_i$ are likely to be large.   In addition to  technical artifacts, \citet{Buettner2015} find biological sources of noise due to processes such as transcriptional bursting and cell-cycle modulation of expression.  Unlike development-driven changes in the underlying expression of a developmental gene, other biological/technical sources of variation are unlikely to follow any sort of trend.
TRENDS thus provides a tool for modeling full distributions, while remaining robust to the undesirable variation rampant in these applications by leveraging independence of the noise between different batches of simultaneously captured and sequenced cells.

\section{Measuring fit, effect size, and statistical significance}
\label{sec:r2}
Analogous to the coefficient of determination used in classic regression, we define the Wasserstein $R^2$ to measure how much of the variation in the observed distributions $\widehat{P}_1, \dots, \widehat{P}_N$ is captured by the TRENDS model's fitted distributions $\widehat{Q}_{1}, \dots, \widehat{Q}_{L}$:
\begin{equation}
R^2 := 1-  \left( \frac{1}{N}  \sum_{i=1}^N d_{L_2}(\widehat{Q}_{\ell_i}, \widehat{P}_{i})^2 \right) {\Bigg{/}} \left( \frac{1}{N}  \sum_{i=1}^N d_{L_2}( \widehat{P}_{i}, \overline{\mathbf{F}}^{-1} )^2  \right) \ \ \in [0, 1]
\label{rsquared}
\end{equation} 
Here, squared distances between scalars in the classic $R^2$ are replaced by squared Wasserstein distances between distributions, and the quantile function $ \overline{\mathbf{F}}^{-1} = \frac{1}{N} \sum_{i=1}^N \widehat{F}^{-1}_{i}$ is the \emph{Wasserstein mean} of all observed distributions.  By Lemma \ref{frechet}, the numerator and denominator in (\ref{rsquared}) are respectively analogous to the residuals and the overall variance from usual scalar regression models.  
 
In classic linear regression, the regression line slope is interpreted as the expected change in the response resulting from a one-unit increase in the covariate.  While TRENDS operates on unit-less covariates, we can instead measure the overall \emph{expected Wasserstein-change} under model (\ref{trendsmodel}) in the $\widehat{P}_i$ over the full ordinal progression $\ell = 1,\dots, L$ using:
\begin{equation}
\Delta := \frac{1}{L} \cdot d_{L_1}(\widehat{Q}_1, \widehat{Q}_L) 
\label{delta}
\end{equation}
The $L_1$ Wasserstein distance is a natural choice, since by Lemma \ref{trendl1}, it measures the aggregate difference over each pair of adjacent $\ell$ levels (just as the difference between the largest and smallest fitted-values in linear regression may be decomposed in terms of covariate units to obtain the regression-line slope).  Thus, $\Delta$ measures the raw magnitude of the inferred trend-effect (depends on the scale of $X$), while $R^2$ quantifies how well the trend-effect explains the variation in the observed distributions (independently of scaling).  Note that if the TRENDS model is fit to the distributions from the example in Figure \ref{trendbadex}B, the TRENDS-inferred effect of sequential-progression is nearly as large as the overall variation in this sequence, which agrees with our visual intuition that the observed distributions  already evolve in a fairly consistent fashion.

Finally, we introduce a test to assess statistical significance of the trend-effect.  We compare the null hypothesis $H_0 : \ Q_1 = Q_2 = \dots = Q_L$  against the alternative that the $Q_i$ are not all equal and follow a trend.  To obtain a $p$-value, we employ permutation testing on the $\ell_i$-labels of our observed distributions $\widehat{P}_i$ with test-statistic $R^2$ \citep{Good1994}.  More specifically, the null distribution is determined by repeatedly executing the following steps: (i) randomly shuffle the $\ell_i$ so that each $\widehat{P}_i$ is paired with a random $\ell_i^{\text{perm}} \in \{1,\dots, L\}$  value, (ii) fit the TRENDS model to the pairs $\{ (\ell_i^{\text{perm}}, \widehat{P}_i) \}_{i=1}^N$ to produce $\widehat{Q}_1^{\text{perm}} , \dots , \widehat{Q}_L^{\text{perm}}$, (iii) use these estimated distributions to compute $R^2_{\text{perm}}$ using (\ref{rsquared}).  Due to the quantile-noise functions $\mathcal{E}_i(\cdot)$ assumed in our model (\ref{trendsmodel}), $H_0$ allows variation in our sampling distributions $P_i$ which stems from  non-$\ell$-trending forces.  Thus the TRENDS test attempts to distinguish whether the effects transforming the $P_i$ follow a trend or not, but does not presume the $P_i$ will look identical under the null hypothesis.   By measuring how much further the $\widehat{P}_i$ lie from one distribution vs.\ a sequence of trending distributions in Wasserstein-space, we note that our $R^2$ resembles a likelihood-ratio-like test statistic between maximum-likelihood-like estimates $\overline{\mathbf{F}}^{-1}$ and $\widehat{Q}_\ell$ (where we operate under the Wasserstein distance rather than Kullback-Leibler which underlies the maximum likelihood framework).  

As we do not parametrically treat the distributions, we find permutation testing more suitable than relying on  asymptotic approximations.  Unfortunately, $N$ and $L$ may be  small, undesirably limiting the number of possible label-permutations.   In \S\ref{sec:testing}, we overcome the granularity problem that arises in such settings by developing a more intricate permutation procedure akin to the smoothed bootstrap of \citet{Silverman1987}.  

To determine whether our model is reasonable when working with real data, it is best to rely on prior domain knowledge regarding whether or not the effects of primary interest should follow a trend.  When this fact remains uncertain, then (as in the case of classical regression) the question is not properly  answered using just our Wasserstein $R^2$ values (which we caution tend to be much larger than the familiar $R^2$ values from linear regression, due to the heightened flexibility of our TRENDS model).  \S\ref{sec:checking} demonstrates a simple method for model checking based on plotting empirically-estimated residual functions $\widehat{\mathcal{E}}_i$ against the sequence-level $\ell$.  Similar plots of scalar residuals are the most common diagnostic employed in standard regression analysis.  While this model-checking procedure is able to clearly delineate simulated deviations from our assumptions, it shows little indication that the TRENDS assumptions are inappropriate for the real scRNA-seq data from major known developmentally-relevant genes.  Our simulation in \S\ref{sec:checking} also empirically demonstrates that despite its restrictive assumptions, the TRENDS model can provide superior estimates of severely-misspecified effects than the initial empirical distributions.

\section{Fitting the TRENDS model}
\label{sec:fitting}

We propose the trend-fitting (TF) algorithm which finds distributions satisfying 
\begin{equation}
\widehat{Q}_1,\dots,\widehat{Q}_L = \arg \min_{Q_1,\dots,Q_L} \bigg\{ \sum_{\ell = 1}^L \ \sum_{i \in I_\ell} w_i \cdot d_{L_2}(Q_\ell, \widehat{P}_i)^2 \bigg\} \ \text{ where } Q_1, \dots, Q_L \text{ follow a trend}
\label{weightedtrends}
\end{equation}
If $\widehat{P}_i$ (the empirical per-batch distributions) are estimated from widely varying sample sizes $n_i$ for different batches $i$,  then it is preferable to replace the objective in (\ref{trends}) with the weighted sum in (\ref{weightedtrends}).  Given weights $w_i$ chosen based on  $n_i$ and $N_\ell$, TRENDS can better model the variation in the empirical distributions that are likely more accurate due to larger sample size.  As $n_i$ and $N_\ell$ are fairly homogeneous in scRNA-seq experiments, we use uniform weights here (but provide an algorithm for the general formulation).  To fit TRENDS to data $\{(\ell_i, \widehat{P}_{\ell_i}, w_i)\}_{i=1}^N$ via our procedure, the user must first specify: 
\vspace*{-1.5mm}
\begin{itemize} \setlength\itemsep{0em}
\item Numerical quadrature points $0 < p_1 < p_2 < \dots < p_{P-1} < 1$ for  evaluating  the Wasserstein distance integral in (\ref{wasserstein}), i.e.\ which $P-1$ quantiles to use for each batch 
\item a quantile estimator $\widehat{F}^{-1}(p)$ for empirical CDF $\widehat{F}$
\end{itemize}
\vspace*{-1.5mm}
Given these two specifications, the TF procedure solves a numerical-approximation of the constrained distribution-valued optimization problem in (\ref{weightedtrends}).
Defining $p_0 := 2p_1 - p_2$ and $p_P := 2 p_{P-1} - p_{P-2}$, we employ the following midpoint-approximation of the integral
\begin{align}
\min_{G_1^{-1},\dots,G_L^{-1}} \bigg\{ \sum_{\ell = 1}^L  \sum_{i \in I_\ell} w_i  \sum_{k=1}^{P-1}  \left( \widehat{F}^{-1}_i (p_k) - G^{-1}_\ell(p_k) \right)^2 \left[ \frac{p_{k+1} - p_{k-1}}{2} \right] \bigg\} \nonumber \\
\text{ where } G_1, \dots, G_L \ \text{must follow a trend} \numberthis \label{numericaltrends}
\end{align}

While this problem is unspecified between the $p_{k}$th and $p_{k+1}$th quantiles, all we require to numerically compute Wasserstein distances (and hence $R^2$ or $\Delta$) is the values of the quantile functions at $p_1, \dots, p_{P-1}$, which are uniquely determined by (\ref{numericaltrends}).  Although our algorithm operates on a discrete set of quantiles like techniques for quantile regression   \citep{Bondell2010}, this is only for practical numerical reasons; the goal of our TRENDS framework is to measure effects across an entire distribution.  Throughout this work, we use $P - 1$ uniformly spaced quantiles between $\frac{1}{P}$ and $\frac{P-1}{P}$ (with $P = 100$) to comprehensively capture the full distributions while ensuring computational efficiency.  In settings with limited  data per batch, one might alternatively select fewer quadrature points (quantiles), avoiding tail regions of the distributions for increased stability (our results were robust to the precise number of quadrature points employed).

Since no unbiased minimum-variance $ \forall p \in (0,1)$ quantile estimator is known, we simply use the default setting in  \emph{R}'s \textsf{quantile} function, which provides the best approximation of the mode (Type 7 of \citet{Hyndman1996}).  Other quantile estimators perform similarly in our experiments, and \citet{keen2010graphics} have found little practical difference between estimation procedures for sample sizes $\ge 30$.  Here, we assume the $n_i$ cells sampled in the $i$th batch are i.i.d.\ samples (reasonable for cell-capture techniques).  If this assumption is untenable in another domain, then the quantile-estimation should be accordingly adjusted (cf.\ \citeauthor{Heidelberger1984} \citeyear{Heidelberger1984}). 

\begin{figure}[h!]
{\setstretch{1.0} \footnotesize
\noindent\rule[0.5ex]{\linewidth}{1pt} 
\textbf{Basic PAVA Algorithm: }  $ \min_{z_\ell} \ \sum_{\ell =1}^L (y_\ell - z_\ell)^2$ \ \ s.t.  $z_1 \le \dots \le z_L$ \\
\noindent\rule[0.5ex]{\linewidth}{1pt}
\textbf{  Input:} \hspace*{5.2mm} A sequence of real numbers $y_1, \dots, y_L$ \\
\textbf{Output:}   \hspace*{2mm} The minimizing sequence $\widehat{y}_1,\dots, \widehat{y}_L$ which is nondecreasing.
\begin{enumerate} \setlength\itemsep{0em}
\item Start with the first level $\ell=1$ and set the fitted value $\widehat{y}_1 = y_1$  
\item While the next $y_\ell \ge \widehat{y}_{\ell-1}$, set $\widehat{y}_\ell = y_\ell$ and increment $\ell$
\item If the next $\ell$ violates the nondecreasing condition, i.e.\ $y_\ell < \widehat{y}_{\ell-1}$, then \emph{backaverage} to restore monotonicity: find the smallest integer $k$ such that replacing $\hat{y}_\ell, \dots, \hat{y}_{\ell-k}$ by their average restores the monotonicity of the sequence $\hat{y}_1,\dots, \hat{y}_\ell$.  Repeat Steps 2 and 3 until $\ell = L$. \end{enumerate}
\vspace*{-3mm}
\noindent\rule[0.5ex]{\linewidth}{1pt} 
\vspace*{-3mm}} \end{figure} 

Our procedure uses the Pool-Adjacent-Violators-Algorithm (PAVA), which given an input sequence $y_1,\dots, y_L \in \mathbb{R}$, finds the least-squares-fitting nondecreasing sequence in only $O(L)$ runtime \citep{DeLeeuw1977}.  The basic PAVA procedure is extended to weighted observations by performing weighted backaveraging in Step 3. When multiple $(\ell_i, y_i)$ pairs are observed with identical covariate-levels, i.e.\ $\exists \ell$ s.t.\ $N_\ell := \left| I_\ell \right| > 1$ where  $ I_\ell := \{i : \ell_i = \ell \}$, we adopt the simple \emph{tertiary} approach for handling predictor-ties \citep{DeLeeuw1977}.  Here, one defines $ \displaystyle \bar{y}_\ell $ as the (weighted) average of the $\{y_i : i \in I_\ell \}$ and for each level $\ell$ all $y_i :  i \in I_\ell$ are simply replaced with their mean-value $\bar{y}_\ell$.  Subsequently, PAVA is applied with non-uniform weights to $\{(\ell, \bar{y}_\ell)\}_{\ell = 1}^L$ where the $\ell$th point receives weight $N_\ell$ (or weight $\sum_{i \in I_\ell} w_i$ if the original points are assigned non-uniform weights $w_1,\dots, w_N$).  By substituting ``nonincreasing''  in place of  ``nondecreasing''  in Steps 2 and 3, the basic PAVA method can be trivially modified to find the least-squares \emph{nonincreasing} sequence.  From here on, we use PAVA$((y_1,w_1), \dots, (y_N, w_N) ; \delta)$ to refer to a more general version of basic PAVA, which incorporates observation-weights $w_i$ (for multiple $y$ values at a single $\ell$), and a user-specified monotonicity condition $\delta \in \{\text{``nonincreasing'', ``nondecreasing''}\}$ that determines which monotonic best-fitting sequence to find.  

\begin{figure}[h!]
{\setstretch{1.0} \footnotesize
\noindent\rule[0.5ex]{\linewidth}{1pt}
\textbf{Trend-Fitting Algorithm: }  Numerically solves (\ref{weightedtrends}) by optimizing (\ref{numericaltrends})  \\
\noindent\rule[0.5ex]{\linewidth}{1pt}
\textbf{Input 1:} \hspace*{2.7mm}  Empirical distributions and associated levels (and optional weights) $\{(\ell_i, \widehat{F}_i , w_i) \}_{i = 1}^N$ \\
\textbf{Input 2:}   \hspace*{2mm}  A grid of quantiles to work with $0< p_1 < \dots < p_{P-1}< 1$ 
\\
\textbf{Output:}  \hspace*{2.5mm}  The estimated quantiles of each $Q_\ell$ \  $\{ \widehat{G}^{-1}_\ell (p_k) : k = 1,\dots, P-1 \} $ for  $\ell \in \{1,\dots, L\}$  \\
 \hspace*{19.8mm}  from which these underlying trending distributions can be reconstructed.
\begin{enumerate} \setlength\itemsep{0em}
\item  \hspace*{2mm}  $\widehat{F}_i^{-1}(p_k) :=  \textbf{quantile}(\widehat{F}_i, p_k) \ $ for each $i \in \{1,\dots, N\}, k \in \{1,\dots, P-1\}$
\item \hspace*{2mm} $ \displaystyle w^*_\ell := \sum_{i \in I_\ell} w_i \ $ for each  $\ell \in \{1,\dots, L\}$  
\item  \hspace*{2mm} $\displaystyle x_\ell[k] := \frac{1}{w^*_\ell} \sum_{i \in I_\ell} w_i \hspace*{0.4mm} \widehat{F}_i^{-1}(p_k)   \ $ for each $\ell \in \{1,\dots, L\}, k \in \{1,\dots, P-1\}$
\item \hspace*{2mm}  \textbf{for } $p^* = 0, p_1, p_2, \dots, p_{P-1}$:
\item \hspace*{12mm} $\delta[k] :=$ ``nondecreasing'' if $p_k > p^*$; otherwise $\delta[k] :=$ ``nonincreasing'' 
\item \hspace*{12mm} $y_1, \dots, y_L $ := \textbf{AlternatingProjections}$ \bigg{(}x_1 , \dots ,x_L $ ; $\delta$ ; $\{{w^*}\}_{\ell=1}^L , \{p_k\}_{k=1}^{P-1} \bigg{)}$
\item  \hspace*{12mm} $W[\delta] $ := the value of  (\ref{numericaltrends}) evaluated with $G^{-1}_\ell(p_k) = y_\ell[k] \ \ \  \ \forall \ell, k$
\item \hspace*{12mm} Redefine $\delta[k] :=$ ``nonincreasing'' if $p_k > p^*$; otherwise $\delta[k] :=$ ``nondecreasing''  \\ \hspace*{12mm} and repeat Steps 6 and 7 with the new $\delta$
\item  \hspace*{2mm} Identify $\displaystyle \min_{ \delta} \ W[\delta]$ and return $\widehat{G}^{-1}_\ell(p_k) = y^*_\ell[k]  \ \ \ \forall \ell, k$ \ \ where $y^*$ was produced at the  \\
 \hspace*{2mm} Step 6 or 8 corresponding to $\delta^* := \arg\max W[\delta]$.
\end{enumerate}
\vspace*{-3mm}
\noindent\rule[0.5ex]{\linewidth}{1pt}
} \end{figure}

\begin{figure}[h!]
{\setstretch{1.0} \footnotesize
\noindent\rule[0.5ex]{\linewidth}{1pt}
\textbf{AlternatingProjections Algorithm: }  Finds the Wasserstein-least-squares sequence of vectors which represent valid quantile-functions and a trend whose monotonicity is specified by $\delta$.  \\
\noindent\rule[0.5ex]{\linewidth}{1pt}
\textbf{Input 1:} \hspace*{2mm} Initial sequence of vectors $x^{(0)}_1, \dots, x^{(0)}_L$
 \\
\textbf{Input 2:}   \hspace*{2mm}  Vector $\delta$ whose indices specify directions constraining the quantile-changes over $\ell$. \\
\textbf{Input 3:}  \hspace*{2mm} Weights $w_\ell^* \in \mathbb{R}$ and quantiles to work with $0 < p_1 < \dots < p_{P-1} < 1$ 
\\
\textbf{Output:}  \hspace*{2.5mm}  Sequence of vectors $y^{(t)}_1, \dots, y^{(t)}_L$ where $\forall \ell, k: \ \ y^{(t)}_\ell[k] \le y^{(t)}_\ell[k+1]$ and the sequence \\
 \hspace*{20mm} $y^{(t)}_1[k], \dots, y^{(t)}_L[k]$ is monotone nonincreasing/nondecreasing as specified by $\delta[k]$,  \\
 \hspace*{20mm} provided that    $x^{(0)}_\ell[k] \le x^{(0)}_\ell[k+1] \ \ \ $ for each $\ell, k$ 
\begin{enumerate}  \setlength\itemsep{0em}
\item \hspace*{2mm} $r^{(0)}_\ell[k] := 0$ ,  $ \ s^{(0)}_\ell[k] := 0 \ \ $ for each $\ell \in \{1,\dots, L\}, k \in \{1,\dots, P-1\}$
\item \hspace*{2mm} \textbf{for } $t = 0, 1, 2, \dots$ until convergence:
\item \hspace*{12mm}  $y^{(t)}_1 [k], \dots, y^{(t)}_L [k]  := \textbf{PAVA}\left( \left( x^{(t)}_1[k] +  r^{(t)}_1 [k], w^*_1 \right), \dots, \left(x^{(t)}_L [k] +  r^{(t)}_L [k] ,  w^*_L \right) ; \delta[k] \right)$ \\
\hspace*{13mm}   for each $ k \in \{1,\dots, P-1\}$.  PAVA computes either the least-squares nondecreasing \\
\hspace*{13mm} or nonincreasing  weighted fit, depending on $\delta[k]$. 
\item \hspace*{12mm} $r^{(t+1)}_\ell [k] := x^{(t)}_\ell [k] +  r^{(t)}_\ell [k] - y^{(t)}_\ell [k]  \ \ \ \text{ for each } \ell , k$ 
\item \hspace*{12mm}  $\forall \ell \in \{1,\dots, L \} : \ \hspace*{15mm} x^{(t+1)}_\ell [1], \dots, x^{(t+1)}_\ell [P-1] := \\
\hspace*{1mm} \textbf{PAVA}\left( \left(y^{(t)}_\ell [1] + s^{(t)}_\ell[1] , \frac{p_2 - p_0}{2} \right),  \dots, \left( y^{(t)}_\ell [P-1] + s^{(t)}_\ell[P-1] , \frac{p_P - p_{P-2}}{2} \right) ; \text{``nondecreasing''}  \right)$
\item \hspace*{12mm} $s^{(t+1)}_\ell [k] := y^{(t)}_\ell [k] +  s^{(t)}_\ell [k] - x^{(t+1)}_\ell [k]  \ \ \ \text{ for each } \ell , k$ 
\end{enumerate}
\vspace*{-3mm}
\noindent\rule[0.5ex]{\linewidth}{1pt} \vspace*{-4mm}
} \end{figure} 

\begin{thm}  
The Trend-Fitting algorithm produces valid quantile-functions $\widehat{G}^{-1}_1,\dots, \widehat{G}^{-1}_L$ which solve the numerical version of the TRENDS objective given in (\ref{numericaltrends}).
\label{tfoptimal}
\end{thm} 

Fundamentally, our TF algorithm utilizes Dykstra's method of alternating projections \citep{Boyle1986} to project between the set of  $L$-length sequences of vectors which are monotone in each index over $\ell$ and the set of $L$-length sequences of vectors where each vector represents a valid quantile function.  Despite the iterative nature of alternating projections, we find that the TF algorithm converges extremely quickly in practice.  This procedure has overall computational complexity $O\large(TLP^2 + NP\large)$, which is efficient when $T$ (the total number of projections performed) is small, since both $P$ and $L$ are limited.  The proof of Theorem \ref{tfoptimal} provides much intuition on the TF algorithm (all proofs are relegated to \S\ref{sec:proofs}).  Essentially, once we fix a $\delta$ configuration (specifying which quantiles are decreasing over $\ell$ and which are increasing), our feasible set becomes the intersection of two convex sets between which projection is easy via PAVA.  Furthermore, the second statement in our trend definition limits the number of possible $\delta$ configurations, so we simply solve one convex subproblem for each possible $\delta$ to find the global solution.

\section{Theoretical results}
\label{sec:theory}

Under the model given in  (\ref{trendsmodel}), we establish some results regarding the quality of the $\hat{Q}_1, \dots, \hat{Q}_L$ estimates produced by the TF algorithm.  To develop pragmatic theory, we use finite-sample bounds defined in terms of quantities encountered in practice rather than the true Wasserstein distance (\ref{wasserstein}), which relies on an integral that must be numerically approximated.  Thus, in this section, $d_W(\cdot, \cdot)$ is used to refer to the midpoint-approximation of the $L_2$ Wasserstein integral illustrated in (\ref{numericaltrends}).  In addition to the conditions of model (\ref{trendsmodel}), we make the following simplifications throughout for ease of exposition:
\vspace*{-1.5mm}
\begin{itemize} \setlength\itemsep{0em}
\refstepcounter{assumptioncount} \label{assumption:uniformbatchnumber} 
\item [(A.\arabic{assumptioncount})] The number of batches at each level is the same, i.e. $N_\ell := N_1 =\dots = N_L \ge 1$
\refstepcounter{assumptioncount} \label{assumption:uniformbatchsize}
\item [(A.\arabic{assumptioncount})] The same number of samples are drawn per batch, i.e. $n := n_i$ for all $1 \le i \le N$
\refstepcounter{assumptioncount} \label{assumption:uniformquantiles}
\item [(A.\arabic{assumptioncount})] For $k =  1, \dots, P- 1$:  the $(k / P)$th quantiles of each distribution are considered
\refstepcounter{assumptioncount} \label{assumption:uniformweights}
\item [(A.\arabic{assumptioncount})] Uniform weights are employed, i.e.\ in (\ref{weightedtrends}): $w_i = 1$ for all $i$ 
\end{itemize} 

\begin{thm} \label{consistency} Under model (\ref{trendsmodel}) and additional conditions (A.\ref{assumption:uniformbatchnumber})-(A.\ref{assumption:uniformweights}), suppose the TF algorithm is applied directly to the true quantiles of $P_1,\dots, P_N$.  Then, given any $\epsilon > 0$, the resulting estimates satisfy: \ \ $ d_W(\widehat{G}^{-1}_\ell, G^{-1}_\ell) < \epsilon \ $ for each $\ell \in \{1,\dots, L\}$ \\ \vspace*{-11mm}
\begin{flalign}
&\text{with probability greater than: } \hspace*{6mm} 1 - 2 P L \exp\left(-\frac{ \epsilon^2 N_\ell}{8 \sigma^2 L}  \right) &&
\end{flalign} \vspace*{-12mm}
\end{thm} 

Thus, Theorem \ref{consistency} implies that our estimators are consistent with asymptotic rate $O_P(1/\sqrt{N_\ell})$ if we directly observe the true per-batch quantiles $P_1,\dots, P_N$ (which are contaminated by $\mathcal{E}_i$ under our model).  By using the union-bound, our proof does not require any independence assumptions for the noise introduced at different  quantiles of the same batch.  Because direct quantile-observation is unlikely in practice, we now examine the performance of TRENDS  when these quantiles are instead estimated using $n$ samples from each $P_i$.  Here, we additionally assume:
\vspace{-1.5mm}
\begin{itemize} \setlength\itemsep{0em}
\refstepcounter{assumptioncount} \label{assumptionsamplequantiles} 
\item [(A.\arabic{assumptioncount})] For $i =1,\dots, N:$ quantiles are estimated from $n$ i.i.d.\ samples $X_{1,i}, \dots, X_{n,i} \sim P_i$ 
\refstepcounter{assumptioncount} \label{assumptionnonzerodensity} 
\item [(A.\arabic{assumptioncount})] There is nonzero density at each of the quantiles we estimate, i.e.\ CDF $F_i$ is strictly increasing around each $F^{-1}_i(k/P)$ for $k = 1,\dots, P-1$.
\refstepcounter{assumptioncount} \label{assumptionecdf} 
\item [(A.\arabic{assumptioncount})]  The simple quantile estimator defined below is used for each $k/P, k = 1,\dots, P-1$ \vspace*{-5mm}
$$\widehat{F}^{-1}_i(p) := \inf\{x : \widehat{F}_i(x) \ge p\} \vspace*{-5mm}
$$ where $\widehat{F}_i(\cdot)$ is the empirical CDF computed from $X_{1,i}, \dots, X_{n,i} \sim P_i$.
\end{itemize}

\begin{thm} Under the assumptions of Theorem \ref{consistency} and (A.\ref{assumptionsamplequantiles})-(A.\ref{assumptionecdf}), suppose the TF algorithm is applied to estimated quantiles $\widehat{F}^{-1}_i(k/P)$ for $i = 1,\dots,N,k = 1,\dots,P-1$.  Then, given any $\epsilon > 0$, the resulting estimates satisfy: \ \ $d_W(\widehat{G}^{-1}_\ell, G^{-1}_\ell) < \epsilon$ \ for each $\ell \in \{1,\dots, L\}$ 
 with probability greater than:
\begin{equation} 1 -  2 PL \left[ \exp \left( \frac{- \epsilon^2 N_\ell }{ 32 \sigma^2 L } \right) + N_\ell \exp \left( -2 n \cdot R \left( \frac{\epsilon}{4 \sqrt{L} } 
\right)^2  \right) \right] \label{nosamplebound} \vspace*{-3mm}
\end{equation}
where for $\gamma > 0$:  \vspace*{-5mm}
\begin{align*} 
R(\gamma) & := \displaystyle \min_{i,k} \{ R(\gamma, i, k / P) : i = 1,\dots, N, k = 1,\dots, P - 1 \}      \\
R(\gamma, i, p) & := \min \left\{  F_i \left(F^{-1}_i(p) + \gamma \right) - p \ , \ p - F_i \left(F^{-1}_i(p) - \gamma \right)  \right\} \numberthis \label{rgamma}
\end{align*}
\label{badfinitesample} \vspace*{-13mm}
\end{thm}
Theorem \ref{badfinitesample} is our most general result applying to arbitrary distributions $P_i$ that satisfy basic condition (A.\ref{assumptionnonzerodensity}).  However, the resulting probability-bound may not converge toward to 1 if $\ n \cdot R (\frac{\epsilon}{4 \sqrt{L}})^2 < O(\log N_\ell)$, which occurs if few samples are available per batch (because then the $P_i$ are can be very poorly estimated).  Thus, TRENDS is in general only designed for applications with large per-batch sample sizes.  The bounds obtained under the extremely broad setting of Theorem \ref{badfinitesample} may be significantly improved by instead adopting one of the following stronger assumptions:
\vspace*{-1.5mm}
\begin{itemize} \setlength\itemsep{0em}
\refstepcounter{assumptioncount} \label{assumptionbounded} 
\item [(A.\arabic{assumptioncount})] The simple quantile-estimator defined in (A.\ref{assumptionecdf}) is used, and the support of each $P_i$ is bounded and connected with non-neglible density, i.e. $\exists \text{ constants } B, c > 0  \text{ s.t.\ }  \forall i: \\ 
f_i(x) = 0 \ \forall x \notin [-B, B] \ \text{ and } \ f_i(x) \ge c \ \forall x \in [-B, B] $ \ 
($f_i$ is density for CDF $F_i$).
\refstepcounter{assumptioncount} \label{assumptiongoodquantileestimator} 
\item [(A.\arabic{assumptioncount})] The following is known regarding the quantile-estimation procedure:
\vspace*{-2.2mm}
\begin{enumerate} \setlength\itemsep{-0em}
\item The quantiles of each $P_i$ are estimated independently of the others.
\item The quantile-estimates converge at a sub-Gaussian rate for each quantile of interest, i.e.\ there exists $c > 0 \text{ such that for each } k , i \text{ and any } \epsilon > 0$: \vspace*{-3mm}
\begin{align*}
\Pr \left( \left| \widehat{F}^{-1}_i(k/P) - F^{-1}_i(k/P)    \right|  > \epsilon \right) \le 2 \exp(- 2 n c^2 \epsilon^2)
\end{align*}
\end{enumerate} \vspace*{-3mm}
\end{itemize}

\begin{thm} Under  the assumptions of Theorem \ref{consistency}, conditions (A.\ref{assumptionsamplequantiles}), (A.\ref{assumptionnonzerodensity}), and one of either   (A.\ref{assumptionbounded}) or (A.\ref{assumptiongoodquantileestimator}), the bound in (\ref{nosamplebound}) may be sharpened to ensure that for any $\epsilon > 0$:  \vspace*{-4mm}
$$ d_W(\widehat{G}^{-1}_\ell, G^{-1}_\ell) < \epsilon \ \text{ for each } \ell \in \{1,\dots, L\}  \vspace*{-3mm}
$$
 with probability greater than: \vspace*{-3mm}
\begin{equation} 1 - 2 P \left[ L \exp \left( \frac{- \epsilon^2 N_\ell }{ 32 \sigma^2 L } \right) +  \exp \left(   -  \frac{c^2}{8} N_\ell \hspace*{0.6mm}  n \epsilon^2    \right)    \right]
\end{equation}
\label{boundedfinitesample}
\end{thm}  \vspace*{-10mm}

In Theorem \ref{boundedfinitesample}, the additional assumption of bounded/connected underlying distributions results in a much better finite sample bound that is exponential in both $n$ and $N_\ell$ (implying asymptotic $O_P(N_\ell^{-1/2} + n^{-1/2})$ convergence).    While this condition and the result of Theorem \ref{badfinitesample} assume use of the simple quantile-estimator from (A.\ref{assumptionecdf}), numerous superior procedures have been developed which can likely improve practical convergence rates \citep{ Zielinski2006}.   Assuming guaranteed bounds for the quantile-estimation error (which may be based on both underlying properties of the $P_i$ as well as the estimation procedure), one can also obtain the same exponential bound.  In fact, condition (A.\ref{assumptionbounded}) is an example of a distribution and quantile-estimator combination which achieves the error required by (A.\ref{assumptiongoodquantileestimator}).   Because the boundedness assumption is undesirably limiting, we also derive a similar result under weaker assumptions: 
\vspace*{-1.5mm}
\begin{itemize} \setlength\itemsep{0em}
\refstepcounter{assumptioncount} \label{assumptionbest} 
\item [(A.\arabic{assumptioncount})] Each $P_i$ has connected support with non-neglible interior density and sub-Gaussian tails, i.e.\ there are constants  $B > b > 0 , a > 0, c > 0  \text{ such that for all } i: $
\vspace*{-6mm}
\begin{align*}
(1) & \hspace*{7mm} F_i \text{ is strictly increasing, } \\
(2) & \hspace*{7mm} f_i(x) \ge c \ \forall x \in [-B, B] \ \text{ where $f_i$ is the density function of CDF $F_i$.}   \\
(3) & \hspace*{7mm} \Pr(X_i > x) \le \exp \left(-a \left[ x - (B - b) \right]^2 \right)  \ \text{ if } x > B \\ 
& \hspace*{7mm} \text{ and } \ \Pr(X_i < x) \le \exp \left(-a \left[ x - (-B + b) \right]^2 \right)  \ \text{ if } x < -B
\end{align*} \vspace*{-15mm}
\refstepcounter{assumptioncount}  \label{assumptionbest2}
\item [(A.\arabic{assumptioncount})] Defining $r :=  \min \left\{ 2c^2 \ , \  \frac{ 2a b^2 -  1}{4PB^2} \right\} $, we have $r > 0$, or equivalently, $2ab^2 > 1$.
\refstepcounter{assumptioncount}  \label{assumptionbest3}
\item [(A.\arabic{assumptioncount})]  We avoid estimating extreme quantiles, i.e.\ $F^{-1}_i (k / P) \in (-B, B) \ \ \forall k \in \{1,\dots, P-1\}$
\end{itemize}

\begin{thm} Under the assumptions of Theorems \ref{consistency} and \ref{badfinitesample} as well as conditions (A.\ref{assumptionbest})-(A.\ref{assumptionbest3}), the previous bound in (\ref{nosamplebound}) may be sharpened to ensure that  for all $\epsilon > 0$:  \vspace*{-3mm}
$$ d_W(\widehat{G}^{-1}_\ell, G^{-1}_\ell) < \epsilon \ \text{ for each } \ell \in \{1,\dots, L\}  \vspace*{-5mm}
$$
 with probability greater than: \vspace*{-3mm}
\begin{equation} 1 - 2 P \left[ L \exp \left( \frac{- \epsilon^2 N_\ell }{ 32 \sigma^2 L } \right) +  \exp \left( -  \frac{r}{16}   N_\ell \hspace*{0.6mm} n \epsilon^2     \right)    \right]
\end{equation}
\label{bestfinitesample}
\end{thm} \vspace*{-13mm}

Theorem \ref{bestfinitesample} again provides an exponential bound in both $n$ and $N_\ell$ under a realistic setting where the distributions are small tailed with connected support, and the simple quantile estimator of (A.\ref{assumptionecdf}) is applied at non-extreme quantiles.  Note that while we specified properties of the distributions, noise, and quantile estimation in order to develop this theory, our nonparametric significance tests do not rely on these assumptions.

\section{Simulation study}
\label{sec:simulation}

We perform a simulation which realistically reflects various properties of scRNA-seq data, based on assumptions  similar to those explicitly relied upon by the model of \citet{Kharchenko2014}.  Samples are generated from one of the following choices of the underlying trending distribution sequence $Q_1, \dots, Q_L$ with $L=5$ (additional details in \S\ref{sec:supsim}): 
\begin{enumerate} \setlength\itemsep{0em}
\item[(S$_1$)] $Q_\ell \sim $ NB($r_\ell, p_\ell$) with $r_\ell = 5$ and $p_\ell = 0.3, 0.3, 0.4, 0.5, 0.8$ \ for  $\ell = 1,\dots, 5$.
\item[(S$_2$)] $Q_\ell$ is a mixture of NB($r = 5, p = 0.3$) and NB($r = 5, p = 0.7$) components, with the mixing proportion of the latter ranging over $\lambda_\ell = 0.1, 0.4, 0.8, 0.8, 0.8$ for $\ell = 1,\dots, 5$.
\item[(S$_3$)] $Q_\ell \sim $ NB($r = 5, p = 0.5$) \ for $\ell = 1,\dots, 5$.
\end{enumerate}
NB($r, p$) denotes the negative binomial distribution parameterized by $r$ (target number of successful trials) and $p$ (probability of success in each trial).  To capture various types of noise affecting scRNA-seq measurements (e.g.\ dropout, PCR amplification bias, transcriptional bursting), noise for the $i^{\text{th}}$ batch is introduced (independently of the other batches) via the following steps: rather than sampling from $Q_{\ell_i}$, we instead sample from $P_{\ell_i} \sim $ NB($\widetilde{r}_\ell, \widetilde{p}_\ell$), where $\widetilde{r}_\ell = r_\ell + r_{\text{noise}}$ and $\widetilde{p}_\ell = p_\ell + p_{\text{noise}}$.  Here, $p_{\text{noise}}$, $r_{\text{noise}}$  are independently drawn from centered Gaussian distributions with standard deviations $\sigma$, $10 \cdot \sigma$ respectively ($\sigma$ thus controls the degree of noise).  For the mixture-models in S$_2$, we sample from $P_{\ell_i}$ which is also a mixture of negative binomials (with the same mixing proportions as $Q_{\ell_i}$) where the parameters of both mixing components are perturbed by noise variables $r_{\text{noise}}, \ p_{\text{noise}}$.  To the observations sampled from $P_{\ell_i}$, we finally apply a $\log_{10} (x + 1)$ transform (also applied to the scRNA-seq data in \S\ref{sec:application}) before proceeding with our analysis.

We first investigate the convergence of TRENDS estimates under each of the models S$_1$, S$_2$, and S$_3$, varying $n$, $N_\ell$, and the amount of noise independently.  Figure \ref{convergence} shows the Wasserstein error (sum over $\ell$ of the squared Wasserstein distances between the underlying $Q_\ell$ and estimates thereof) of our TRENDS estimates vs.\ the error of the empirical distributions.  The plot demonstrates rapid convergence of the TRENDS estimator (as guaranteed by our theory in \S\ref{sec:theory}) and shows that TRENDS can produce a much better picture of the underlying distributions than the (noisy) observed empirical distributions.  As shown in Figure \ref{convergence}A, this may occur even in the absence of noise, thanks to the additional structure of the trend-assumption exploited by our estimator.   Thus, when the underlying effects follow a trend, our $\Delta$ statistic provides a much more accurate measure of their magnitude than distances between the empirical distributions.  These results indicate that the largest benefit of our TRENDS approach is for small to moderate sized samples.  

\begin{figure}[h!] \centering
\begin{tabularx}{0.75\textwidth}{X X} 
\includegraphics[width=0.4\textwidth]{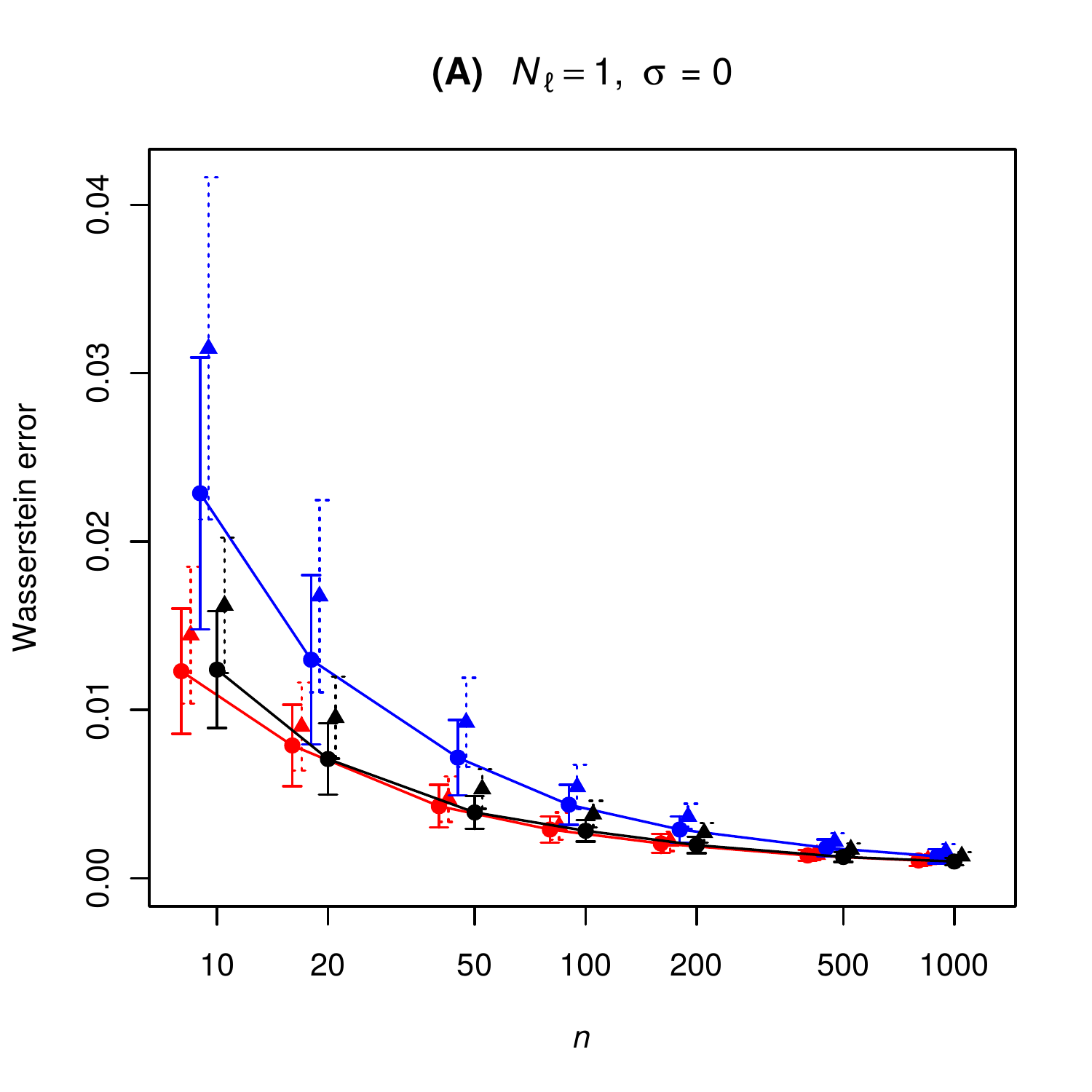} &
\includegraphics[width=0.4\textwidth]{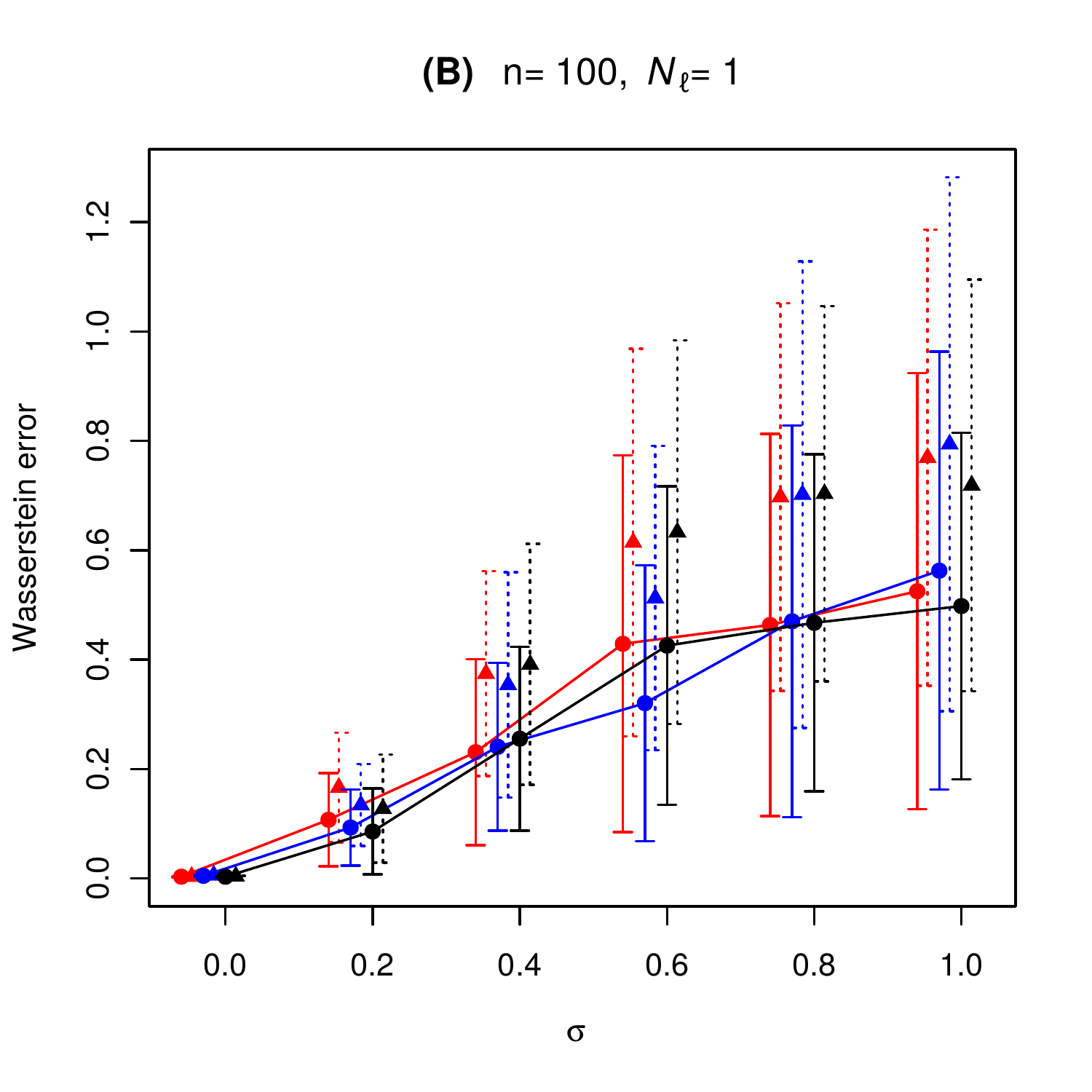}
\\[-11pt]
\includegraphics[width=0.4\textwidth]{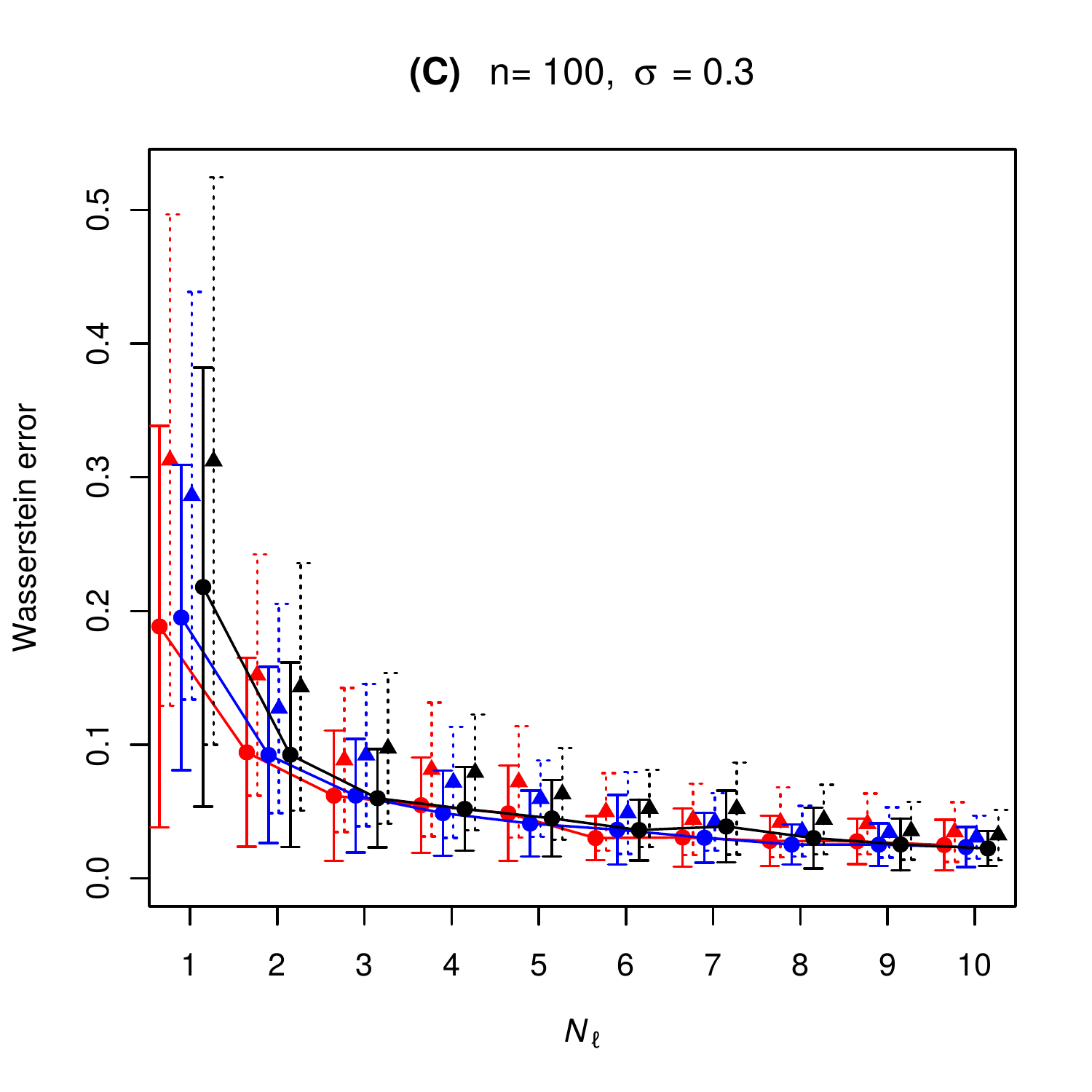} &
\hspace*{5mm} \includegraphics[width=0.35\textwidth]{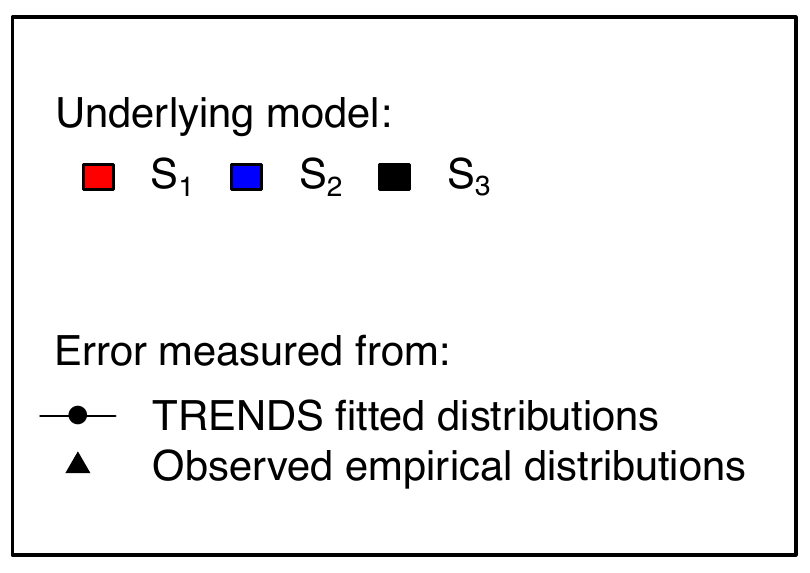}
\end{tabularx} \vspace*{-4mm}
\caption{The Wasserstein error of the TRENDS fitted distributions vs.\ the observed empirical distributions, under models S$_1$ - S$_3$ with various settings of $n$, $\sigma$, and $N_\ell$.  Depicted is the average error (and standard deviation) over 100 repetitions.}
\label{convergence}
\end{figure}

To compare performance, we evaluate TRENDS against alternative methods under our models S$_1$-S$_3$ with  substantial batch-noise ($\sigma = 0.1$).  Fixing $N_\ell = 1, n_i = 1000$ for all $\ell, i$, we generate 400 datasets from the different underlying trending models described above (100 from each of S$_1$, S$_2$, and 200 from S$_3$).  TRENDS is applied to each dataset to obtain a $p$-value (via the permutation procedure described in \S\ref{sec:testing}).  In this analysis, we also apply the following alternative methods (detailed in \S\ref{sec:othermthds}): a linear variant of our TRENDS model (where quantiles are restricted to evolve linearly rather than monotonically), an omnibus-testing approach (using the maximal Kolmogorov-Smirnov (KS) statistic between any pair of distributions), and a measure of the (marginally-normalized) mutual information (MI) between $\ell$ and the values in each batch.  The latter two alternative methods make no underlying assumption and capture arbitrary variation in distributions over $\ell$.    We employ the same approach to ascertain statistical significance (at the 0.05 level) under each method.  All p-values are obtained via permutation-testing (with 1000 permutations).  To correct these p-values for multiple comparisons, we employ the step-down minP adjustment algorithm of \citet{Ge2003}, which cleverly avoids double permutations to remain computationally efficient.

\begin{table}[h!] \centering
\begin{tabular}{c c c c} 
\hline
 Method & FPR & TPR & AUROC    \\ 
 \hline
 TRENDS & 0.02 & 0.35 & 0.87  \\
 Linear-TRENDS & 0.03 & 0.32 & 0.85 \\
 KS & 1.0 & 1.0 &  0.44 \\
 MI & 1.0 & 1.0 & 0.53 \\
 \hline
\end{tabular} 
\caption{False-positive rate (FPR) and true-positive rate (TPR) produced by different methods, as well as AUROC values.  FPR is determined by the fraction of datasets generated under model S$_3$ deemed statistically significant (or S$_1$, S$_2$ for TPR).}
\label{tab:power}
\end{table}

Table \ref{tab:power} demonstrates that methods sensitive to arbitrary differences in distributions are highly susceptible to spurious batch effects (both the KS and MI identify all 400 datasets as statistically significant), whereas our TRENDS method has the lowest false-positive rate, only incorrectly rejecting its null hypothesis for 4 out of the 200 datasets from S$_3$.     TRENDS also exhibits the greatest power in these experiments.  To ascertain how well these methods distinguish the trending data from the non-trending samples, we computed area under the ROC curve (AUROC) by generating ROC curves for each method using its $p$-values (ties broken using test statistics) as a classification-rule for determining which simulated datasets the method would correctly distinguish from constant model S$_3$ at each possible cutoff value.  The results of Table \ref{tab:power}  show that TRENDS is superior at drawing this distinction in these simulations. 

\section{Single cell RNA-seq analysis}
\label{sec:application}

To evaluate the practical utility of our method, we analyze two scRNA-seq time course experiments and compare TRENDS against the alternative approaches described in \S\ref{sec:othermthds}.  The first dataset is from \citet{Trapnell2014} who profiled single-cell transcriptome dynamics of skeletal myoblast cells at 4 time-points during differentiation (myoblasts are embryonic progenitor cells which undergo myogenesis to become muscle cells).  In a second larger-scale scRNA-seq experiment, \citet{Zeisel2015} isolated 1,691 cells from the somatosensory cortex (the brain's sensory system) of juvenile CD1 mice aged P22-P32.  We treat age (in postnatal days) as our batch-labels, with $L = 10$ possible levels.  \S\ref{sec:scrsdets} contains detailed descriptions of the data and our analysis.

Assuming that trending temporal-progression effects on expression reflect each gene's importance in development, we measure the size of these effects using our $\Delta$ statistic (\ref{delta}).  Fitting a separate TRENDS model to each gene's measurements, we thus produce a ranking of the genes' presumed developmental importance.  If instead, one's goal is simply to pinpoint high-confidence candidate genes relevant at all in development (ignoring the degree to which their expression transforms in the developmental progression), then our permutation test can be applied to establish which genes exhibit strong statistical evidence of an underlying nonconstant TREND effect.  For all methods, $p$-values are obtained using the same procedure as in the simulation study (1000 permutations with step-down minP multiple-testing correction).  In these analyses, significance testing (which identifies high-confidence effects) and the $\Delta$ statistic (which identifies very large effects) both produce informative results.

As the myoblast data only contains four $\ell$-levels and one batch from each, the TRENDS permutation test stringently identifies only 20 genes with significant non-constant trend at the 0.05 level (with multiple-testing correction).  
 Terms which are statistically overrepresented in the Gene Ontology (GO) annotations of these significant genes \citep{Kamburov2011}, indicate the known developmental relevance of a large subset (see Figure \ref{myoblastwordcloud}A).  Enriched biological process annotations include ``anatomical structure development'' and ``cardiovascular system development'' (Table \ref{tab:termssignificant}A).  In contrast, the cortex data are much richer, and TRENDS accordingly finds far stronger statistical evidence of trending genes, identifying 212 as significant (at the 0.05 level with multiple testing correction).  A search for GO enriched terms in the annotations of these genes shows a large subset to be developmentally relevant (Figure \ref{myoblastwordcloud}B), with enriched terms such as   ``neurogenesis'' and ``nervous system development'' (Table \ref{tab:termssignificant}B).  Due to the limited batches in these scRNA-seq data (each of which may be corrupted under our model), the TRENDS significance-tests act conservatively (a desirable property given the pervasive noise in scRNA-seq data), identifying small sets of genes we have high-confidence are primarily developmentally relevant.

\begin{figure}[h!] \centering
 \textbf{(A)} \   Myoblast  \hspace*{59mm} \textbf{(B)}  Cortex \ \ \ \  \ \ \ \ \\
\vspace*{-4mm} 
\begin{minipage}[t]{0.52\textwidth}
\vspace*{0pt}
\includegraphics[width=\textwidth]{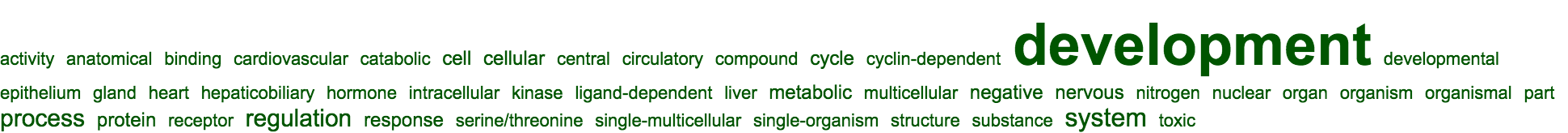}
\end{minipage}  
\hspace*{1mm}
\begin{minipage}[t]{0.45\textwidth}
\vspace*{0pt}
\includegraphics[width=\textwidth]{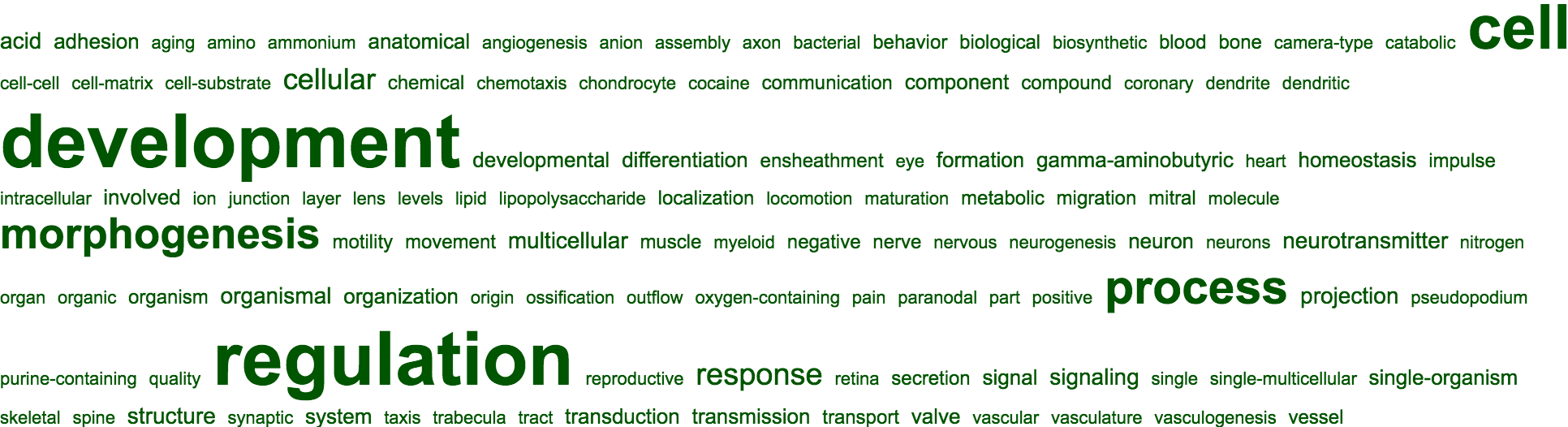}
\end{minipage}
\caption{Word clouds of terms significantly enriched (at the 0.01 level) in GO annotations of the genes with significantly trending expression in each analysis \citep{Kamburov2011}.} 
\label{myoblastwordcloud}
\end{figure}

Ranking the genes by their TRENDS-inferred developmental effects (using $\Delta$), 9 of the top 10 genes in the myoblast experiment have been previously discovered as significant regulators of myogenesis and some are currently employed as standard markers for different stages of differentiation (see Table \ref{tab:top10genes}A).  Also, 7 of the top 10 genes in the cortex analysis have been previously implicated in brain development, particularly in sensory regions (Table \ref{tab:top10genes}B).  Thus, TRENDS accurately assigns the largest inferred effects to clearly developmental genes (see also Table \ref{tab:termseffectsize}).  Since experiments to probe putative candidates require considerable effort, this is a very desirable feature for studying less well-characterized developmental systems than our cortex/myoblast examples.  Figure \ref{geneexamples}A shows TRENDS predicts that MT2A (the gene with the largest $\Delta$-inferred effect in myogenesis and a known regulator of this process)  is universally down-regulated in  development across the entire cell population.  Interestingly, the majority of cells express MT2A at a uniformly high level of $\ge 3$ log FPKM just before differentiation is induced, but almost no cell exhibits this level of  expression 24 hours later.  MT2A expression becomes much more heterogenous with some cells retaining significant  MT2A expression for the remainder of the time course while others have stopped expressing this gene entirely by the end.  TRENDS accounts for all of these different changes via the Wasserstein distance which appropriately quantifies these types of effects across the population.  

Because any gene previously implicated in muscle development is of interest in the myoblast analysis, we can form a lower-bound approximation of the fraction of ``true positives'' discovered by different methods by counting the genes with a GO annotation containing both the words ``muscle'' and ``development'' (e.g.\ ``skeletal muscle tissue development'').  Table \ref{myoblastgoterms} contains all GO annotations meeting this criterion.   Figure \ref{precisionmyoblast}A depicts a pseudo-sensitivity plot based on this approximation over the genes with the highest presumed developmental importance inferred under different methods.   Here, the Tobit models are censored regressions specifically designed for scRNA-seq data, which solely model conditional expectations rather than the full distribution of expression across the cells (see \S\ref{sec:othermthds}).  A larger fraction of the top genes found by TRENDS and our closely-related Linear TRENDS method have been previously annotated for muscle development than top candidates  produced by the other methods.  

We repeat this analysis for the cortex data using a different set of ``ground truth'' GO annotations (listed in Table \ref{cortexgoterms}), and again find that TRENDS produces higher sensitivity than the other approaches (Figure \ref{precisionmyoblast}B) based on this crude measure.  As researchers cannot practically probe a large number of genes in greater detail, it is important that a computational method for developmental gene discovery produces many high ranking true positives which can be verified through limited additional  experimentation.  While TRENDS appears to display greater sensitivity than other methods, we note that it is difficult to evaluate other performance-metrics (e.g.\ specificity) using the scRNA-seq data, since the complete set of genes involved in relevant developmental processes remains unknown.

\begin{figure}[h!] \centering
\includegraphics[width=0.4\textwidth]{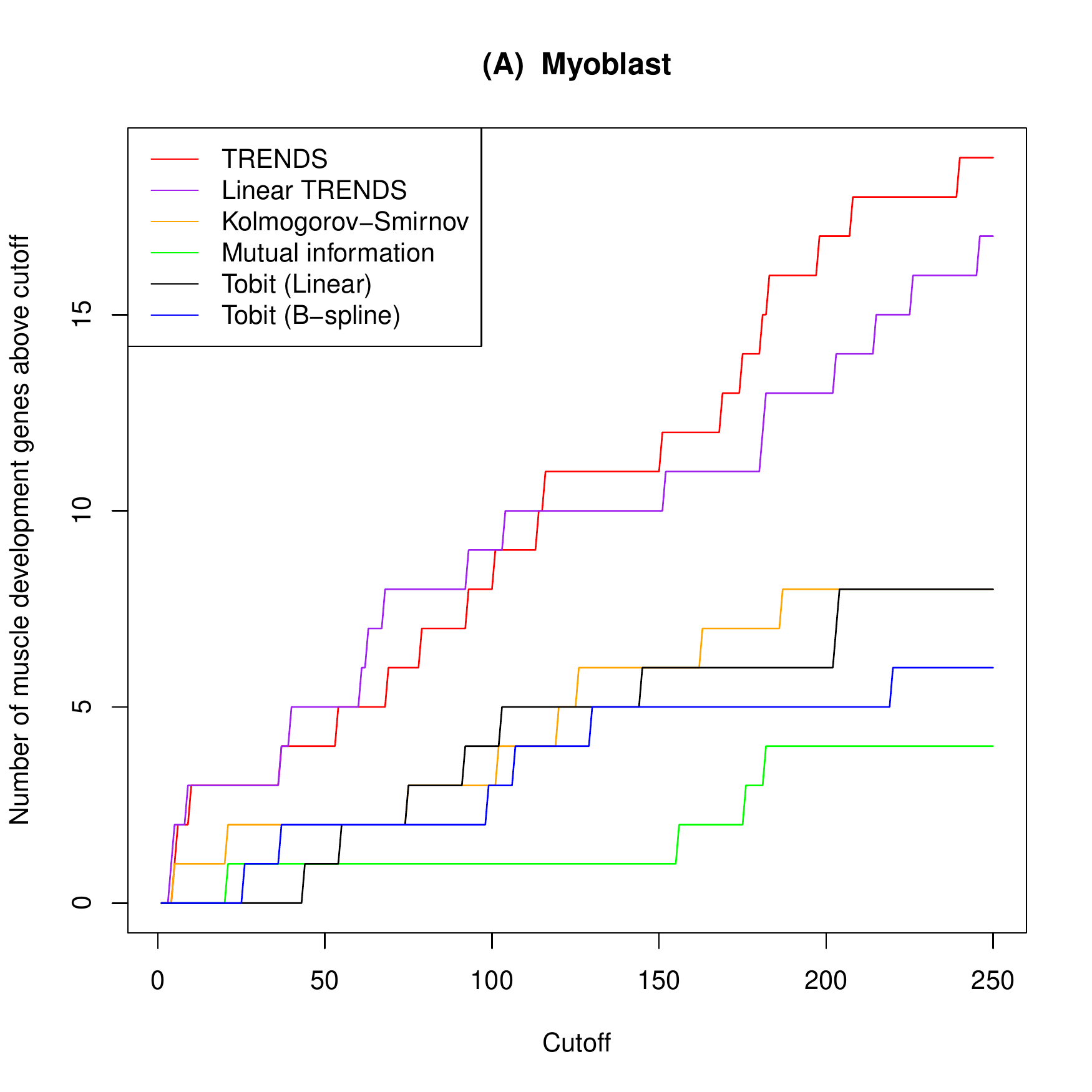}
\includegraphics[width=0.4\textwidth]{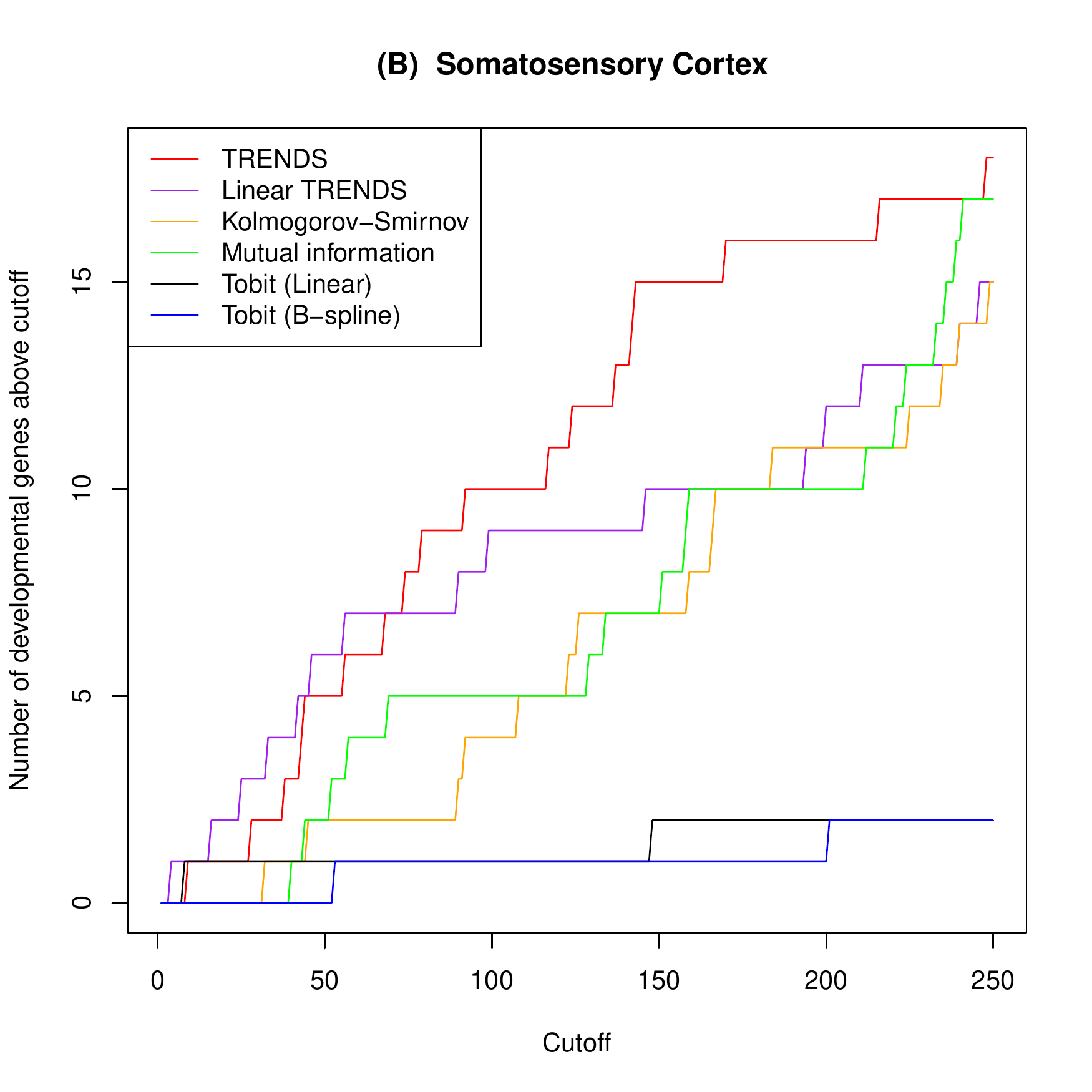}
\caption{Pseudo-sensitivity of various methods based on their ability to identify known developmental genes. (A) the number of  genes with a GO annotation containing both ``muscle'' and ``development'' found in the top $K$ genes (ranked by the different methods for the myoblast data), over increasing $K$.  (B) similar plot for the cortex data, where developmental genes are now those annotated with a relevant GO term from Table \ref{cortexgoterms}.} 
\label{precisionmyoblast}
\end{figure}

The Nestin gene in the myoblast data provides one example demonstrating the importance of treating full expression distributions rather than just mean-effects.  Nestin plays an essential role in myogenesis,  determining the onset and pace of myoblast differentiation, and its overexpression can also bring differentiation to a halt \citep{Pallari2011}, a process possibly underway in the high-expression cells from the later time points depicted in Figure \ref{geneexamples}B.  TRENDS ranks Nestin 35th in terms of  inferred developmental effect-size (with TRENDS $p$-value = 0.02 before multiple-testing correction and 0.09 after), but this gene is overlooked by the scalar regression methods (only ranking 3,291 and 5,094 in the linear / B-spline Tobit results).  Although Figure \ref{geneexamples}B depicts a clear temporal effect on mean Nestin expression, scalar regression does not prioritize this gene because these methods fail to   properly consider the full spectrum of changes affecting different segments of the cell population in the multitude of other genes with similar mean-effects as Nestin.  

Although the closely-related Linear TRENDS model appears to do nearly as well as TRENDS in our Figure \ref{precisionmyoblast} pseudo-sensitivity analysis, we find the linearity assumption overly restrictive, preventing the Linear TRENDS model from identifying important genes like TSPYL5, a nuclear transcription factor which suppresses levels of well-known myogenesis regulator p53 \citep{Epping2011, Porrello2000}.  Linear TRENDS model only assigns this gene a $p$-value of 0.2 whereas TRENDS identifies it as significant ($p = 0.05$), since TSPYL5 expression follows a monotonic trend fairly closely ($R^2 = 0.95$) but is not as well approximated by a linear trend  ($R^2 = 0.68$).

\section{Discussion}
\label{sec:discussion}

While established methods exist to quantify change over a sequence of probability distributions, TRENDS addresses the scientific question of how much of the observed change can be attributed to sequential progression rather than nuisance variation.  Although the TF algorithm resembles quantile-modeling techniques, our ideas are grounded under the unifying lens of the Wasserstein distance, which we use to measure effects (\ref{delta}), goodness-of-fit (\ref{rsquared}), and a distribution-based least-squares fit (\ref{trends}).  Like linear regression, an immensely popular scientific method despite rarely reflecting true underlying relationships, our TRENDS model is not intended to accurately model/predict the data, which are  likely subject to many more effects than our simple \emph{trend} definition encompasses.  Rather, TRENDS quantifies effects of interest, which remain highly interpretable (via our Wasserstein-perspective) despite being considered across fully nonparametric populations.

We recommend our model for data in which the underlying population is heterogeneous (possibly subject to diverse effects), each batch contains many samples ($n_i \ge 50$), and the sequence of levels $L \ge 3$ is short enough that effects of interest should follow persistent trends.   When considering TRENDS analysis, it is important to ensure that the primary effects of interest are a priori expected to follow our trend definition.  For the developmental scRNA-seq data considered in this work, this is a reasonable assumption because the experiments typically focus on a limited window of the underlying process.  Furthermore, the severe prevalence of nuisance variation makes it preferable to identify a high-confidence developmentally-relevant subset of genes (e.g. because they display consistent effects over time), rather than attempting to characterize the complete set of genes displaying interesting effects.  

While our trend definition produces good empirical results in these scRNA-seq analyses (and encompasses various conceptually interesting effects discussed in \S\ref{sec:trendexs}), we emphasize that adopting this assumption narrowly restricts the sort of effects measured by our approach.  Our limited definition is unlikely to characterize more complex effects of interest in general settings (particularly for longer sequences), and future work should explore extensions such as allowing change-points in the model.  Note that our proposed Wasserstein-least-squares fit objective and Wasserstein-$R^2$ measure remain applicable for more general classes of regression functions on distributions.  Furthermore, Lemma \ref{trendl1} provides an alternative definition of a trend which also applies to multidimensional distributions, and thus may be useful for applications such as spatiotemporal modeling.  Nevertheless, the basic TRENDS methodology presented in this work can produce valuable insights.  As simultaneously-profiled  cell numbers grow to the many-thousands thanks to technological advances \citep{Macosko2015}, significant discoveries may be made by studying the evolution of population-wide expression distributions, and TRENDS provides a principled framework for this analysis.  


  \let\oldthebibliography=\thebibliography
  \let\endoldthebibliography=\endthebibliography
  \renewenvironment{thebibliography}[1]{
    \begin{oldthebibliography}{#1} 
      \setlength{\parskip}{0ex}
      \setlength{\itemsep}{0ex}
  }
  {
    \end{oldthebibliography}
  }
  \clearpage
\bibliographystyle{agsm}
\bibliography{TRENDS}




\clearpage
\newpage \beginsupplement
\begin{center}
{\large\bf SUPPLEMENTARY MATERIAL FOR \\  \emph{\papertitle}}
\end{center}

\thispagestyle{empty} 

\vspace*{-5mm}
\begingroup 
\let\orignumberline\numberline
\def\numberline#1{\orignumberline{#1}\kern1ex}
\setcounter{tocdepth}{0}
\tableofcontents
\addtocontents{toc}{\setcounter{tocdepth}{2}}
\endgroup
\clearpage
\setcounter{page}{1} 

\section{Conceptual examples of trends}
\label{sec:trendexs}

\noindent \textbf{Example 1.}  Any sequence of  \emph{stochastically ordered}  distributions follows a trend.  One considers random variable $ X_1 \sim P_1$ less than $X_2 \sim P_2$ in the  stochastic order (which we denote $P_1 \preceq P_2$) if $F_1(x) \ge F_2(x) \ \forall x$ (equivalently characterized as $ \Pr(X_1 > x) \le \Pr(X_2 > x) \ \forall x$)  \citepsi{Shaked1994, Wolfstetter1993}.  Thus, the defining characteristic  of a trend -- the local monotonicity restriction independently applied to each quantile -- is more general than imposing a consistent  \emph{stochastic ordering/dominance} across the distribution-sequence (either  $P_1 \preceq P_2 \preceq  \dots \preceq P_L$ or  $P_1 \succeq P_2 \succeq  \dots \succeq P_L$), as this alternative requires that local changes to each segment of the distribution \emph{all} proceed in the same direction.  
\vspace{3mm}

\noindent \textbf{Example 2.}  Our trend definition also encompasses sequences where the distributions at intermediate values of $\ell$ are \emph{monotonic quantile mixtures} of $P_1$ and $P_L$, i.e.\ 
\begin{align} \forall \ell : \ F^{-1}_\ell = & \ \omega_\ell F^{-1}_1 + (1-\omega_\ell) F^{-1}_L \nonumber \\
& \text{ s.t. } \{\omega_\ell \in [0,1] : \ell = 1,\dots, L\} \text{ form a monotonic sequence }
\label{quantmix} 
\end{align}
Quantile mixtures are typically more appropriate than mixture distributions when there is no evident switching mechanism between distributions in the data-generating process \citepsi{Gilchrist2000}.  Condition (\ref{quantmix}) thus naturally characterizes the situation in which the underlying forces of interest gradually evolve distribution $P_1$ into $P_L$ over $\ell = 1,\dots, L$.  
\vspace{3mm}

\noindent  \textbf{Example 3.}  In many applications, each $P_\ell$ is a mixture of the \emph{same} $K$ underlying subpopulation-specific distributions, where we let  $G_k$ denote the CDF of the $k$th subpopulation-specific distribution (mixing component) with $\ell$-dependent mixing proportion $\pi^{(k)}_\ell$.  Each observed distribution can thus be expressed as:
\begin{equation}
\forall \ell \in \{1,\dots, L\}: \ F_\ell = \sum_{k=1}^K \pi^{(k)}_\ell G_k \ \ \ \text{ where } \forall \ k, \ell : \ \pi^{(k)}_\ell \in [0,1] \text{ , }  \pi^{(K)}_\ell = 1- \sum_{k=1}^{K-1} \pi^{(k)}_\ell
\end{equation}
Here, the effects of interest alter the mixing proportions, so that a fraction of the individuals of one subpopulation transition to become part of another as $\ell$ increases.  Equivalently, this implies that the mixing proportion of one component falls while the probability assigned to the other grows by the same amount.  To ensure the generality of this example, we avoid imposing a specific parameterization for $G_k$.  Rather, we merely assume these mixture components are stochastically ordered with  $G_1 \preceq G_2 \preceq  \dots \preceq G_K$ because subpopulations by definition have distinct characterizations (note that imposing a stochastic ordering is much weaker than requiring $G_k$ to have disjoint support).  

To formalize the types of migration between subpopulations which meet our trend criterion, we conceptualize a graph $\mathcal{G}$ with vertices $1,\dots, K$ representing each mixture component.  If there is migration from subpopulation $i$ to $j > i$ in the transition between level $(\ell - 1) \rightarrow \ell$ (i.e. $\pi^{(i)}_\ell = \pi^{(i)}_{\ell-1} - \Delta$ and $\pi^{(j)}_\ell = \pi^{(j)}_{\ell-1} + \Delta$), then directed edges $i \rightarrow (i+1), (i+1) \rightarrow (i+2), \dots, (j-1) \rightarrow j$ are added to $\mathcal{G}$ (and in the case where $j < i$, these same edges are added to $\mathcal{G}$, only their direction is reversed).   The case in which multiple simultaneous migrations between subpopulations take place between $(\ell - 1) \rightarrow \ell$ is handled more delicately:  First, we identify the sequence $\mathcal{S}$ of operations which produces the optimal transformation from mixing proportions vector $[\pi^{(1)}_{\ell-1}, \dots, \pi^{(K)}_{\ell-1}] \rightarrow [\pi^{(1)}_\ell, \dots, \pi^{(K)}_\ell]$, where the only possible operation is to select $k \in \{1,\dots, K-1\}$ and enact the simultaneous pair of reassignments $\pi^{(k)}_\ell = \pi^{(k)}_{\ell - 1} - \Delta; \ \pi^{(k+1)}_\ell = \pi^{(k+1)}_{\ell-1} + \Delta$ for some $\Delta \in [-1,1]$ whose magnitude is the cost of this operation.  Subsequently, for each operation in $\mathcal{S}$, we introduce an edge into $\mathcal{G}$ between the corresponding nodes $k$ and $k+1$ whose direction is specified by the sign of $\Delta$ (edge $k \rightarrow (k + 1)$ if $\Delta > 0$, the reverse edge otherwise).  

$\mathcal{G}$ is initialized as the empty graph and for $\ell = 2, \dots, L$, the necessary edges are added to the graph corresponding to the mixing-proportion changes between $(\ell - 1) \rightarrow \ell$ as described above.  Then, the sequence of distributions $P_1, \dots, P_L$ follows a trend if $\mathcal{G}$ contains \emph{no} cycles after step $L$ and at most one node with two incoming edges.  Intuitively, this implies that a trend captures the phenomenon in which the underlying forces of progression that induce migration from one subpopulation to a larger one as $\ell$ increases, do not also cause migration in the reverse direction between these subpopulations at different values of $\ell$.  Figure \ref{trendgoodex}D depicts an example of an evolving 3-component mixture model which follows a trend.

\section{Permutation testing with small batch numbers} 
\label{sec:testing}

Unfortunately, in many settings of interest such as most currently existing scRNA-seq time course data, $N$ and $L$ are both small.  This limits the number of possible-permutations of distribution-labels and hence the granularity and accuracy with which we can determine $p$-values in the our test.  Note that TRENDS estimation is completely symmetric with respect to a reversal of the distributions' associated levels (i.e.\ replacing each $\ell_i \leftarrow L - \ell_i + 1$), so if $B$ denotes the number of possible permutations, we can only obtain $p$-values of minimum granularity $2/B$ which may be unsatisfactory in the small $N, L$ regime (e.g.\ $N < 7$).  In the classical tissue-level differential gene expression analyses (in which sample sizes are typically small), this problem has been dealt with by permuting the genes (of which there are many) rather than the sample labels.  However, this approach is not entirely valid as it discards the (often substantial) correlations between genes and has been found to produce suboptimal results \citepsi{Phipson2010}.

 To circumvent these issues, we propose a variant of our label-permutation-based procedure to obtain finer-grained but only approximate $p$-values (in the small $N, L$ setting, rough approximations are all one can hope for since asymptotics-derived $p$-values are also error-prone).  The underlying goal of our heuristic is to produce a richer picture of the null distribution of $R^2$ (at the cost of resorting to approximation), which is accomplished as follows:
 
 \begin{enumerate} 
 \item Shuffle the distributions' $\ell_i$-labels as described above, but now explicitly perform all possible permutations, except for the permutations that produce a sequence $\{\ell_1^{\text{perm}}, \dots, \ell_N^{\text{perm}} \}$ which equals either the sequence of actual labels $\{\ell_1, \dots, \ell_L \}$ or its reverse in which each $\ell_i$ is replaced by $L - \ell_i + 1$. 
 
\item For data in which each distribution $\widehat{P}_i$ is estimated from a set of samples $\{ X_{i,s} \}_{s=1}^{n_i}$, one can obtain a diverse set of $K$ null-distributed datasets from a single permutation of the labels  by employing the bootstrap.  For each $k = 1,\dots,K$ and $i =1,\dots,N$: draw $n_i$ random samples $Z^{(k)}_{i,s}$ with replacement from $\{ X_{i,s} \}_{s=1}^{n_i}$, compute a bootstrapped empirical distribution $\widehat{P}^{(k)}_i$ using $\{Z^{(k)}_{i,s} \}_{s=1}^{n_i}$, and assemble the $k$th null-distributed dataset (under the current labels-permutation) by pairing the bootstrapped empirical distributions with the permuted labels $\ell_i^{\text{perm}}$.  
\item Apply TRENDS to each null-distributed dataset $\{ (\ell_i^{\text{perm}}, \widehat{P}^{(k)}_i) \}_{i=1}^N$ and compute  a $R^2_{\text{perm},k}$ value via  (\ref{rsquared}) which is distributed according to the desired null (where $K=1$ and $\widehat{P}^{(k)}_i = \widehat{P}_i$ if bootstrapping is not performed).
\item Form a smooth approximation of the null distribution by fitting a kernel CDF estimate $\widehat{F}$ to the collection of $(B-2) \cdot K$ null samples $\{R^2_{\text{perm},k} \}$ where $k = 1,\dots, K$ and perm is an index over the possible label-permutations under consideration (we use the  Gaussian kernel with the plug-in bandwidth proposed by Altman and L\'eger, which has worked well even when only 10 samples are available \citepsi{Altman1995}).  Finally, the approximate $p$-value is computed as $\widehat{p} := 1 - \widehat{F}(R^2)$, where $R^2$ corresponds to the fit of TRENDS on the original dataset.

\end{enumerate}
 
Note that under the exchangeability of labels assumed in $H_0$, the sequence of $\ell_i$ corresponding to the actual ordering or its reverse are equally likely a priori as any other permutation of the $\ell_i$.  Thus, Step 1 above is unbiased, despite the omission of two permutations from the set of possibilities.  Producing a much richer null distribution than the empirical version based on few permutation samples, the bootstrap and kernel estimations steps enable us to obtain continuum of (approximate) $p$-values.  Intuitively, our richer approximation is especially preferable for differentiating between significant $p$-values despite its sensitivity to the bandwidth setting, because the standard permutation test offers no information when the actual test statistic is greater than every permuted statistic (a common occurrence if $B$ is small), whereas our approach assigns smaller $p$-values based on the distance of the actual test statistic from the set of permuted values.  Finally, we remark that the kernel estimation step in our $p$-value approximation is similar to the approach of Tsai and Chen \citepsi{Tsai2007}, and point out that as the number of distributions per level $N_\ell$ grows, the approximation factor of  our procedure shrinks, as is the case for $p$-values based on asymptotics which are themselves only approximations.

\addtocontents{toc}{\protect\setcounter{tocdepth}{1}}
\subsection{Evaluating TRENDS $p$-values}
\addtocontents{toc}{\protect\setcounter{tocdepth}{2}}

Under the simulation setup of \S\ref{sec:simulation}, we investigate the performance of our permutation technique to obtain TRENDS $p$-values.  We draw samples from each of the underlying models S$_1$, S$_2$, S$_3$ with $n=100, N_\ell=1$, and $\sigma = 0.1$.  To each simulated dataset (in total, 100 datasets are drawn from each model), we apply the TRENDS model and then determine the significance of the TRENDS $R^2$ via a standard permutation test utilizing all possible permutations of the batch labels (here $L=5$ so the number of distinct possible permuted-$R^2$ values from the null is $5! / 2 = 60$).  We subsequently employ our $p$-value approximation to assess the significance of the same $R^2$ value using the same permutations as before, but with additional bootstrapped samples drawn under each permutation of the batch labels until the total number of null samples is enlarged to at least 1000.  Subsequently, the kernel CDF procedure is applied to these 1000 null samples as detailed in the technique described above for obtaining an approximate $p$-value.  

To compare our approximation with the standard permutation test $p$-value, we require the actual $p$-value of the observed $R^2$ describing the TRENDS fit, which is estimated as follows:  a minimum of $J=10,000$ new datasets (i.e.\ batch sequences) from the same underlying model are drawn in which $\ell$ is randomly permuted among the different batches within a single dataset.  TRENDS $R^2$ values are then computed for each of these null datasets (which resemble the permuted data we use in practice, but each permutation of the labels is matched with freshly sampled batches   corresponding to a new dataset), and we can subsequently define the underlying $p$-value as in permutation testing.  Note that this approach can approximate the actual null distribution of $R^2$ arbitrarily well as we increase $J$, and in our experiments, we begin with $J= 10,000$ and gradually increase up to $1,000,000$ while at least 5 null-$R^2$ values greater than the one observed in the original data have not yet been observed.  Table \ref{pvaltable} demonstrates that our approximation produces much better $p$-values than the basic permutation method. 

\begin{table}[ht]
\centering 
\begin{tabular}{crrrrrrrrrr}
  \hline
 Model &  Average $p$ &  $\mathbb{E}[\widehat{p} - p]$ & SD$(\widehat{p})$ & MSE($\widehat{p}$) &  $\mathbb{E}[p_{\text{perm}} - p]$ & SD$(p_{\text{perm}})$ & MSE($p_{\text{perm}}$)  \\ 
  \hline
S$_1$ & 0.13 & -0.012 & 0.036 & \textbf{1.2e-3} & -0.015 & 0.036 & 1.3e-3 \\ 
S$_2$ & 0.19 & 0.039  & 0.068 &  \textbf{5.2e-3} &  0.085 & 0.117 & 1.8e-2 \\ 
S$_3$ & 0.51 & 0.056  & 0.084 &  \textbf{8.8e-3} & 0.092 & 0.157 & 2.8e-2 \\ 
   \hline
\end{tabular}
\caption{Comparing our approximate $p$-values ($\widehat{p}$) against the standard permutation test ($p_{\text{perm}}$).  Column 2 lists the average true $p$-value (over 100 datasets) for each model S$_1$-S$_3$.}
\label{pvaltable}
\end{table}
\addcontentsline{toc}{subsection}{Table \ref{pvaltable}}

\section{Description of alternative methods}
\label{sec:othermthds}
Here, we describe different methods that TRENDS is compared against.  Note that the methods which model full distributions may be ordered based on increasing generality of the underlying assumption as follows:  Linear TRENDS $\rightarrow$ TRENDS $\rightarrow$ KS / MI.   By selecting a model later in this ordering, one can capture a wider diversity of underlying effects but only with decreased statistical power (and robustness against batch-effects).

\addtocontents{toc}{\protect\setcounter{tocdepth}{1}}
\subsection{Kolmogorov-Smirnov method (KS)}
\addtocontents{toc}{\protect\setcounter{tocdepth}{2}}
This approach performs an omnibus test of the hypothesis that there exist $\ell_1$ and $\ell_2$ such that $\Pr(X \mid \ell_1) \neq \Pr(X \mid \ell_2)$.  As a test statistic and measure of effect-size, we use the maximum Kolmogorov-Smirnov test statistic between these empirical conditional distributions over all possible pairs $\ell_1 < \ell_2 \in \{1,\dots, L\}$.  Statistical significance is assessed via permutation testing, since the usual asymptotics are no longer valid after maximization.   

\addtocontents{toc}{\protect\setcounter{tocdepth}{1}}
\subsection{Mutual information method (MI)}
\addtocontents{toc}{\protect\setcounter{tocdepth}{2}}
Here, we estimate the size of the effect using the mutual information between $\ell$ and $X$.  Because we operate in the fixed-design setting, $\ell$ is technically not a random variable, so we instead employ a conditional variant of the mutual information in which the marginal distribution of $\ell$ is disregarded, following the DREMI method of \citesi{Krishnaswamy2014}.  First, we simply reweigh our batches to ensure the marginal distribution of $\ell$ is uniform over $\{1,\dots, L\}$ in the given labels $\{\ell_i\}_{i=1}^N$.  Subsequently, kernel density estimates of the reweighed joint $(X,\ell)$ distribution as well as each conditional $\Pr(X \mid \ell_1)$ are used to calculate the (conditional) mutual information, which is used to produce a ranking of genes' inferred developmental importance according to this method.  A $p$-value is obtained via permutation testing, using the mutual information as the test statistic.

\addtocontents{toc}{\protect\setcounter{tocdepth}{1}}
\subsection{Linear TRENDS (LT) model}
\addtocontents{toc}{\protect\setcounter{tocdepth}{2}}
This method is very similar to TRENDS, except it uses a more restrictive class of regression functions where each quantile evolves linearly (rather than the assumption of monotonicity used in our trend criterion).  We thus operate on real-valued rather than ordinal covariates (e.g.\ the actual values of the time points $t_\ell$ when available in the scRNA-seq context, or the integer $\ell$-values when there are no definitive numerical batch-labels, as in our simulation study).  Linear TRENDS also relies on our notion of Wasserstein least-squares fit, the $\Delta$ effect-size measure (used to rank genes), and the same permutation-procedure for testing significance in TRENDS (the sole difference between these models is that the former accounts for covariate scaling assuming that effects manifest linearly on this scale).  

A similar linear multiple-quantile regression framework has been previously proposed in numerous contexts, although it is designed only for simultaneously estimating a few specific quantiles of the conditional distribution \citepsi{Takeuchi2006, Bondell2010}.  Takeuchi et al.\ and Bondell et al.\ both fit  this model jointly over the quantiles of interest via a quadratic program with constraints to ensure non-crossing quantiles.  Linear quantile regression (with non crossing) could nonetheless be employed to model the full distribution by simply selecting a grid of quantiles spanning $(0,1)$ as is done in TRENDS, but note that simple scalar measures such as our  $\Delta$ and $R^2$ values do not exist in standard quantile regression which lacks the unifying Wasserstein perspective presented in this work.   

In our setting, the empirical quantiles of each conditional distribution are available, so one can directly employ the usual squared error loss on the empirical quantiles themselves (as done in our TF algorithm)  rather than relying on the quantile regression loss function used by Takeuchi et al.\ and Bondell et al.  Analogous to the proof of Theorem \ref{tfoptimal}, one can easily show that optimizing the squared error loss (on each quantile) implies the distributions constructed from the set of fitted quantiles are the Wasserstein least-squares fit under the restriction that each quantile evolves linearly over $t_\ell$, the time at which the batch is sampled.  By replacing the PAVA step (over $\ell$) of the TF algorithm with standard linear regression (where $t_\ell$ is the sole covariate) and also omitting the $\delta$-search for the split between increasing and decreasing quantiles, our alternating projections method is trivially adapted to fit the set of non-crossing quantile linear regressions under the squared-loss.  In the case where we estimate around $100$ quantiles to represent the entire distributions, we find that this linearized TF algorithm is orders of magnitude faster than the quadratic program, which has difficulty dealing with the large number of constraints required in this setting (these methods are not intended to estimate full distributions).  We therefore fit the Linear TRENDS model using this linearized TF algorithm in our applications (computational efficiency is crucial when the model is fit thousands of times as in our gene-expression analyses), and find that besides the marked runtime improvement, Linear TRENDS produces nearly identical estimates as the linear multiple-quantile regression model of Bondell et al.

\addtocontents{toc}{\protect\setcounter{tocdepth}{1}}
\subsection{Tobit model (censored regression)}
\addtocontents{toc}{\protect\setcounter{tocdepth}{2}}
 \citesi{Trapnell2014} introduce a scalar regression model specifically tailored for the analysis of single-cell gene expression over time (which only considers conditional expectations rather than the complete expression distribution across the cell population).  Their approach ranks genes based on the significance of the regression coefficients in a Tobit-family generalized additive model fit to log-FPKM values vs.\ time.  It is thus assumed that measured expression follows a log-normal distribution, and the Tobit link function is introduced to deal with the scarcity of observed reads from some genes expected to be highly expressed (this missing data issue plagues scRNA-seq measurements due to the small amount of RNA that can be isolated from one cell).  We try both directly regressing $X$ against $t_\ell$ (referring to this generalized linear model as the linear Tobit), as well as initially using a B-spline basis expansion of the $t_\ell$ values so the subsequent Tobit regression can capture diverse nonlinear effects \citepsi{Trapnell2014}.

\section{Simulation study details} 
\label{sec:supsim}

Our negative binomial distribution parameters $r_\ell$ and $p_\ell$ correspond to the arguments \texttt{size} and \texttt{prob} used by the \texttt{NegBinomial} functions in the $R$ \texttt{stats} package (here, a negative binomial random variable represents the number of failures occurring in a series of Bernoulli trials before $r_\ell$ successes take place).  To ensure we are sampling from valid distributions after the introduction of noise, we subsequently enforce the following additional constraints: $\widetilde{r}_\ell \ge 1$,  $0.05 \le \widetilde{p}_\ell   \le  0.95$ before drawing our observations.

\section{Single cell RNA-seq analysis details} 
\label{sec:scrsdets}

 \citetsi{Trapnell2014} recently studied the single-cell transcriptome dynamics of skeletal myoblast cells during differentiation to identify the genes which orchestrate the morphological/functional changes observed in this process.  After inducing differentiation in a culture of primary human myoblast cells, cells were sampled (and sequenced) in batches every 24 hours.  While the microfluidic system in this experiment can capture 96 cells (one batch is sampled per time point), some of the captures contain visible debris and cannot be confirmed to come from a whole single cell.  In addition to discarding these, Trapnell et al.\ stringently omit cells whose libraries were not sequenced deeply ($\ge 1$ million reads), since their analysis uses high-dimensional manifold methods which are not robust to noise.  Because TRENDS is designed to distinguish biological effects from noise, we retain these cells embracing the additional (albeit noisy) insight on underlying expression.  Omitting only the debris-cells, the data\footnote{Myoblast FPKM values are available in the Gene Expression Omnibus under accession GSE52529.} we analyze consists of  17,341 genes profiled in the following number of cells at each time point: 0h: 93 cells, 24h: 93 cells, 48h: 93 cells, 72h: 76 cells.  

In a scRNA-seq experiment of much larger scale, Zeisel et al.\ (2015) isolated 1,691 cells from the somatosensory cortex (the brain's sensory system) of juvenile CD1 mice aged P22-P32.  We treat age (in postnatal days) as our covariate, whose ordinal representation takes one of $L = 10$ possible levels. Numerous batches of cells were captured from some identically-aged mice, implying $N_\ell > 1$ for many  $\ell$, and a total of 14,575 genes have nonzero expression measurements\footnote{We compute FPKM values from the somatosensory cortex sequencing read counts available in the Gene Expression Omnibus under accession GSE60361.} in the sampled cells.

In both analyses, gene expression is  represented in ($\log_{10}(x+1)$ transformed) Fragments Per Kilobase of transcript per Million mapped reads (FPKM) \citep{Trapnell2014}.  Although TRENDS is nonparametric and can be applied to any expression representation, we find log-FPKM values favorable due to their interpretability and direct comparability between different genes.    The methods we compare TRENDS against (\S\ref{sec:othermthds}) are all suited for log-FPKM values and do not hinge on the specific distributional assumptions often required for other expression-measures such as  read counts \citesi{Risso2014} or negative-binomial rates \citesi{Kharchenko2014}.

The word clouds (Figure \ref{myoblastwordcloud}) summarizing enriched biological process terms in the  GO annotations for genes with significantly trending expression were made using the ConsensusPathDB\footnote{ConsensusPathDB Link: \url{http://cpdb.molgen.mpg.de}} tool \citepsi{Kamburov2011}.  Table \ref{tab:termssignificant} provides additional detail listing the most highly enriched terms identified in the significantly trending gene set from each dataset.  
Table \ref{tab:top10genes} contains previously characterized developmental genes found among those with the ten largest TRENDS $\Delta$ values (i.e.\ the genes with the largest inferred effect-size).   
Table \ref{tab:termseffectsize} lists the highly enriched GO terms (again found via ConsensusPathDB) in the 100 genes with largest $\Delta$ values in each dataset.

\begin{singlespace} 

\begin{table}[ht!] \centering
\hspace*{5mm} \textbf{(A) Myoblast}  \hspace*{37mm} \textbf{(B) Somatosensory Cortex} \\ \vspace*{1mm}
{\scriptsize
\begin{tabular}[b]{lcc}
  \hline
Term & p-value & q-value  \\ 
  \hline
    liver development  & 1e-4  & 6e-3 \\
    hepaticobiliary system development & 1e-4  & 6e-3 \\
    anatomical structure development & 3e-4  & 8e-3 \\
    gland development &  3e-4 & 0.03 \\
    system development &  2e-3 & 0.08 \\
    regulation of cyclin-dependent protein \\ \hspace*{2mm}  serine/threonine kinase activity &  2e-3 & 0.08 \\
     single-multicellular organism process & 3e-3 & 0.04 \\
     single-organism \\ \hspace*{2mm} developmental process & 4e-3 & 0.04 \\
     central nervous system development &  5e-3 & 0.07 \\
     cardiovascular system development & 5e-3 & 0.07 \\
     circulatory system development & 5e-3 & 0.07 \\
     multicellular organismal \\ \hspace*{2mm} development & 5e-3 & 0.08 \\
     cellular nitrogen compound \\
     catabolic process & 5e-3 & 0.07 \\
     response to hormone & 5e-3 & 0.08 \\
     nervous system development & e-3 & 0.07  \\
      heart development & 5e-3 & 0.08 \\
      regulation of cell cycle & 6e-3 & 0.07 \\
      organ development & 6e-3 & 0.08 \\
  \hline
\end{tabular} \hspace*{1.1mm} 
\begin{tabular}[b]{lcc}
  \hline
Term & p-value & q-value  \\ 
  \hline
   transmission of nerve impulse & 6e-8 & 2e-5 \\
   multicellular organismal \\ \hspace*{2mm} signaling & 1e-7 & 3e-5 \\
   cell communication & 6e-7 & 7e-5 \\
   neuron differentiation & 1e-6  & 2e-4 \\
   cell development & 3e-6  & 2e-4 \\
   ensheathment of neurons & 3e-6  & 2e-4 \\
   axon ensheathment & 3e-6  &3e-4 \\
   single organism signaling & 4-e6  & 3e-4 \\
   neurogenesis & 1e-5  & 1e-3 \\
   regulation of biological quality	& 1e-5  & 4e-4 \\
	system development	& 1e-5  & 5e-4  \\
   neuron projection development &  1e-5 & 1e-3 \\
   cell projection organization &  1e-5 &5e-4  \\
	single-organism cellular process & 2e-5  & 4e-4 \\
	neuron development		& 2e-5   & 1e-3 \\
	anatomical structure development & 3e-5   & 5e-4 \\
	 nervous system development &  3e-5 & 2e-3   \\
	cellular developmental process & 5e-5  & 6e-4  \\
	cell differentiation & 6e-5   & 2e-3 \\
	single-organism \\ \hspace*{2mm}  developmental process & 7e-5 & 7e-4 \\
  \hline
\end{tabular}
} \vspace*{-4mm}
\caption{Most highly enriched terms in the biological process annotations of significantly trending genes.  The $p$-values correspond to the statistical significance of each term's enrichment in the set of genes (false-discovery-rate correction produces $q$-values).}
\label{tab:termssignificant}
\end{table}
\addcontentsline{toc}{subsection}{Table \ref{tab:termssignificant}}

\begin{table}[ht!] \centering
\textbf{(A) Myoblast} \\
{\scriptsize
\begin{tabular}[b]{lcccl}
  \hline
Gene  & $\Delta$ & $R^2$ & $p$-value &   Developmental Evidence   \\ 
  \hline
MT2A & 0.46 & 0.98 & 0.11  &  \citesi{Apostolova1999} \\ 
ACTA2 & 0.44 & 0.99 & 0.08  &  \citesi{Petschnik2010} \\ 
MT1L & 0.43 & 0.99 & 0.09  &  \citesi{Apostolova1999} \\ 
TNNT1 & 0.42 & 0.95 & 0.13  &  \citesi{Sebastian2013} \\ 
MYLPF & 0.41 & 0.99 & 0.03  &  \citesi{Sebastian2013} \\ 
 MYH3 & 0.39 & 0.99 & 0.04  &  \citesi{Trapnell2014} \\ 
MT1E  & 0.39 & 0.99 & 0.11  &  \citesi{Apostolova1999} \\ 
AC004702.2 & 0.37 & 0.99 & 0.23  & Unknown  \\
FABP3 & 0.35 & 0.98 & 0.18  & \citesi{Myers2013} \\ 
DKK1 & 0.34 & 0.99 & 0.12  & \citesi{Han2011} \\ 
\hline
\end{tabular}} \\ \vspace*{2mm}
 \textbf{(B) Somatosensory Cortex} \\
{\scriptsize
\begin{tabular}[b]{lcccl}
  \hline
Gene & $\Delta$ & $R^2$ & $p$-value &  Developmental Evidence  \\ 
  \hline
 Sst & 0.23 & 0.22 & 0.05  & \citesi{Zeisel2015} \\
Xist & 0.14 & 0.09 & 0.35 & Unknown  \\
Ptgds & 0.13 & 0.24 & 0.02 & \citesi{Trimarco2014} \\
Plp1 & 0.13 & 0.16 & 0.14  & \citesi{Zeisel2015}  \\
Mog  & 0.13 & 0.13 & 0.16  & \citesi{Zeisel2015} \\
 Npy  & 0.12 & 0.11 & 0.23  & \citesi{Zeisel2015} \\
 Rps26 & 0.11 & 0.12 & 0.20 & Unknown  \\
 Tsix & 0.11 & 0.12 & 0.23 & Unknown \\
Apod  & 0.11 & 0.16 & 0.11 &  \citesi{Sanchez2002} \\
Ermn & 0.10 & 0.11 & 0.20  & \citesi{Zeisel2015} \\
  \hline
\end{tabular}} \vspace*{-1mm}
\caption{The top ten inferred developmental genes (with the largest $\Delta$ value) from each experiment.  Shown are the TRENDS $\Delta$, $R^2$, and $p$-value (after multiple-testing correction) for each gene, as well as existing literature (if known) which previously characterized the gene as playing an important role in developmental processes.}
\label{tab:top10genes}
\end{table}
\addcontentsline{toc}{subsection}{Table \ref{tab:top10genes}}

\vspace*{2mm}

\begin{table}[ht!] \centering
\hspace*{5mm} \textbf{(A) Myoblast}  \hspace*{37mm} \textbf{(B) Somatosensory Cortex}  \\ \vspace*{1mm}
{\scriptsize
\begin{tabular}[b]{lcc}
  \hline
 Term & p-value & q-value  \\ 
 \hline
 actin-mediated cell contraction & 4e-9 & 9e-7 \\
 muscle structure development & 6e-9 & 1e-6 \\
striated muscle tissue development & 8e-9 & 9e-7 \\
muscle tissue development & 1e-8 & 2e-6 \\
muscle organ development & 1e-8 & 2e-6 \\
response to zinc ion & 2e-8 & 2e-6 \\
actin filament-based movement & 3e-8  & 2e-6 \\
organ development & 1e-7 & 1e-5 \\
muscle system process & 1e-7 & 7e-6 \\
response to inorganic substance & 2e-7 & 1e-5 \\
muscle contraction & 2e-7 & 2e-5 \\
negative regulation of growth & 2e-7 & 1e-5 \\
response to metal ion & 2e-7 & 1e-5 \\
 mitotic cell cycle & 3e-7 & 1e-5 \\
response to transition \\ \hspace*{2mm}  metal nanoparticle & 5e-7 & 2e-5 \\
cellular response to metal ion & 5e-7 & 3e-5 \\
cellular response to \\ \hspace*{2mm}  inorganic substance & 1e-6 & 4e-5 \\
muscle cell development & 2e-6 & 6e-5 \\
cell cycle & 5e-6 & 2e-4 \\
muscle tissue morphogenesis & 6e-6 & 2e-4 \\
muscle organ morphogenesis & 9e-6 & 2e-4 \\
 heart development & 1e-5 & 4e-4 \\
  regulation of mitotic cell cycle & 1e-5  & 6e-4 \\
  striated muscle cell development & 2e-5 & 6e-4 \\
 \hline
\end{tabular}}  \hspace*{1.1mm} 
{\scriptsize
\begin{tabular}[b]{lcc}
  \hline
 Term & p-value & q-value  \\ 
 \hline
ensheathment of neurons & 2e-10 & 3e-8 \\
axon ensheathment & 2e-10  & 5e-8 \\
cellular homeostasis & 3e-8 & 2e-6 \\
cellular chemical homeostasis & 4e-8 & 4e-6 \\
transmission of nerve impulse & 7e-8 & 5e-6 \\
multicellular organismal signaling & 1e-7 & 6e-6 \\
glial cell differentiation & 3e-7 & 2e-5 \\
regulation of biological quality & 4e-7 & 2e-5 \\
glial cell development & 7e-7 & 3e-5 \\
chemical homeostasis & 2e-6 & 6e-5 \\
 response to inorganic substance & 4e-6 & 1e-4 \\
 homeostatic process & 8e-6 & 3e-4 \\
nervous system development & 1e-5 & 5e-4 \\
response to metal ion & 2e-5 & 6e-4 \\
response to oxygen-\\ \hspace*{2mm} containing compound & 4e-5 & 1e-3 \\
system development & 6e-5 & 1e-3 \\
central nervous system development & 6e-5 & 2e-3 \\
detoxification of copper ion  & 7e-5 & 2e-3 \\
response to steroid \\ \hspace*{2mm}  hormone stimulus & 1e-4 & 2e-3 \\
response to lipid  & 1e-4 & 2e-3 \\
  response to reactive oxygen species  & 2e-4 & 3e-3 \\
 response to toxic substance   & 2e-4 & 3e-3 \\
  anatomical structure development   & 2e-4 & 6e-3 \\
neurogenesis      & 3e-4 & 5e-3 \\
 \hline
\end{tabular}} \vspace*{-2mm}
\caption{Most highly enriched terms in the biological process annotations of the top 100 genes with largest $\Delta$ values in each experiment.  The $p$-values correspond to the statistical significance of each term's enrichment in the set of genes (false-discovery-rate correction produces $q$-values).}
\label{tab:termseffectsize}
\end{table}
\addcontentsline{toc}{subsection}{Table \ref{tab:termseffectsize}}

\begin{table}[h!]
\centering
{\scriptsize \begin{tabular}{rcl}
  \hline
 & Gene Ontology ID & Annotation Term \\ 
  \hline
1 & GO:0048745 & smooth muscle tissue development \\ 
  2 & GO:0048747 & muscle fiber development \\ 
  3 & GO:0048742 & regulation of skeletal muscle fiber development \\ 
  4 & GO:0048739 & cardiac muscle fiber development \\ 
  5 & GO:0048635 & negative regulation of muscle organ development \\ 
  6 & GO:0007517 & muscle organ development \\ 
  7 & GO:0007519 & skeletal muscle tissue development \\ 
  8 & GO:0048743 & positive regulation of skeletal muscle fiber development \\ 
  9 & GO:0048738 & cardiac muscle tissue development \\ 
  10 & GO:0055013 & cardiac muscle cell development \\ 
  11 & GO:0048741 & skeletal muscle fiber development \\ 
  12 & GO:0055014 & atrial cardiac muscle cell development \\ 
  13 & GO:0055015 & ventricular cardiac muscle cell development \\ 
  14 & GO:0048643 & positive regulation of skeletal muscle tissue development \\ 
  15 & GO:0097084 & vascular smooth muscle cell development \\ 
  16 & GO:0060948 & cardiac vascular smooth muscle cell development \\ 
  17 & GO:0055001 & muscle cell development \\ 
  18 & GO:0055026 & negative regulation of cardiac muscle tissue development \\ 
  19 & GO:0045843 & negative regulation of striated muscle tissue development \\ 
  20 & GO:0016202 & regulation of striated muscle tissue development \\ 
  21 & GO:0048642 & negative regulation of skeletal muscle tissue development \\ 
  22 & GO:0055024 & regulation of cardiac muscle tissue development \\ 
  23 & GO:0061049 & cell growth involved in cardiac muscle cell development \\ 
  24 & GO:0014706 & striated muscle tissue development \\ 
  25 & GO:0007525 & somatic muscle development \\ 
  26 & GO:0061052 & negative regulation of cell growth involved in  cardiac muscle cell development \\ 
  27 & GO:0045844 & positive regulation of striated muscle tissue development \\ 
  28 & GO:0014707 & branchiomeric skeletal muscle development \\ 
  29 & GO:0007522 & visceral muscle development \\ 
  30 & GO:0048641 & regulation of skeletal muscle tissue development \\ 
  31 & GO:1901863 & positive regulation of muscle tissue development \\ 
  32 & GO:0072208 & metanephric smooth muscle tissue development \\ 
  33 & GO:0003229 & ventricular cardiac muscle tissue development \\ 
  34 & GO:0060538 & skeletal muscle organ development \\ 
  35 & GO:0061050 & regulation of cell growth involved in cardiac muscle cell development \\ 
  36 & GO:0055020 & positive regulation of cardiac muscle fiber development \\ 
  37 & GO:0061061 & muscle structure development \\ 
  38 & GO:0061051 & positive regulation of cell growth involved in cardiac muscle cell development \\ 
  39 & GO:0055002 & striated muscle cell development \\ 
  40 & GO:0060537 & muscle tissue development \\ 
  41 & GO:0007527 & adult somatic muscle development \\ 
  42 & GO:0002074 & extraocular skeletal muscle development \\ 
   \hline
 
\end{tabular}}
\caption{A list of all GO annotation terms containing both the words ``muscle'' and ``development'', used to produce the pseudo-sensitivity plots in Figure \ref{precisionmyoblast}A.}
\label{myoblastgoterms}
\end{table}
\addcontentsline{toc}{subsection}{Table \ref{myoblastgoterms}}

\begin{table}[ht!]
\centering
{\footnotesize
\begin{tabular}{rll}
  \hline
 & Gene Ontology ID & Annotation Term \\ 
  \hline
1 & GO:0007420 & brain development \\ 
2 & GO:0007399 & nervous system development \\ 
3 & GO:0014003  & oligodendrocyte development \\ 
4 & GO:0021860  & pyramidal neuron development \\ 
5 & GO:0022008 & neurogenesis \\ 
  \hline
\end{tabular}}
\caption{A list of the GO annotation terms relevant to the somatosensory cortex development, used to produce the pseudo-sensitivity plots in Figure \ref{precisionmyoblast}B.   This brain region is primarily composed of oligodendrocyte and pyramidal neuron cells \citep{Zeisel2015}.}
\label{cortexgoterms}
\end{table}
\addcontentsline{toc}{subsection}{Table \ref{cortexgoterms}}

\begin{figure}[h!] \centering
\includegraphics[width= 0.68\textwidth]{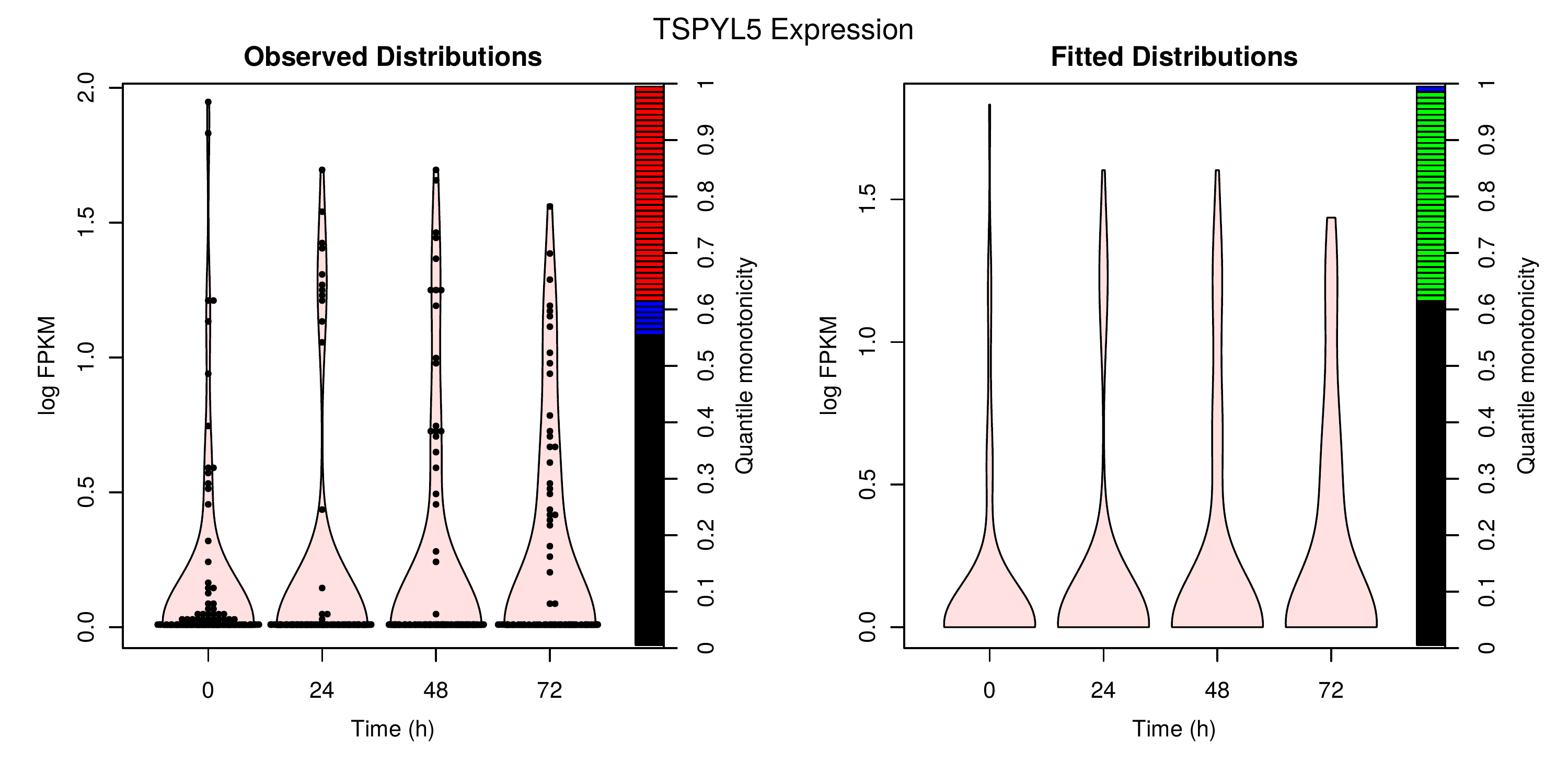} \hspace*{1.3mm} \raisebox{0.7\height}{\includegraphics[width=0.2 \textwidth]{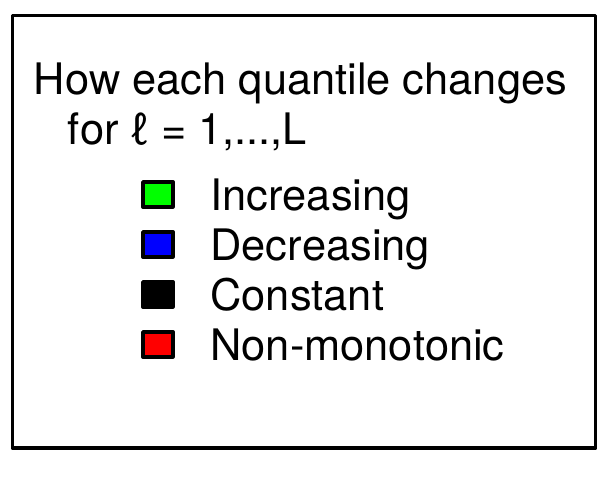}} \vspace*{-5mm} \\
\caption{Violin plots depicting the empirical distribution of TSPYL5 expression measured in myoblast cells (on left), and the corresponding TRENDS fitted distributions (on right). Each point shows a sampled cell.}
\label{tspyl5}
\end{figure}
\addcontentsline{toc}{subsection}{Figure \ref{tspyl5}}

\end{singlespace} 
\clearpage

\section{Model checking}
\label{sec:checking}

In this section, we perform another simulation to demonstrate our proposed procedure for checking whether the TRENDS model is appropriate in analyses lacking prior domain knowledge about the effects of interest.  
Samples are generated from one of the following choices of the underlying trending distribution sequence $Q_1, \dots, Q_L$ (with $L=7$): 
\begin{enumerate} \setlength\itemsep{0em}
\item[(R$_1$)] $Q_\ell \sim N(0, 1)$  \ for  $\ell = 1,\dots, 7$.
\item[(R$_2$)] $Q_\ell \sim N(\mu_\ell, 1)$ with $\mu_\ell = 0, 0.1, 0.1, 0.2, 0.5, 0.9, 1$ \ for  $\ell = 1,\dots, 7$.
\item[(R$_3$)] $Q_\ell  \sim N(\mu_\ell, 1)$ with $\mu_\ell = 0, 0.1, 0.3, 0.5, 0.4, 0.2, 0$ \ for  $\ell = 1,\dots, 7$.
\end{enumerate}
Note that the underlying sequence of distributions for R$_3$ severely violates our trend condition. 
Under each of these models, observed values for the $i$th batch is generated according to $x_{i,s} = \widetilde{x}_{i,s} +  z_i$ where $\widetilde{x}_{i,s} \overset{iid}{\sim} Q_{\ell_i}$, and we independently draw a single noise-variable (i.e.\ batch-effect) $z_i \sim N(0, \sigma^2)$ for the entire batch.

For each quantile $p \in (0,1)$ used in our TRENDS-fit, we compute the value of the empirical residual function $\widehat{\mathcal{E}}_i(p) = \widehat{F}^{-1}_i(p) -  \widehat{G}^{-1}_{\ell_i}(p)$, where $\widehat{F}^{-1}_i$ denotes the empirical quantiles of the distribution for the $i$th batch (estimated from $\{x_{i,s}\}_{s=1}^{n_i} \sim P_{i}$) 
 and $\widehat{G}^{-1}_{\ell_i}$ denote the fitted quantiles produced by the TF algorithm applied the data (corresponding to inferred trending distributions  $Q_{\ell_i}$).  Figure \ref{fig:violator} depicts a diagnostic plot showing the distribution of $\widehat{\mathcal{E}}_i(p)$ vs.\ $\ell$ when TRENDS is fit to data from each of these models.   Based on the clear pattern displayed by the residuals in the R$_3$ plot, one can easily correctly conclude that the TRENDS model is not very appropriate for this dataset.  In contrast, the residual functions appear random for data from the other two underlying settings (which meet our TRENDS assumptions).  

\begin{figure}[h!] \centering
\begin{tabular}{ccc} 
(A) R$_1$  & (B) R$_2$ & (C) R$_3$
\\
\includegraphics[width=0.33\textwidth]{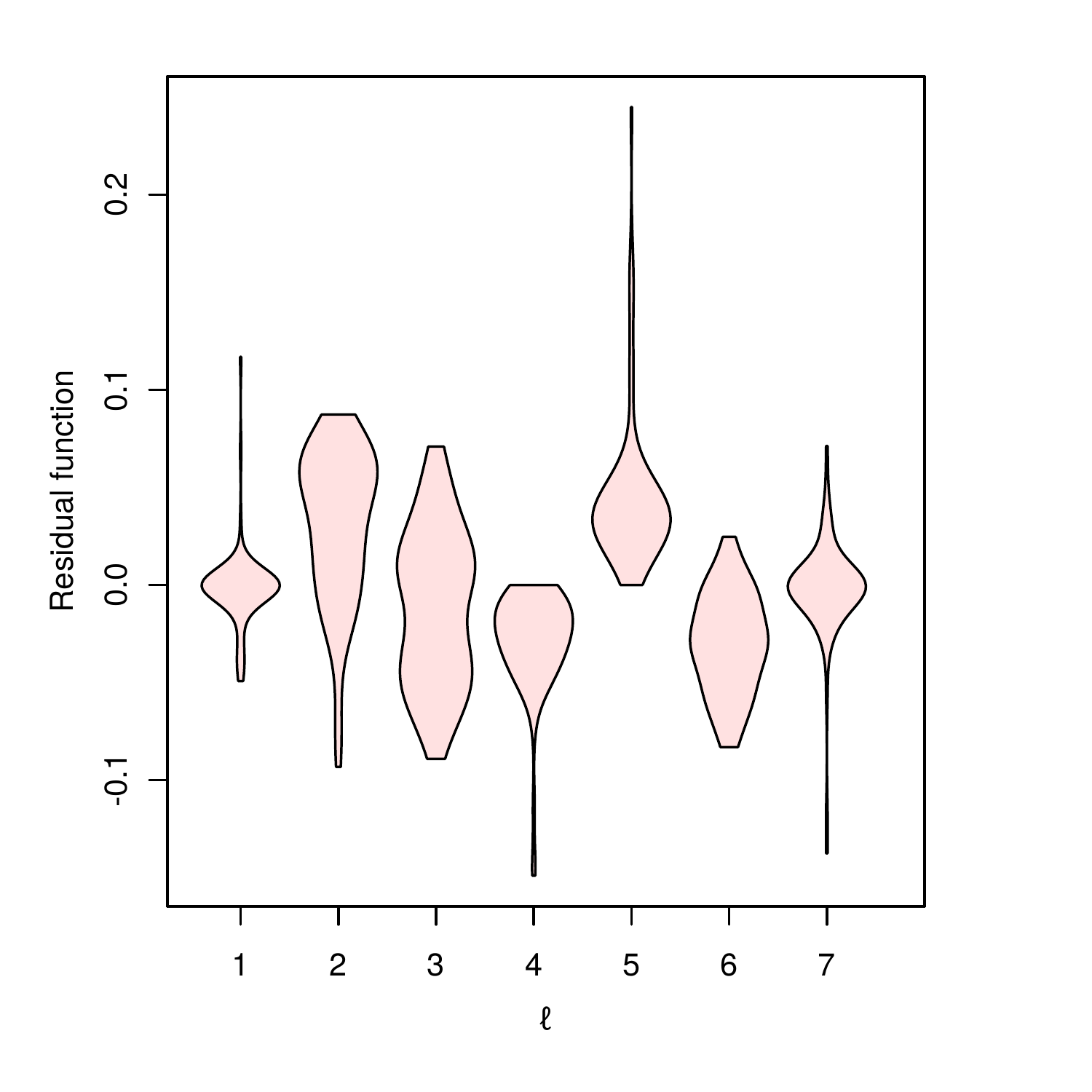}
& \includegraphics[width=0.33\textwidth]{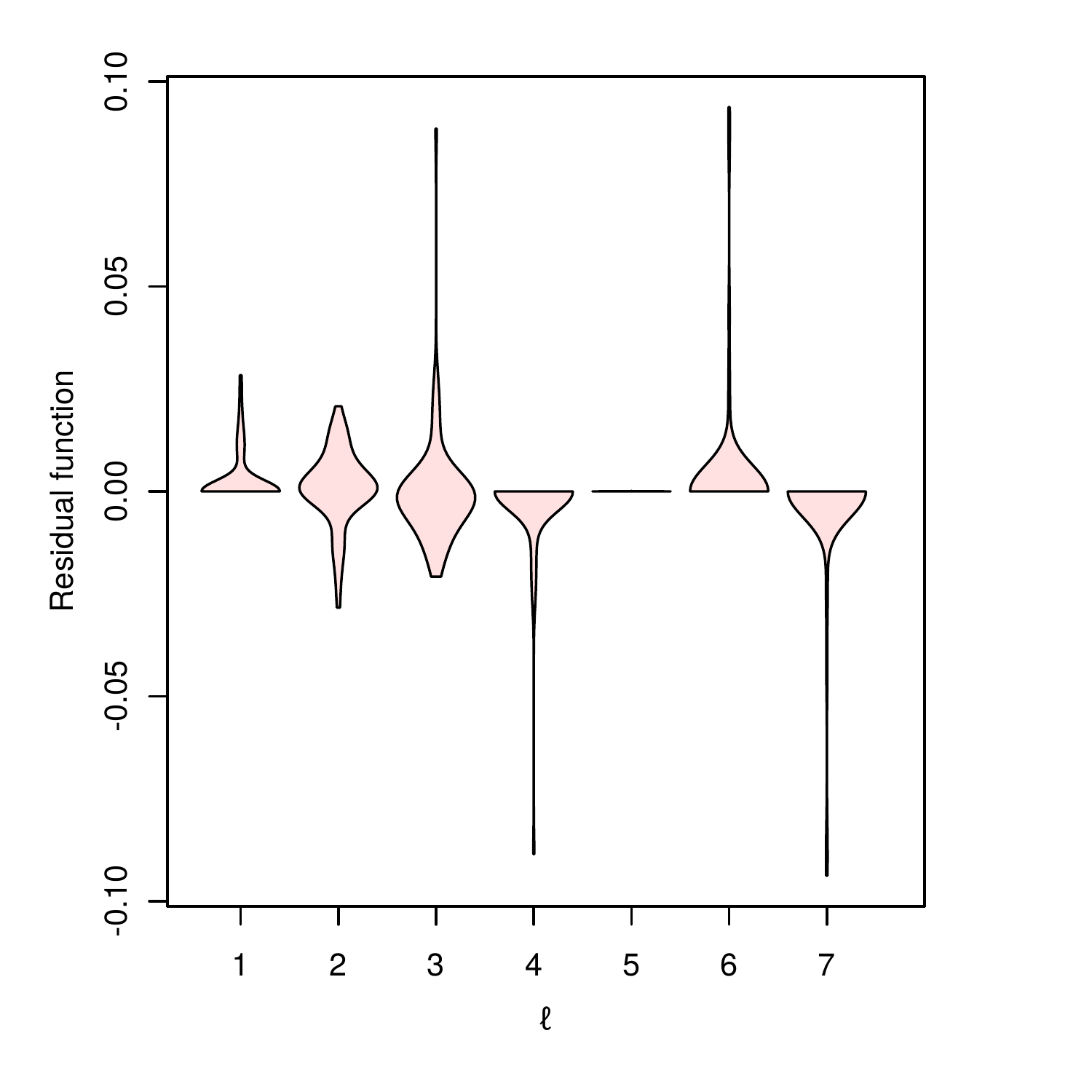}
& \includegraphics[width=0.33\textwidth]{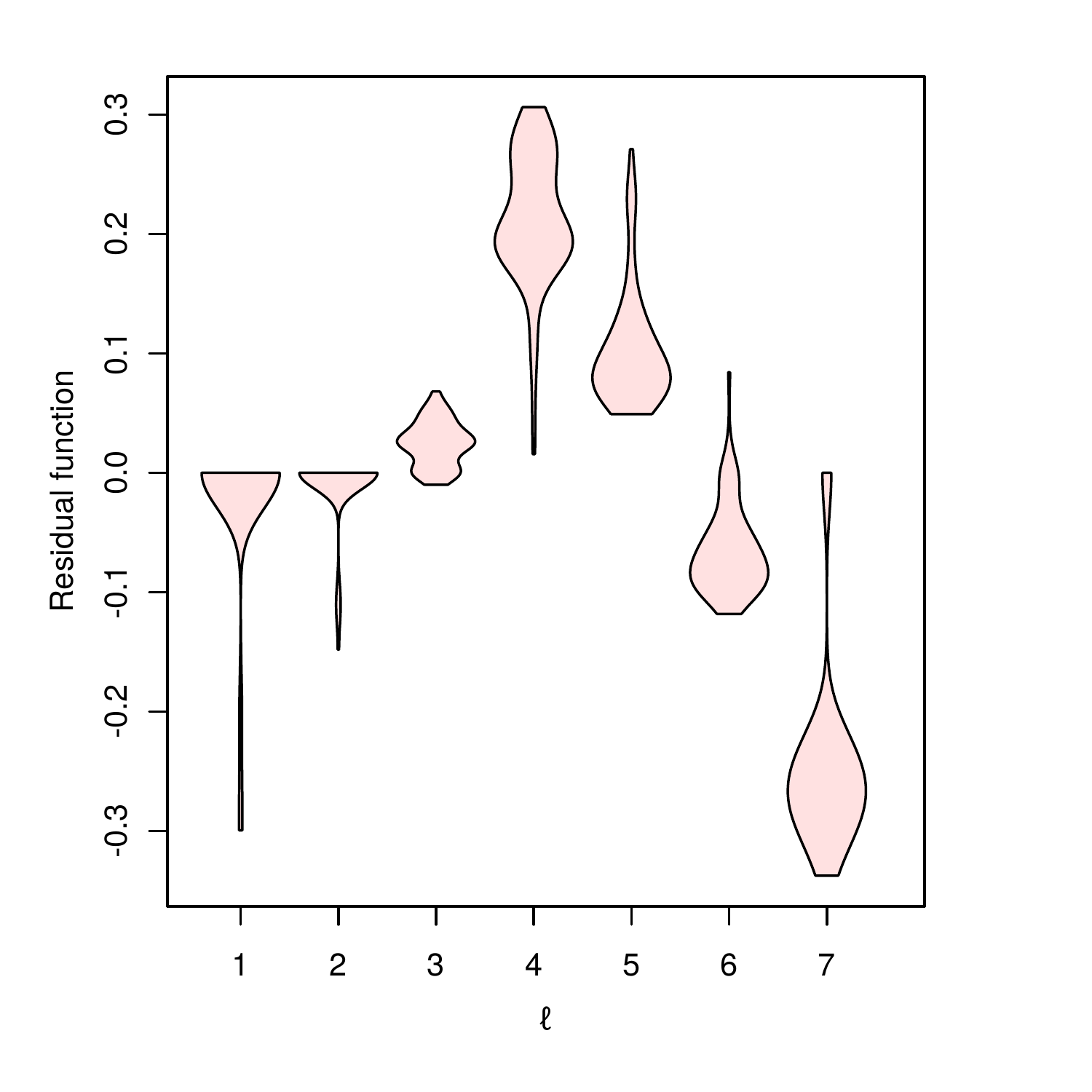}
\end{tabular} \vspace*{-2mm}
\caption{Diagnostic plot of the residual functions $\widehat{\mathcal{E}}_i(p)$ when TRENDS is fit to data from each underlying setting R$_1$, R$_2$, R$_3$ ($N_\ell = 1, n_i = 1000, \sigma = 0.1$).  For each batch $i$, the plot depicts a kernel density estimate of the values taken by  $\widehat{\mathcal{E}}_i(p)$ over $p=0.01,0.02,\dots,0.99$. } 
\label{fig:violator}
\end{figure} 

Under this simulation, we can evaluate the performance of our TRENDS estimates of misspecified effects.   
Motivated by our $\Delta$ statistic and Lemma \ref{trendl1}, we employ the $L_1$ Wasserstein distance to define the true overall sequential-progression effect in this simulation as $\Delta_{\text{true}} = \sum_{\ell =2}^{L} d_{L_1}(Q_{\ell-1}, Q_{\ell})$, which is simply 1 for setting R$_3$.  When all $N_\ell =1$ (one batch per level), we can simply incorporate the Wasserstein distances between adjacent observed empirical distributions $\Delta_{\text{emp}} = \sum_{\ell =2}^{L} d_{L_1}(P_{\ell-1}, P_{\ell})$ as a basic estimate of $\Delta_{\text{true}}$.  Note that the batch-effects cause $\Delta_{\text{emp}}$ to have inflated variance beyond random-sampling deviations in the empirical quantile-estimates.  In contrast, the $\Delta_{\text{TRENDS}}$ estimate produced by our TRENDS model is downwardly biased when applied to data from R$_3$, because of our restriction to monotone quantiles.  Even in this misspecified setting, Figure \ref{fig:misspecified} shows that under non-trivial amounts of noise, $\Delta_{\text{TRENDS}}$ remains a far superior estimator of $\Delta_{\text{true}}$ than $\Delta_{\text{emp}}$, which is highly susceptible to variation arising from these  batch-effects.  

\begin{figure}[h!] \centering
\includegraphics[width=0.5\textwidth]{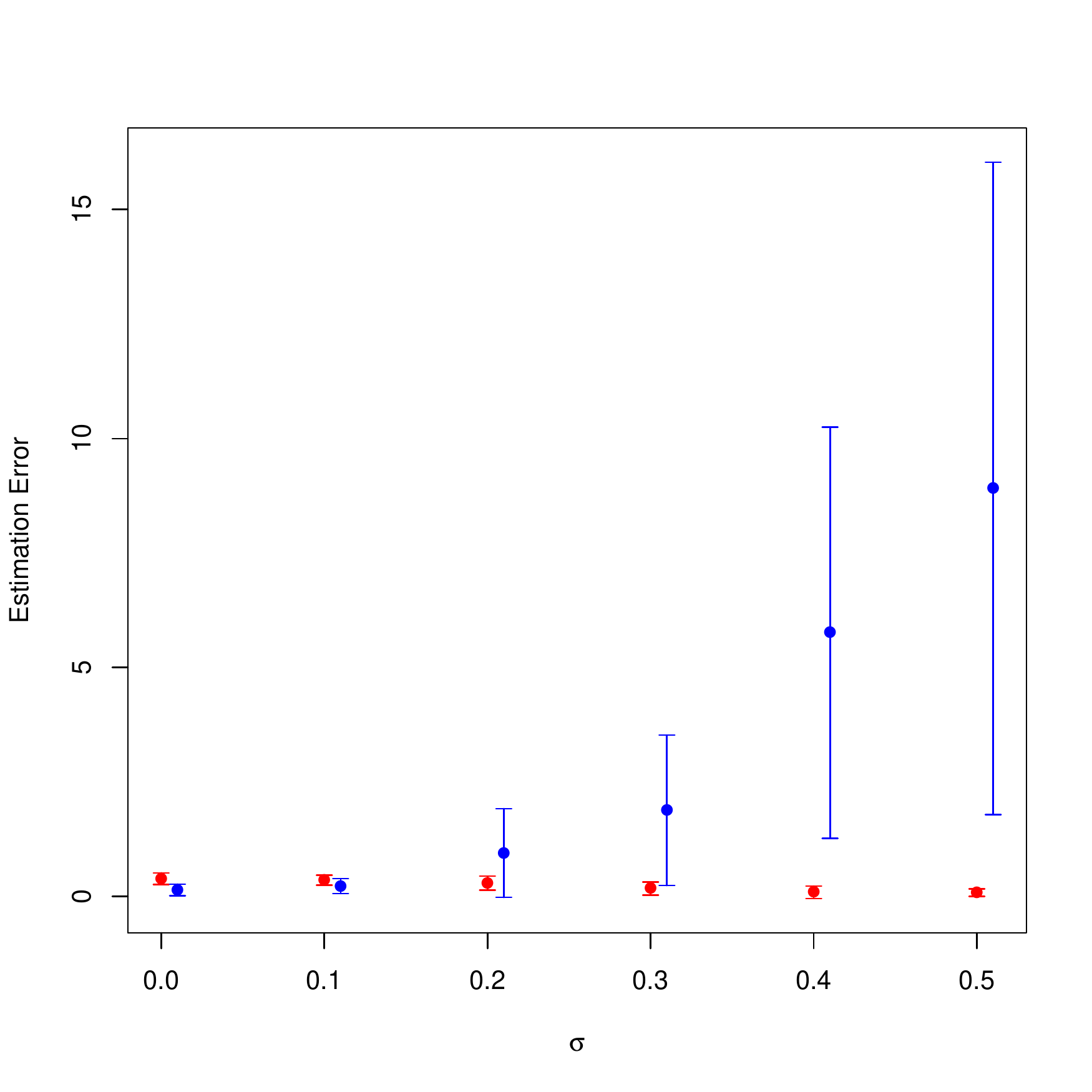}
\caption{The mean/standard-deviation of the squared error of $\Delta_{\text{emp}}$ estimates (blue) and  $\Delta_{\text{TRENDS}}$ estimates (red) over 100 datasets drawn from R$_3$ (under each value of $\sigma$, with $n_i = 100$ for each batch). } 
\label{fig:misspecified}
\end{figure} 

Finally, we investigate the residual functions when TRENDS is fit to the scRNA-seq data from genes known to play a major role in regulating developmental processes.  Figure \ref{fig:genecheck} does not indicate any systematic pattern in the residuals that would suggest our model is inappropriate for these data.

\begin{figure}[h!] \centering
\begin{tabular}{ccc} 
(A) MT2A  & (B) Nestin & (C) TSPYL6 \\
\includegraphics[width=0.33\textwidth]{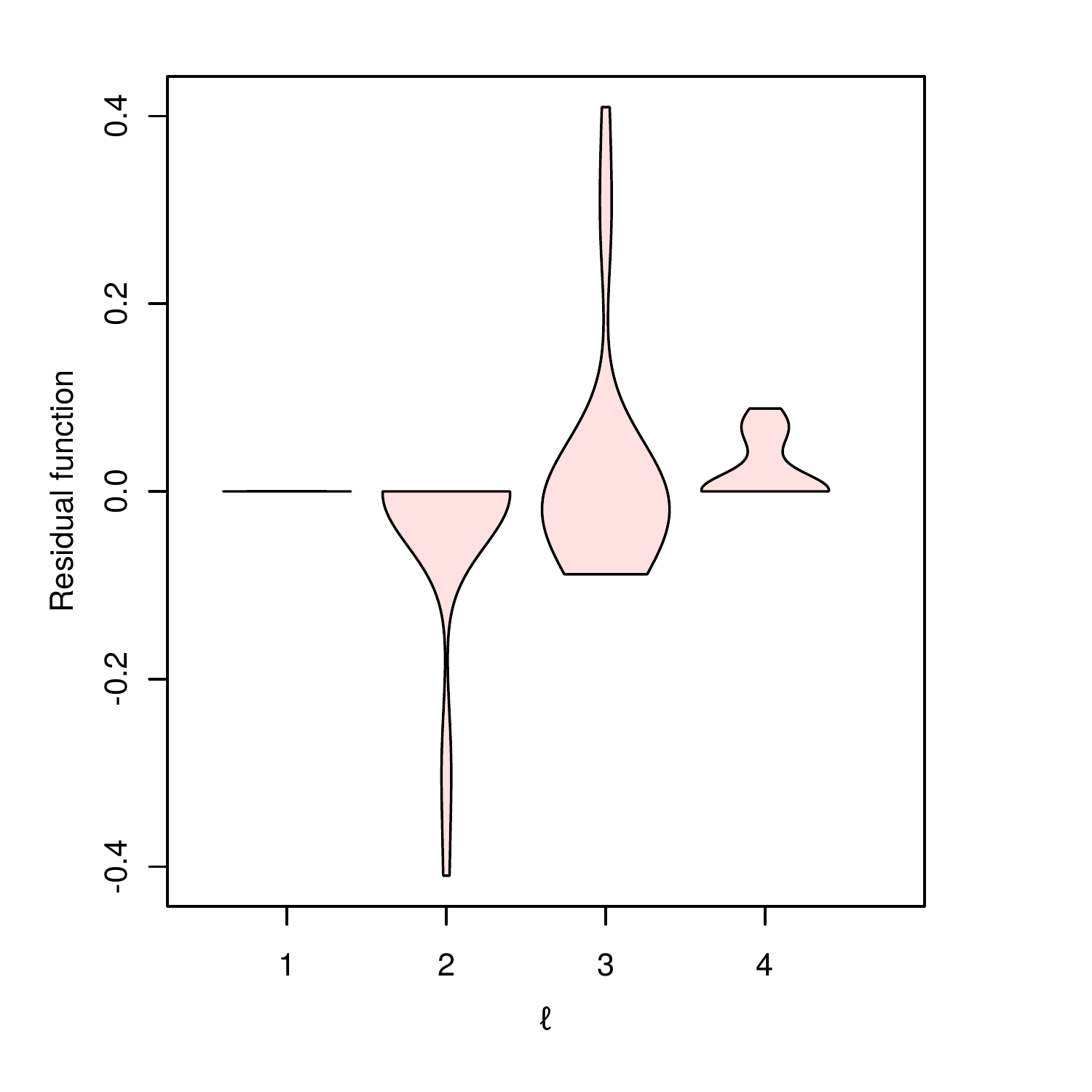} &
\includegraphics[width=0.33\textwidth]{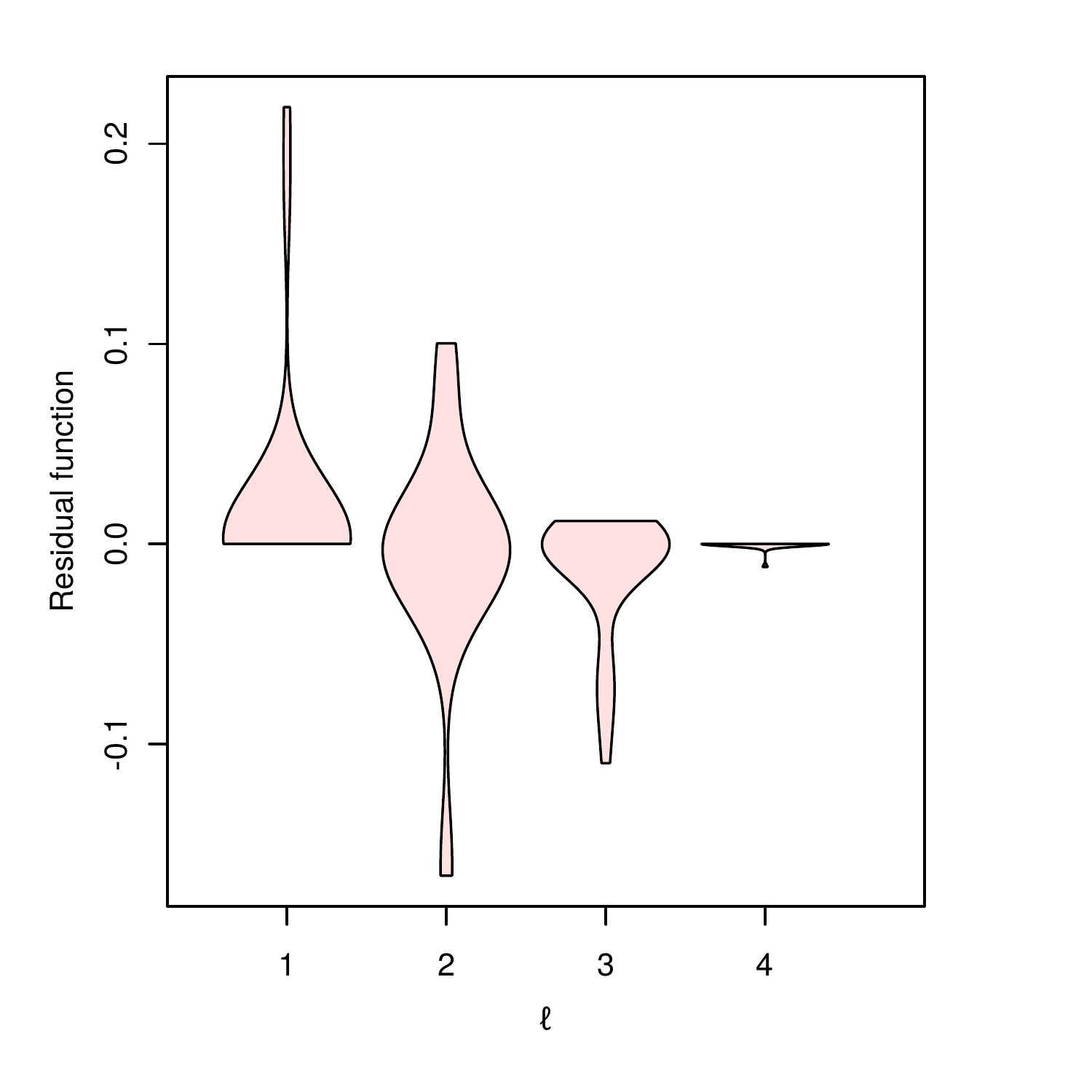} &
\includegraphics[width=0.33\textwidth]{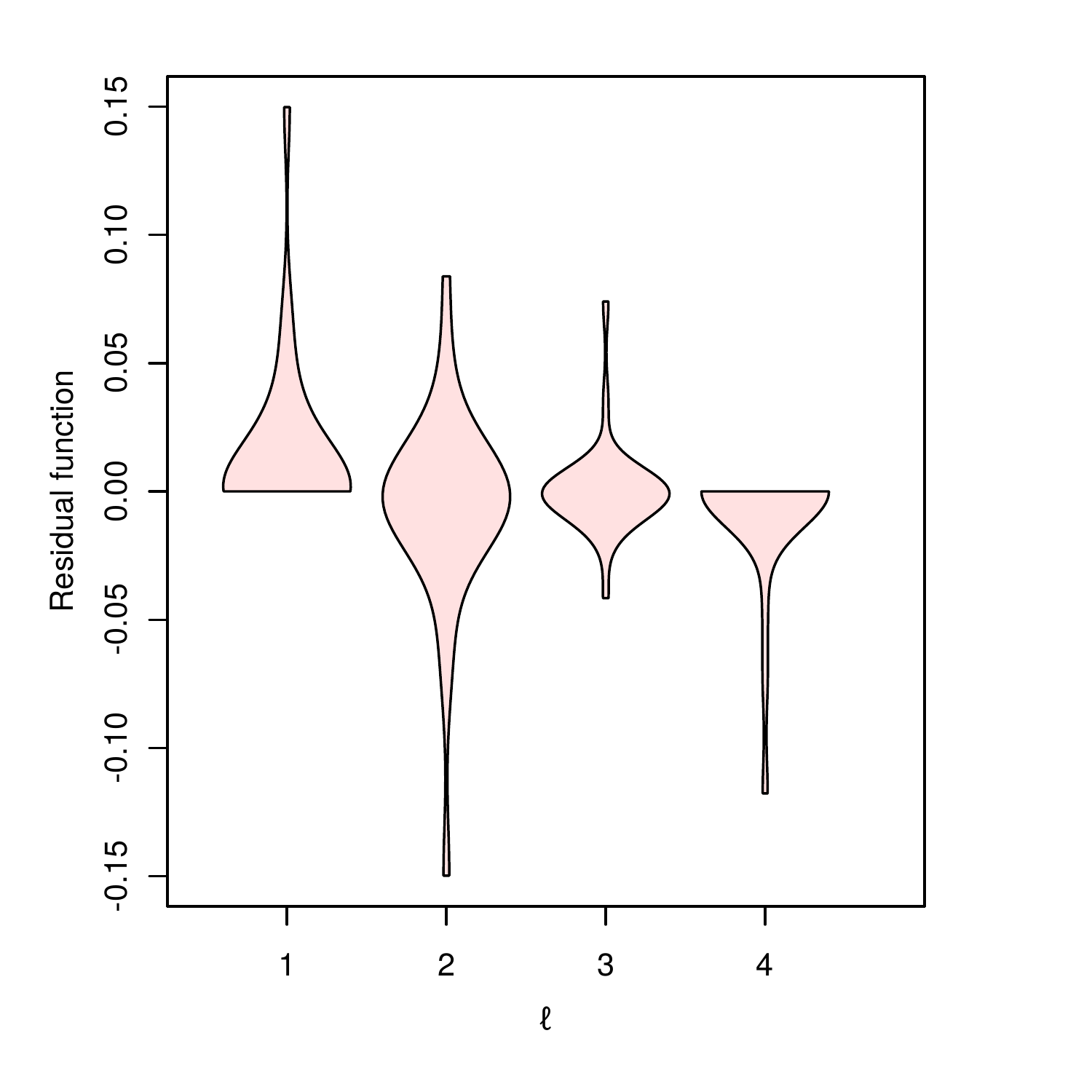}
\end{tabular} \vspace*{-2mm}
\caption{Diagnostic plot of the residual functions $\widehat{\mathcal{E}}_i(p)$ for TRENDS fit to scRNA-seq data from known regulatory genes of myoblast development.  For each batch $i$, the plot depicts a kernel density estimate of the values taken by  $\widehat{\mathcal{E}}_i(p)$ over $p=0.01,0.02,\dots,0.99$.} 
\label{fig:genecheck}
\end{figure} 

\clearpage

\section{ACS income distribution analysis}
\label{sec:acs}

To demonstrate the broader utility of TRENDS beyond scRNA-seq analysis, we present a brief study of incomes in various industries during the years 2007-2013 following the economic recession.  Our goal is to quantify and compare effects across different industries' incomes during this post-recession period.  Rather than measuring ephemeral decline/rebound in this analysis, our interests lie in consistent effects which enduringly altered an industry's incomes through 2013.  American Consensus Survey (ACS) reported income data from 12,020,419 individuals across the USA in the years 2007-2013 were obtained from the Integrated Public Use Microdata Series \citepsi{Ruggle2010}.  After filtering out individuals with missing or \$1 and under  reported income, the data consists of 257 industries from which at least 100 people were surveyed in each of the years under consideration.   We fit TRENDS to the data from each industry separately, treating the observations from each year as a single batch and year-index in this time series as the label ($\ell= 1,\dots, 7$).  

\begin{table}[ht]
\centering {\footnotesize
\begin{tabular}{lrrrr}
  \hline
 Industry & $R^2$ & $p$-value & $\Delta$ \\ 
  \hline
 \textcolor{magenta}{Other information services} & 0.97 & 0.02 & 5465  \\ 
 \textcolor{magenta}{Software publishers} & 0.78 & 0.10  & 2991 \\ 
 \textcolor{magenta}{Electronic auctions} & 0.86 & 0.04  & 2584 \\ 
\textcolor{blue}{Oil and gas extraction } & 0.78 & 0.12  & 2454 \\ 
\textcolor{blue}{Miscellaneous petroleum and coal products} & 0.52  & 0.38 & 2415 \\ 
 \textcolor{magenta}{Other telecommunication services} & 0.80 & 0.07 & 2414 \\ 
\textcolor{red}{Pharmaceutical and medicine manufacturing} & 0.98 & 0.04 & 2220 \\ 
\textcolor{green}{Management of companies and enterprises} & 0.66 & 0.12 & 2194 \\ 
Metal ore mining & 0.89  & 0.02 & 2074 \\ 
\textcolor{blue}{Support activities for mining} & 0.88 & 0.03 & 1915 \\ 
 \textcolor{blue}{Electric and gas, and other combinations} & 0.82  & 0.03 & 1910 \\ 
\textcolor{green}{Non-depository credit and related activities} & 0.92 & 0.06   & 1860 \\ 
Sound recording industries & 0.51 & 0.38 & 1731 \\ 
\textcolor{red}{Electronic component and product manufacturing} & 0.99 & 0.02 & 1719  \\ 
\textcolor{green}{Securities, commodities, funds, trusts, and other financial investments} & 0.57 & 0.23 & 1665  \\ 
\textcolor{red}{Agricultural chemical manufacturing} & 0.77 & 0.09 & 1635 \\ 
\textcolor{red}{Communications, and audio and video equipment manufacturing} & 0.72  & 0.09 & 1628 \\ 
\textcolor{blue}{Pipeline transportation} & 0.70  & 0.14 & 1620 \\ 
\textcolor{blue}{Coal mining} & 0.90 & 0.04  & 1573 \\ 
\textcolor{blue}{Natural gas distribution} & 0.69  & 0.11 & 1546 \\ 
   \hline
\end{tabular}}
\caption{The 20 industries with annual incomes most affected by temporal progression from 2007-2013 (as inferred by TRENDS).  Broader sectors are:  manufacturing (red), business/finance (green), energy (blue), technology (magenta).}
\label{IPUMStop20}
\end{table}

\begin{figure}[h!] \centering
\includegraphics[width=0.8\textwidth]{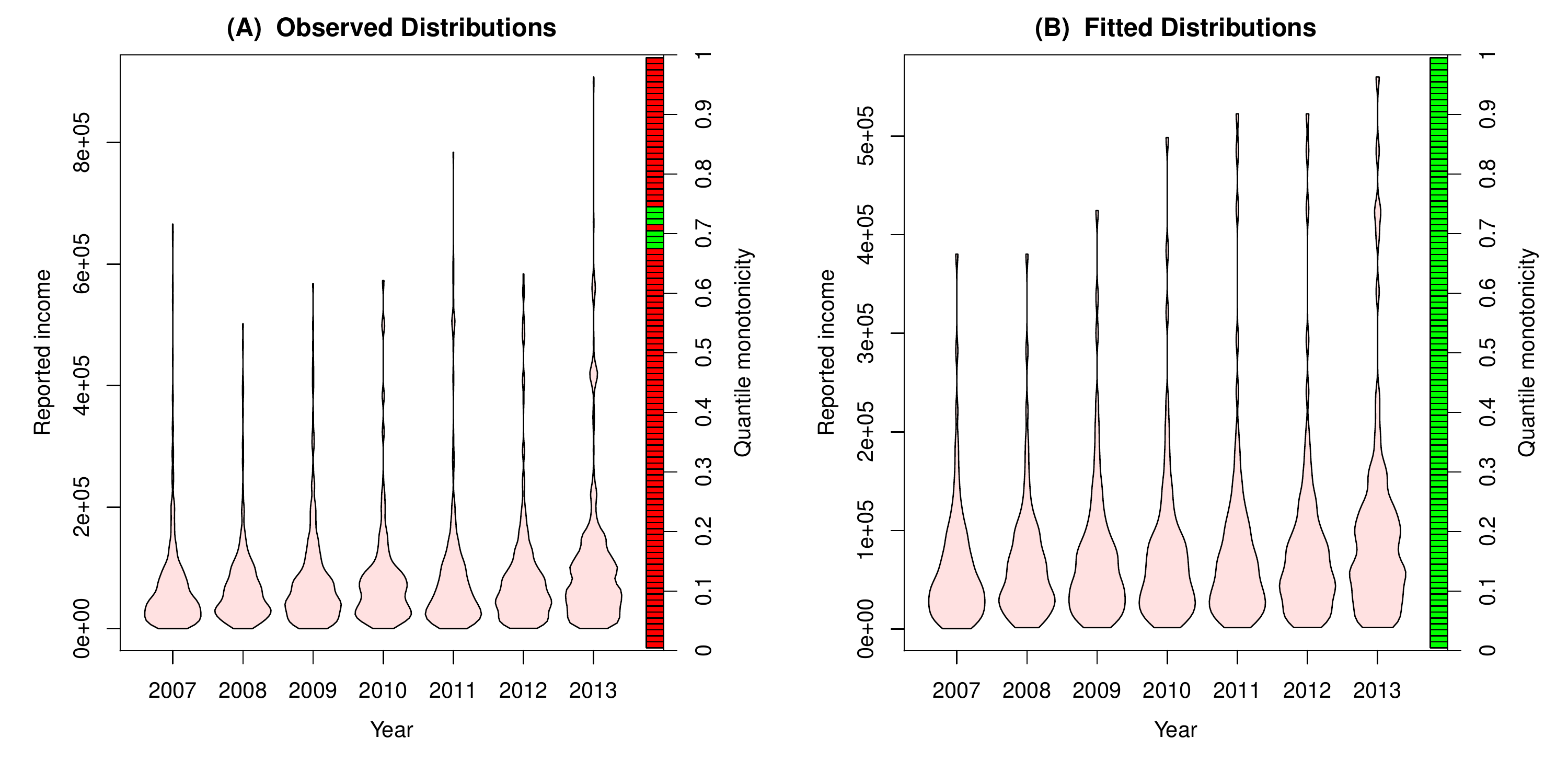}
\caption{Distributions of reported income of individuals in the ``other information services'' industry. (A) kernel density estimates applied to the ACS survey results from each year (B) corresponding TRENDS fitted distributions.}
\label{informationincomes}
\end{figure}

Table \ref{IPUMStop20} lists the industries which according to TRENDS are subject to the largest trending temporal effects in income distribution over this post-recession period.  The table contains numerous industries from the business/financial and manufacturing sectors, which were known to be particularly affected by the recession.  Interestingly, many industries from the energy sector are also included in the table\footnote{Reflecting the enactment of the Energy Independence and Security Act of 2007, which sought to move the U.S. toward greater energy efficiency and reduce reliance on imported oil.}.  The other industries in which income distributions were subject to the largest  temporal progression effects are predominantly technology-related, representing the continued growth in incomes in this sector, which has been unaffected by the recession.  

Of particular note is the ``other information services'' industry (includes web search, internet publishing/broadcasting), where we observe the emergence of a distinct subgroup with reported incomes in the hundreds of thousands.  While a few of the extreme reported incomes fell from 07-08, TRENDS conservatively estimates the underlying effects as consistently increasing all quantiles rather than including this change in $\Delta$ (such extrema are highly-variable, even at our large sample size).  For reference, the average reported incomes of this industry in 2007-13 were: \$65.8k, \$66.6k, \$77.9k, \$78.7k, \$82.1k, \$84k.

\clearpage

\section{Proofs and auxiliary lemmas}
\label{sec:proofs}
\begin{singlespace}
\addtocontents{toc}{\protect\setcounter{tocdepth}{1}}

\subsection{Proof of Lemma \ref{frechet}}
\begin{proof} Given any $G^{-1} \in \mathcal{Q}$, we can define function $H : [0,1] \rightarrow \mathbb{R}$ such that $ \displaystyle G^{-1} \equiv H +  \frac{1}{N} \sum_{i=1}^N F^{-1}_{i}$.  We have:
 \begin{align*}
 & \sum_{j=1}^N \int_0^1 \left( F^{-1}_{j}(p) -  G^{-1}(p) \right)^2 \mathrm{d}p  \\
 = &  \int_0^1 \sum_{j=1}^N \left( F^{-1}_{j}(p) - H(p) - \frac{1}{N} \sum_{i=1}^N F^{-1}_{i}(p) \right)^2 \mathrm{d}p \\ 
 \ge &  \int_0^1 \sum_{j=1}^N \left( F^{-1}_{j}(p) - \frac{1}{N} \sum_{i=1}^N F^{-1}_{i}(p) \right)^2 \mathrm{d}p \\
 & \text{ regardless of the value taken by $H(p)$ for each $p \in [0, 1]$}
 \end{align*} 
 \end{proof}
 
 \subsection{Proof of Lemma \ref{trendl1}}
\begin{proof} For any  $i < j \in \{1,\dots, L\}$:
\begin{align*}
& d_{L_1}(P_i, P_j) = \int_0^1 \left| F_i^{-1}(p) -   F_j^{-1}(p) \right| dp = \int_0^1 \sum_{\ell=i+1}^j \left| F_{\ell}^{-1}(p) -   F_{\ell-1}^{-1}(p) \right| dp = \sum_{\ell = i+1}^j d_{L_1}(P_{\ell-1}, P_\ell)   \\
\end{align*} 
where the second equality follows from the fact that $ F_i^{-1}(p), F_{i+1}^{-1}(p) \dots, F_j^{-1}(p) $  \\
is assumed to be monotone for each $p$.
 \end{proof}

\subsection{Proof of Theorem \ref{tfoptimal}}
\begin{proof}  We have: 
 \begin{align*}
 & \argmin_{G_1^{-1}, \dots,G_L^{-1}}  \bigg\{ \sum_{\ell = 1}^L  \sum_{i \in I_\ell} w_i  \sum_{k=1}^{P-1} \left( \widehat{F}^{-1}_i (p_k) - G^{-1}_\ell(p_k) \right)^2 \left[ \frac{p_{k+1} - p_{k-1}}{2} \right] \bigg\}  \\
 & \hspace*{60mm} \ \text{ where $G_1,\dots, G_L$ follow a trend} 
 \\
\equiv & \argmin_{v^{(1)}, \dots, v^{(L)}} \bigg\{ \sum_{k=1}^{P-1}  \left( \frac{p_{k+1} - p_{k-1}}{2} \right)  \sum_{\ell = 1}^L  \sum_{i \in I_\ell}  w_i  \left( \widehat{F}^{-1}_\ell (p_k) - v_{ k}^{(\ell)} \right)^2 \bigg\}  \\
&  \hspace*{40mm}  \ \text{ for }  v^{(\ell)} \in \mathbb{R}^{P-1} \text{ with entry }  v^{(\ell)}_k \text{ at $k$th index} \\
&  \text{ s.t.\ } \forall \ k < k' \in \{1,\dots, P-1\} : \left\{
     \begin{array}{l}
        \forall \ell : \ v_{k}^{(\ell)} < v_{ k'}^{(\ell)} \ \ \text{ since } G^{-1}_\ell \text{ must be a valid quantile function} \\
        v_{k}^{(1)},\dots,   v_{k}^{(L)} \text{ is a monotone sequence whose direction $ = \delta[k]$}
     \end{array}
   \right. 
 \\
 &\hspace*{61mm}  \text{ for one of the $\delta$ constructed in Step 6 or 8 of the procedure. }
 \end{align*}
 This is because the set of all $\delta$ considered by the TF algorithm contains every possible increasing/decreasing configuration \\ (mappings from $k \in  \{1,\dots, P-1\} \rightarrow \{ \text{``nonincreasing'', ``nondecreasing''}\}$) whose corresponding quantile-sequence satisfies the second condition of the trend definition.
 \begin{align*} 
  = & \argmin_{v^{(1)}, \dots, v^{(L)}} \bigg\{ \sum_{k=1}^{P-1}   \left( \frac{p_{k+1} - p_{k-1}}{2} \right)    \sum_{\ell = 1}^L   w^*_\ell  \left( \overline{ \widehat{F}}^{-1}_\ell (p_k) - v_{ k}^{(\ell)} \right)^2 \bigg\} \numberthis \label{apopt}\\
&  \hspace*{5mm} \text{ s.t.\ } \forall \ k < k' \in \{1,\dots, P-1\} : \left\{
     \begin{array}{l}
        \forall \ell : \ v_{k}^{(\ell)} < v_{ k'}^{(\ell)} \ \ \text{ since } G^{-1}_\ell \text{ must be a valid quantile function} \\
        v_{k}^{(1)},\dots,   v_{k}^{(L)} \text{ is a monotone sequence whose direction $ = \delta[k]$} 
     \end{array}
   \right. \\
 & \text{where we defined } \ w^*_\ell := \sum_{i \in I_\ell} w_i \ , \  \overline{ \widehat{F}}^{-1}_\ell (p) := \frac{1}{w^*_\ell} \sum_{i \in I_\ell} w_i \hspace*{0.4mm} \widehat{F}_i^{-1}(p_k) \\
\end{align*}
We will now show that for any $\delta$ constructed in Step 6 or 8, the corresponding $y_\ell$ produced by the AlternatingProjections algorithm are the optimal valid quantile-functions if we impose the additional constraint that for any $k$, the $p_k$th quantile-sequence must be increasing/decreasing as specified by $\delta[k]$.   Establishing this fact completes the proof because the trends-condition is simply the union of $2P$ such constraints, each of which is tested by the TF procedure.  Therefore, one of corresponding $y_1, \dots, y_L$ sequences must be the global minimum.  

Having fixed an increasing/decreasing configuration $\delta$, let $\mathcal{H}$ denote the Hilbert space of all $L \times (P-1)$ matrices, and  $\mathcal{X}$ be the vector-space of all sequences (a.k.a.\  $L \times (P-1)$ matrices) $[v^{(1)}, \dots, v^{(L)}]$ s.t. $\forall \ell \in \{1,\dots,L\}, k \in \{1,\dots, P-1\}: $ $v^{(\ell)} \in \mathbb{R}^{P-1}$  and $v^{(\ell)}_1 , \dots , v^{(\ell)}_{P-1}$ is a nondecreasing sequence.  Similarly, define $\mathcal{Y}$ to be the vector-space of all sequences $ [v^{(1)}, \dots, v^{(L)}]$ s.t. $\forall \ell , k: $ $v^{(\ell)} \in \mathbb{R}^{P-1}$  and $v^{(1)}_k , \dots , v^{(L)}_{k}$ is a monotone sequence which is increasing if and only if $\delta[k]$ specifies it.  Finally, we also define the following metric over these sequences
\begin{equation}
d_W \left( [v^{(1)}, \dots, v^{(L)}] , [w^{(1)}, \dots, w^{(L)}]  \right) = \sum_{k=1}^{P-1}   \left( \frac{p_{k+1} - p_{k-1}}{2} \right)   \sum_{\ell = 1}^L w^*_\ell  \left(v^{(\ell)}_k - w^{(\ell)}_k \right)^2
\end{equation}
Lemmas \ref{xyproj} and \ref{yxproj} show that our AlternatingProjections algorithm is equivalent to Dykstra's method of alternating projections  \citepsi{Boyle1986} between $\mathcal{X}$ and $\mathcal{Y}$ under metric $d_W$.  \\
Furthermore, both $\mathcal{X}$ and $\mathcal{Y}$ are closed and convex, and the initial point (i.e.\ sequence) $\left[x^{(1)}, \dots, x^{(L)} \right] $ must lie in $\mathcal{X}$ because $\forall \ell, k: $ the TF algorithm initializes  $x^{(\ell)}$ as a (weighted) average of valid quantile-functions (assuming the quantile-estimators do not produce invalid quantile-functions), and thus itself must be nondecreasing in $k$. \\
 Therefore, we can apply the celebrated result stated in \citesi{Combettes2011, Boyle1986} which implies that Dykstra's  algorithm must converge to the projection of the initial-sequence onto ${\mathcal{X}} \cap {\mathcal{Y}}$.   \\
By construction, this projection (under metric $d_W$) exactly corresponds to the solution of the constrained optimization in (\ref{numericaltrends}) under the additional constraint imposed by $\delta$.
\end{proof} 

\begin{lem}[\citesi{DeLeeuw1977}]
Given weights $w_1, \dots, w_N \ge 0$ and pairs $(\ell_1, y_1),\dots, (\ell_N, y_N)$ where each $\ell \in \{1,\dots, L\}$ appears at least once, the fitted values $\widehat{y}_1,\dots, \widehat{y}_L$ produced by tertiary-variant of PAVA are guaranteed to be the best-fitting nondecreasing sequence in the least-squares sense, i.e.\
\begin{equation*} \vspace*{-3mm}
\widehat{y}_1,\dots, \widehat{y}_L = \arg \min_{z_1\le \dots \le z_L} \ \sum_{\ell=1}^L \sum_{i \in I_\ell} w_i (z_\ell - y_i)^2   
\end{equation*}
\label{pavafact}
\end{lem} 
 
\begin{lem} Recall the definitions from the TF algorithm and the proof of Theorem \ref{tfoptimal}.  Given any $[x^{(1)}, \dots, x^{(L)}] \in \mathcal{X}$, its projection onto $\mathcal{Y}$ under metric $d_W$, $[y^{(1)}, \dots, y^{(L)}]$, may be computed  $\forall k \in \{1, \dots, P-1\}$ as
$$  y^{(1)}_k, \dots, y^{(L)}_k  = {\normalfont \textbf{PAVA}}\left( (x^{(1)}_k, w*_1), \dots, (x^{(L)}_k , w^*_L) ;  \delta[k] \right) 
$$
\label{xyproj}
\end{lem}
\begin{proof}[Proof of Lemma \ref{xyproj}]
Choose any $[z^{(1)}, \dots, z^{(L)}] \in \mathcal{Y}$.  By consequence of Lemma \ref{pavafact} 
\begin{align*}
& {\normalfont \textbf{PAVA}}\left( (x^{(1)}_k, w^*_1), \dots, (x^{(L)}_k , w^*_L) ; \delta[k] \right) \\
= & \ \argmin_{\text{monotone } \lambda_1,\dots, \lambda_L}\left\{ \sum_{\ell = 1}^L w^*_\ell \left( x^{(\ell)}_k - \lambda_\ell \right)^2 \right\} \ \ \text{ where the $\lambda_\ell$ are only increasing if specified by $\delta[k]$} \\
\implies & \sum_{\ell = 1}^L w^*_\ell \left(y^{(\ell)}_k - x^{(\ell)}_k  \right)^2 \le  \sum_{\ell = 1}^L w^*_\ell \left( z^{(\ell)}_k - x^{(\ell)}_k \right)^2 \ \ \forall k \\
& \hspace*{40mm} \ \text{ since } z^{(1)}_k, \dots, z^{(L)}_k \text{ have monotonicity specified by $\delta$} \\
\implies & \sum_{k=1}^{P-1} \left( \frac{p_{k+1} - p_{k-1}}{2} \right)  \sum_{\ell = 1}^L w^*_\ell \left( y^{(\ell)}_k - x^{(\ell)}_k \right)^2 \le  \sum_{k=1}^{P-1} \left( \frac{p_{k+1} - p_{k-1}}{2} \right)  \sum_{\ell = 1}^L w^*_\ell \left( z^{(\ell)}_k - x^{(\ell)}_k\right)^2 
\end{align*}
\end{proof}

\begin{lem} Recall the definitions from the TF algorithm and the proof of Theorem \ref{tfoptimal}.  Given any $[y^{(1)}, \dots, y^{(L)}] \in \mathcal{Y}$, its projection onto $\mathcal{X}$ under metric $d_W$, $[x^{(1)}, \dots, x^{(L)}]$, may be computed $\forall \ell \in \{1, \dots, L\}$ as
$$  x^{(\ell)}_1, \dots, x^{(\ell)}_{P-1}  = {\normalfont \textbf{PAVA}}\left( \left(y^{(\ell)}_1, \frac{p_{2} - p_0}{2} \right), \dots, \left(y^{(\ell)}_{P-1} , \frac{p_P - p_{P-2}}{2} \right) ; \text{``nondecreasing''} \right)
$$
\label{yxproj}
\end{lem}
\begin{proof}[Proof of Lemma \ref{yxproj}]
Choose any $[z^{(1)}, \dots, z^{(L)}] \in \mathcal{X}$.
By Lemma \ref{pavafact} 
\begin{align*}
& {\normalfont \textbf{PAVA}}\left( \left(y^{(\ell)}_1,  \frac{p_{2} - p_0}{2} \right), \dots, \left(y^{(\ell)}_{P-1} , \frac{p_P - p_{P-2}}{2} \right) ;  \text{``nondecreasing''}  \right) \\
& = \ \argmin_{\lambda_1 \le \dots \le \lambda_{P-1}} \left\{ \sum_{k = 1}^{P-1}  \left( \frac{p_{k+1} - p_{k-1}}{2} \right)   \left( y^{(\ell)}_k - \lambda_k \right)^2 \right\} \ \ \text{ for each } \ell \\
\implies & \sum_{k = 1}^{P-1}  \left( \frac{p_{k+1} - p_{k-1}}{2} \right)   \left(x^{(\ell)}_k -  y^{(\ell)}_k \right)^2 \le  \sum_{k = 1}^{P-1}  \left( \frac{p_{k+1} - p_{k-1}}{2} \right)   \left(z^{(\ell)}_k -  y^{(\ell)}_k \right)^2 \ \ \forall \ell \\
& \hspace*{5mm} \text{ since } [z^{(1)}, \dots, z^{(L)}] \in \mathcal{X} \implies \ \ \forall \ell : z^{(\ell)}_1 \le \dots \le z^{(\ell)}_{P-1} \\
\implies & \sum_{k=1}^{P-1} \left( \frac{p_{k+1} - p_{k-1}}{2} \right)  \sum_{\ell = 1}^L w^*_\ell \left( x^{(\ell)}_k - y^{(\ell)}_k \right)^2 \le  \sum_{k=1}^{P-1} \left( \frac{p_{k+1} - p_{k-1}}{2} \right)  \sum_{\ell = 1}^L w^*_\ell \left( x^{(\ell)}_k - z^{(\ell)}_k\right)^2 
\end{align*}
\end{proof}

\subsection{Proof of Theorem \ref{consistency}}
\begin{proof}
Recalling that $G^{-1}(p)$ denotes the $p$th quantile of  $Q_\ell \equiv f(\ell)$, we also define: 
\begin{equation} \widebar{F}_\ell^{-1}(p) := \frac{1}{N_\ell} \displaystyle \sum_{i \in I_\ell} F^{-1}_i(p)
\label{meanqf}
\end{equation}
By a standard application of the Chernoff bound \citepsi{Vershynin2012, boucheron2013}:
$$ \Pr \left( \left| \widebar{F}^{-1}(p) - G_\ell^{-1}(p) \right| > \eta  \right) =  \Pr \left(\left| \frac{1}{N_\ell} \sum_{i \in I_\ell}\mathcal{E}_i(p) \right| > \eta   \right) \le 2 \exp \left( - \frac{\eta^2 N_\ell}{2 \sigma^2}  \right) \ \forall \eta > 0 
$$
Recall that we compute the Wasserstein integral using $P - 1$ equally-spaced quantiles and the midpoint approximation, so
$$ d \left(\widebar{F}^{-1}_\ell , G^{-1}_\ell \right)^2 \approx d_W \left(\widebar{F}^{-1}_\ell , G^{-1}_\ell \right)^2 = \sum_{k=1}^{P-1} \frac{1}{P} \left(\widebar{F}^{-1}_\ell (k/P) - G^{-1}_\ell (k/P)  \right)^2 
$$
\begin{align*}
\Pr \left( \sum_{\ell = 1}^L d_W \left( \widebar{F}^{-1}_\ell , G^{-1}_\ell \right)^2 > \eta \right)  & \le     \sum_{\ell = 1}^{L} \sum_{k =1}^{P-1}  \Pr \left( \frac{1}{P} \left( \widebar{F}^{-1}_\ell (k/P) - G^{-1}_\ell (k/P) \right)^2 > \frac{\eta}{P L} \right) \\
& \hspace*{60mm} \text{ by a union-bound} \\
& = L \cdot P \cdot \Pr \left( \left| \widebar{F}^{-1}_\ell (k/P) - G^{-1}_\ell (k/P) \right| > \sqrt{\frac{\eta}{L}} \right)   \\ 
&   \le        2 P L \exp\left(-\frac{\eta N_\ell}{2 \sigma^2 L}  \right) \numberthis \label{etaterm}
\end{align*}
Note that $\widehat{G}^{-1}_1, \dots, \widehat{G}^{-1}_L$ form the best trending approximation to the $F^{-1}_i$ by Theorem \ref{tfoptimal}, and since $G^{-1}_1, \dots, G^{-1}_L$ are valid quantile functions which also follow a trend, this implies:
\begin{align*}
& \sum_{\ell = 1}^L \sum_{i \in I_\ell} d_W \left(F^{-1}_i , \widehat{G}^{-1}_\ell \right)^2 \le \sum_{\ell = 1}^L  \sum_{i \in I_\ell}  d_W \left(F^{-1}_i , G^{-1}_\ell \right)^2 \\ 
\Rightarrow & \sum_{\ell = 1}^L  d_W \left(\widebar{F}^{-1}_\ell , \widehat{G}^{-1}_\ell \right)^2 \le \sum_{\ell = 1}^L  d_W \left(\widebar{F}^{-1}_\ell , G^{-1}_\ell \right)^2 \ \ \text{ by Lemma \ref{frechet}} \\ 
\Rightarrow & \forall \ell: \  d_W \left(\widebar{F}^{-1}_\ell , \widehat{G}^{-1}_\ell \right)^2 \le \sum_{\ell = 1}^L  d_W \left(\widebar{F}^{-1}_\ell , G^{-1}_\ell \right)^2  
\end{align*}
Thus, by the triangle-inequality: 
$$ d_W \left( \widehat{G}^{-1}_\ell, G^{-1}_\ell \right)      \le d_W \left(\widebar{F}^{-1}_\ell , G^{-1}_\ell \right) + d_W \left(\widebar{F}^{-1}_\ell , \widehat{G}^{-1}_\ell \right)     \le        2 \left[ \sum_{\ell = 1}^L  d_W \left(\widebar{F}^{-1}_\ell , G^{-1}_\ell \right)^2 \right]^{1/2} \ \forall \ell \\
$$
which implies $ \ \forall \epsilon > 0$ we can combine this result with (\ref{etaterm}) setting $\eta := \epsilon^2 / 4$  to get:
\begin{align*}
& \Pr \left( \exists \ell: d_W(\widehat{G}^{-1}_\ell, G^{-1}_\ell) > \epsilon \right) \le \Pr \left( \sum_{\ell = 1}^L  d_W \left(\widebar{F}^{-1}_\ell , G^{-1}_\ell \right)^2  > \frac{\epsilon^2}{4} \right) \le  2 P L \exp\left(-\frac{\epsilon^2 N_\ell}{8 \sigma^2 L}  \right)
\end{align*}
\end{proof}

\subsection{Proof of Theorem \ref{badfinitesample}}
\begin{proof} We proceed similarly as in the proof of Theorem \ref{consistency}.  Defining 
\begin{equation} \displaystyle \overline{\widehat{F}}^{-1}_\ell(p) := \frac{1}{N_\ell} \sum_{i \in I_\ell} \widehat{F}^{-1}_i (p)
\label{meanestimatedqf} 
\end{equation}
by Theorem \ref{numericaltrends} and Lemma \ref{frechet}, we have:
\begin{align*} & \sum_{\ell = 1}^L d_W \left(\widehat{G}^{-1}_\ell , \overline{\widehat{F}}^{-1}_\ell \right)^2 \le \sum_{\ell = 1}^L d_W \left(G^{-1}_\ell , \overline{\widehat{F}}^{-1}_\ell \right)^2 \\
\Rightarrow & \ \ d_W \left(\widehat{G}^{-1}_\ell , \overline{\widehat{F}}^{-1}_\ell \right)^2 \le \sum_{\ell = 1}^L d_W \left(G^{-1}_\ell , \overline{\widehat{F}}^{-1}_\ell \right)^2 \ \ \ \forall \ell
\end{align*}
since $G^{-1}_1,\dots,G^{-1}_L$ are valid quantile functions which follow a trend.   Thus: 
\begin{align*} \forall \ell: \ d_W \left(\widehat{G}^{-1}_\ell, G^{-1}_\ell \right) &      \le        d_W \left(\widehat{G}^{-1}_\ell, \overline{\widehat{F}}^{-1}_\ell \right) + d_W \left(\overline{\widehat{F}}^{-1}_\ell, G^{-1}_\ell \right) \ \ \text{ by the triangle-inequality}  \\ 
& \le 2 \left[ \sum_{\ell = 1}^L d_W \left(\overline{\widehat{F}}^{-1}_\ell, G^{-1}_\ell \right)^2  \right]^{1/2} \\
& \le 2 \left[ \sum_{\ell = 1}^L \left( d_W \left( \widebar{F}^{-1}_\ell, G^{-1}_\ell \right)+ d_W \left( \overline{\widehat{F}}^{-1}_\ell , \widebar{F}^{-1}_\ell \right) \right)^2   \right]^{1/2} \ \ \text{ by the triangle-inequality}   \\
& \le 2\sqrt{2}  \left[ \sum_{\ell = 1}^L d_W \left( \widebar{F}^{-1}_\ell, G^{-1}_\ell \right)^2+ \sum_{\ell = 1}^L d_W \left( \overline{\widehat{F}}^{-1}_\ell , \widebar{F}^{-1}_\ell \right)^2   \right]^{1/2} \ \text{ by Cauchy-Schwartz}   \\
\end{align*}
Therefore $\forall \epsilon > 0$: 
\begin{align*} 
& \Pr \left( \exists \ell : d_W \left( \widehat{G}^{-1}_\ell, G^{-1}_\ell \right) > \epsilon  \right) \le            \Pr \left( \sum_{\ell = 1}^L d_W \left( \widebar{F}^{-1}_\ell, G^{-1}_\ell \right)^2+ \sum_{\ell = 1}^L d_W \left( \overline{\widehat{F}}^{-1}_\ell , \widebar{F}^{-1}_\ell \right)^2 > \frac{\epsilon^2}{8}  \right) \\
&  \le            \Pr \left( \sum_{\ell = 1}^L d_W \left( \widebar{F}^{-1}_\ell, G^{-1}_\ell \right)^2 > \frac{\epsilon^2}{16} \right)    +     \Pr \left( \sum_{\ell = 1}^L d_W \left( \overline{\widehat{F}}^{-1}_\ell , \widebar{F}^{-1}_\ell \right)^2 > \frac{\epsilon^2}{16}  \right) \ \ \text{ by the union-bound}
\end{align*}
and we can use (\ref{etaterm}) to bound the first summand, resulting in the following bound \\
\begin{equation}
\Pr \left( \exists \ell : d_W \left( \widehat{G}^{-1}_\ell, G^{-1}_\ell \right) > \epsilon  \right) \le  2 PL \exp \left( \frac{- \epsilon^2 N_\ell }{ 32 \sigma^2 L } \right) +       \Pr \left(  \sum_{\ell = 1}^L d_W \left( \overline{\widehat{F}}^{-1}_\ell , \widebar{F}^{-1}_\ell \right)^2 > \frac{\epsilon^2}{16}  \right) \ \numberthis  \label{twoterm} 
\end{equation}
Finally, Lemma \ref{simpleboundlemma} implies:
$$  \Pr \left(  \sum_{\ell = 1}^L d_W \left( \overline{\widehat{F}}^{-1}_\ell , \widebar{F}^{-1}_\ell \right)^2 > \frac{\epsilon^2}{16}  \right)       \le       2 N_\ell P L \exp \left( -2 n R \left( \frac{\epsilon}{4\sqrt{L}} \right)^2  \right)
$$
which produces the desired bound when combined with (\ref{twoterm}).  
\end{proof}

\subsection{Proof of Theorem \ref{boundedfinitesample}}
\begin{proof}
By Lemma \ref{boundedlemma}, (A.\ref{assumptionbounded}) $\Rightarrow$ (A.\ref{assumptiongoodquantileestimator}), so we only need to show the result assuming (A.\ref{assumptiongoodquantileestimator}) holds.  Lemma \ref{boundedtermlemma} then implies:
$$ \Pr \left(  \sum_{\ell = 1}^L d_W \left( \overline{\widehat{F}}^{-1}_\ell , \widebar{F}^{-1}_\ell \right)^2 > \frac{\epsilon^2}{16}  \right) \le 2 P \exp \left(   -  \frac{c^2}{8} N_\ell \hspace*{0.6mm}  n \epsilon^2    \right)
$$
Note that the bound in (\ref{twoterm}) only requires the assumptions from Theorem \ref{consistency}, so we can combine it with the above expression to obtain the desired bound.
\end{proof}

\subsection{Proof of Theorem \ref{bestfinitesample}}
\begin{proof}
\begin{align*}
\text{Consider  } \  \ & \Pr \left( \widehat{F}_i^{-1}(k/P) - F^{-1}_i(k/P) > \epsilon \right) \\
 = & \Pr \left( \widehat{F}_i \left( F^{-1}_i(k/P) +\epsilon \right) \le \frac{k}{P} \right) \\
 = & \Pr \left( \sum_{j=1}^n \mathds{1} \left[X_{i,j} \le F_i^{-1}(k/P) +\epsilon \right] \le \frac{n k}{P} \right) \numberthis \label{hoeffdingbound}
\end{align*}
This is the CDF evaluated at $\widetilde{x} := \frac{n k}{P}$ of a binomial random variable with success probability $\widetilde{p} := F_i \left( F_i^{-1}(k/P) +\epsilon \right)$ in $n$ trials. \\
Now assume  $\epsilon + F^{-1}_i (k/P) \ge B > 0$, which implies $n \widetilde{p} \ge \widetilde{x}$. \\
Letting $D(\alpha \mid \mid \beta)$ denote the relative entropy between the Bernoulli($\alpha$) and Bernoulli($\beta$) distributions, we can thus apply a tail-inequality for the binomial CDF which \citesi{Arratia1989} derived from the Chernoff bound to upper-bound (\ref{hoeffdingbound}) by
\begin{flalign*}
\le & \exp \left( -n D\left( \frac{ \widetilde{x}}{n}  \mid \mid \widetilde{p}  \right) \right) \\
= & \exp \left( -n \left[  \frac{ \widetilde{x}}{n} \log \left( \frac{ \widetilde{x} / n }{ \widetilde{p}} \right)  + \left( 1 - \frac{ \widetilde{x}}{n} \right) \log \left(\frac{1 - \widetilde{x} / n }{1 - \widetilde{p}} \right) \right]  \right) \\
= & \exp \left( -n \left[ \frac{k}{P} \log \left( \frac{ k /P }{ F_i \left( F_i^{-1}(k/P) +\epsilon \right)} \right)  + \left( 1 - \frac{ k}{P} \right) \log \left(\frac{1 - k / P }{1 - F_i \left( F_i^{-1}(k/P) +\epsilon \right)} \right) \right]  \right) \\
\le & \exp \left( -n  \left[ \frac{k}{P} \log \left( \frac{ k} {P } \right) + \left( 1 - \frac{ k}{P} \right) \log \left(\frac{1 - k / P }{1 - F_i \left( F_i^{-1}(k/P) +\epsilon \right)} \right) \right] \right)  \ \ \text{ since } F_i(\cdot) \le 1 \\
= & e^{-n C(k)}    \cdot            \exp \left ( n \left( 1 - \frac{ k}{P} \right) \log \left(1 - F_i \left( F_i^{-1}(k/P) +\epsilon \right) \right) \right)  \\
& \text{ where } C(k) := \frac{k}{P} \log \left( \frac{ k} {P } \right)  +    \left( 1 - \frac{ k}{P} \right) \log \left(1 - \frac{k} { P} \right) \ge -1 
&
\end{flalign*}
\begin{flalign*}
\le & e^{n}    \cdot            \exp \left ( n \left( 1 - \frac{ k}{P} \right) \log \left(1 - F_i \left( F_i^{-1}(k/P) +\epsilon \right) \right) \right)  \\
&  \text{ since the fact } \log x \ge \frac{x-1}{x} \ \ \forall x > 0 \text{ implies } C(k) \ge -1\ \ \forall k \in \{1,\dots, P-1 \} \\
\le &  e^{-n }    \cdot            \exp \left ( n \left( 1 - \frac{ k}{P} \right) \log \left(1 - z \right) \right) \ \  \ \text{ where } z := 1 - \exp\left(-a (F_i^{-1}(k/P) +\epsilon - B + b)^2 \right)  \\
& \text{ because } 1 - k/P > 0 \text{ and by (A.\ref{assumptionbest}): }   F_i \left(F_i^{-1}(k/P) +\epsilon \right) \ge z \ \\
& \text{ since we've assumed } F_i^{-1}(k/P) +\epsilon \ge B
&
\end{flalign*}
\begin{flalign*}
= &  e^{-n }   \cdot            \exp \left (- 2an \left( 1 - \frac{ k}{P} \right)  \left( F_i^{-1}(k/P) +\epsilon - B + b \right)^2 \right) 
\\
\le &  e^{-n}    \cdot            \exp \left (- 2an \left( 1 - \frac{ k}{P} \right) \frac{\min \left\{ b^2, \left( B -F_i^{-1}(k/P) \right)^2 \right\} }{(B - F_i^{-1}(k/P))^2}  \epsilon^2 \right)  \\
& \text{ because }  \epsilon \ge  B - F_i^{-1}(k/P)  \text{ implies } \\
&  \hspace*{20mm} \ \frac{\min \left\{ b^2, \left( B -F_i^{-1}(k/P) \right)^2 \right\} \hspace*{0.5mm} \epsilon^2}{(B - F_i^{-1}(k/P))^2}  \le \left( F_i^{-1}(k/P) +\epsilon - B + b \right)^2 
&
\end{flalign*}
\begin{flalign*}
 = &   \exp \left (-n \left[ 2a \left( 1 - \frac{ k}{P} \right) \frac{\min \left\{ b^2, \left( B -F_i^{-1}(k/P) \right)^2 \right\} }{(B - F_i^{-1}(k/P))^2}  \epsilon^2 - 1 \right] \right) \\
 \le & \exp \left (-n \left( \frac{ 2a \left( 1 - \frac{ k}{P} \right) \min \left\{ b^2, \left( B -F_i^{-1}(k/P) \right)^2 \right\}  -  1}{(B - F_i^{-1}(k/P))^2} \right)  \epsilon^2 \right) \ \ \\
 & \hspace*{50mm} \text{ since we assumed }  \epsilon \ge  B - F_i^{-1}(k/P)    \\
  \le & \exp \left (-n \left( \frac{ 2a \left( 1 - \frac{ k}{P} \right) b^2 -  1}{4B^2} \right)  \epsilon^2 \right) \ \ \  \text{ because by (A.\ref{assumptionbest}) and (A.\ref{assumptionbest3}): } \\
  & \hspace*{65mm}  - F_i^{-1}(k/P) \le B \text{ and } 0 < b \le B  
  \end{flalign*}
  And finally, we can use the fact that $k \le P-1$ to obtain the following bound
  \begin{equation}
\Pr \left( \widehat{F}_i^{-1}(k/P) - F^{-1}_i(k/P) > \epsilon \right) \le  \exp \left (-n \left( \frac{ 2a b^2 -  1}{4PB^2} \right)  \epsilon^2 \right)     
 \numberthis \label{bigepsbound}  
\end{equation}
Following the proof of Lemma \ref{boundedlemma}, one can show that (A.\ref{assumptionbest}) implies
\begin{equation}
\Pr \left( \widehat{F}_i^{-1}(k/P) - F^{-1}_i(k/P) > \epsilon \right)  \le \exp (- 2 n c^2 \epsilon^2) \ \text{ if } 0 < \epsilon < B -  F_i^{-1}(k/P) 
\label{smallepsbound}
\end{equation}
Combining (\ref{smallepsbound}) with (\ref{bigepsbound}), we thus have 
\begin{equation*}
\Pr \left( \widehat{F}_i^{-1}(k/P) - F^{-1}_i(k/P) > \epsilon \right)  \le \exp \left( -n r \epsilon^2  \right) \ \ \forall \epsilon > 0
\end{equation*}
where $r := \min \left\{ 2c^2 \ , \  \frac{ 2a b^2 -  1}{4PB^2} \right\} > 0$ by (A.\ref{assumptionbest2}).  \\
\\
One can show by an identical argument that 
\begin{equation*}
\Pr \left(F^{-1}_i(k/P) -  \widehat{F}_i^{-1}(k/P) > \epsilon \right)  \le \exp \left( -n r \epsilon^2  \right) \ \ \forall \epsilon > 0
\end{equation*}
and therefore
\begin{equation}
\Pr \left( \left| \widehat{F}_i^{-1}(k/P) - F^{-1}_i(k/P) \right|  > \epsilon \right)  \le 2 \exp \left( -n r \epsilon^2  \right) \ \ \forall \epsilon > 0
\end{equation}
\\
\\
$ \widehat{F}_i^{-1}(k/P) - F^{-1}_i(k/P) $ is thus sub-Gaussian with parameter $\frac{1}{2nr}$ and independent of $ \widehat{F}_j^{-1}(k/P) - F^{-1}_j(k/P) \ \ \forall j \neq i $ because we assumed the simple quantile-estimator defined in (A.\ref{assumptionecdf})  is used.  Following the proof of Lemma \ref{boundedtermlemma}, $\forall \gamma > 0$:
\begin{equation}
 \Pr \left(  \sum_{\ell = 1}^L d_W \left( \overline{\widehat{F}}^{-1}_\ell , \widebar{F}^{-1}_\ell \right)^2 > \frac{\epsilon^2}{16} \right) \le 2 P \exp \left( -  \frac{r}{16}  \hspace*{0.5mm}  N_\ell \hspace*{0.8mm} n \hspace*{0.5mm}  \epsilon^2     \right)
\end{equation}
Note that the bound in (\ref{twoterm}) only requires the assumptions from Theorem \ref{consistency}, so we can combine it with the above inequality to obtain the desired bound.
\end{proof}
 \vspace*{3mm}
\begin{lem}[\citesi{Serfling1980}: Theorem 2.3.2]
For $p \in (0,1)$: if $\exists$ unique $x$ s.t.\ $F(x) = p$ and $\widehat{F}^{-1}(p)$ is estimated using $n$ i.i.d.\ samples from CDF $F_i$, then $\forall \gamma > 0$: 
$$ \Pr \left( \left| \widehat{F}^{-1}_i(p) - F^{-1}_i(p)  \right|    > \gamma  \right) \le 2 \exp \left( -2 n R(\gamma, i, p)^2  \right)
$$
where $R(\gamma, i, p) := \min \left\{  F_i \left(F^{-1}_i(p) + \gamma \right) - p \ , \ p - F_i \left(F^{-1}_i(p) - \gamma \right)  \right\}$
\label{serflinglemma}
\end{lem}
 \vspace*{3mm}
\begin{lem} Under the assumptions of Theorem \ref{badfinitesample} and definitions (\ref{rgamma}), (\ref{meanqf}),  (\ref{meanestimatedqf})

$$ \forall \gamma > 0 : \ \ \ \Pr \left(  \sum_{\ell = 1}^L d_W \left( \overline{\widehat{F}}^{-1}_\ell , \widebar{F}^{-1}_\ell \right)^2 > \gamma \right) \le 
 2 N_\ell P L \exp \left( -2 n R \left(\sqrt{\gamma / L} \right)^2  \right)
$$ \label{simpleboundlemma}
\end{lem}
\begin{proof}[Proof of Lemma \ref{simpleboundlemma}]
\begin{align*}
 & \Pr \left(  \sum_{\ell = 1}^L d_W \left( \overline{\widehat{F}}^{-1}_\ell , \widebar{F}^{-1}_\ell \right)^2 > \gamma \right) \\
= &  \Pr \left(  \sum_{\ell = 1}^L \frac{1}{N_\ell} \sum_{i \in I_\ell} \sum_{k=1}^{P-1} \frac{1}{P}  \left( \widehat{F}^{-1}_i(k/P) - {F}^{-1}_i(k/P) \right)^2 > \gamma \right)  \\
\le & N_\ell L \sum_{k=1}^{P-1} \Pr \left( \left| \widehat{F}^{-1}_i(k/P) - {F}^{-1}_i(k/P)  \right| > \sqrt{\frac{\gamma}{L}} \right) \ \ \text{ by the union-bound} \\
\le & 2 N_\ell L \sum_{k=1}^{P-1} \exp \left( - 2 n R \left( \sqrt{\gamma/ L}, i, k/P \right)^2 \right) \ \ \text{ by (A.\ref{assumptionnonzerodensity}) and Lemma \ref{serflinglemma}} \\
\le & 2 N_\ell L P \exp \left( - 2 n R \left( \sqrt{\gamma/L} \right) ^2 \right)  \ \ \text{  by definition (\ref{rgamma})} 
\end{align*}
\end{proof}
\vspace*{3mm}

\begin{lem} If we assume (A.\ref{assumptionsamplequantiles}) and (A.\ref{assumptionnonzerodensity}), then condition (A.\ref{assumptionbounded}) implies condition (A.\ref{assumptiongoodquantileestimator}).
\label{boundedlemma}  
\end{lem}
\begin{proof}[Proof of Lemma \ref{boundedlemma}]
Assume WLOG that $F^{-1}_i(k/P) \ge 0$ and note that $F^{-1}_i(k/P) \le B$ by (A.\ref{assumptionbounded}). \\
Then, by a bound established in the proof of Lemma \ref{serflinglemma} given in \citepsi{Serfling1980}, $\forall \epsilon  > 0:$
\begin{equation}  
\Pr \left( \widehat{F}^{-1}_i(k/P) - F^{-1}_i(k/P)   > \epsilon \right) \le \exp \left(- 2 n \hspace*{0.5mm} R(\epsilon, i, k/P)^2 \right) 
 \label{serflinglower}
 \end{equation} 
and 
\begin{equation}
\Pr \left( F^{-1}_i(k/P)  - \widehat{F}^{-1}_i(k/P)  > \epsilon \right) \le \exp \left(- 2 n \hspace*{0.5mm} R(\epsilon, i, k/P)^2 \right) 
 \label{serflingupper}
\end{equation}
By (A.\ref{assumptionbounded}):  $ \ f_i(x) = \frac{d}{dx} F_i(x) \ge c \ \forall x \in (-B , B) \ \ $ which implies   \\
\begin{equation}
R(\gamma, i, p) \ge c \gamma > 0 \ \text{ if } F_i^{-1}(p) \pm \gamma \in (-B , B)
\label{deriv}
\end{equation}
because recall that we defined $R(\gamma, i, p) := \min \left\{  F_i \left(F^{-1}_i(p) + \gamma \right) - p \ , \ p - F_i \left(F^{-1}_i(p) - \gamma \right)  \right\}$.  \\
Together with (\ref{deriv}), (\ref{serflinglower}) and (\ref{serflingupper}) imply
\begin{equation}  
\Pr \left( \widehat{F}^{-1}_i(k/P) - F^{-1}_i(k/P)   > \epsilon \right) \le \exp(- 2 n c^2 \epsilon^2) \ \text{ if } F_i^{-1}(k/P) +  \epsilon < B    \label{lowernonzero}
 \end{equation}
and 
\begin{equation}
\Pr \left( F^{-1}_i(k/P)  - \widehat{F}^{-1}_i(k/P)  > \epsilon \right) \le \exp(- 2 n c^2 \epsilon^2) \ \text{ if } F_i^{-1}(k/P) -  \epsilon > -B       \label{uppernonzero}
\end{equation}
Note that because $f_i(x) = 0 \ \forall x \ge B$, we have 
\begin{align*}
& \Pr \left( \widehat{F}^{-1}_i(k/P) > F^{-1}_i(k/P) + \epsilon \right) = 0 \ \ \text{ if  } \epsilon \ge B - F_i^{-1}(k/P) \\
\implies & \Pr \left( \widehat{F}^{-1}_i(k/P) - F^{-1}_i(k/P)   > \epsilon \right)  = 0 \ \ \text{ if  } \epsilon \ge B - F_i^{-1}(k/P)\numberthis \label{lowerf} 
\end{align*}
as well as 
\begin{align*} & \Pr \left( \widehat{F}^{-1}_i(k/P) < F^{-1}_i(k/P) - \epsilon \right) = 0 \ \ \text{ if  } \epsilon \ge B + F_i^{-1}(k/P) \\
\implies & \Pr \left( F^{-1}_i(k/P)  - \widehat{F}^{-1}_i(k/P)  > \epsilon \right) = 0  \ \ \text{ if  } \epsilon \ge B + F_i^{-1}(k/P) \numberthis \label{upperf} 
\end{align*}
Putting together (\ref{lowernonzero}), (\ref{uppernonzero}), (\ref{lowerf}), and (\ref{upperf}), we thus have \\
$$  \Pr \left( \widehat{F}^{-1}_i(k/P) - F^{-1}_i(k/P)   > \epsilon \right) \le  \exp(- 2 n c^2 \epsilon^2) \ \ \ \ \forall \epsilon > 0
$$ 
and 
$$
\Pr \left( F^{-1}_i(k/P) -  \widehat{F}^{-1}_i(k/P)   > \epsilon \right) \le  \exp(- 2 n c^2 \epsilon^2) \ \ \ \ \forall \epsilon > 0 
$$
\text{which implies }
$$
 \ \Pr \left( \left| F^{-1}_i(k/P) -  \widehat{F}^{-1}_i(k/P) \right|   > \epsilon \right) \le 2 \exp(- 2 n c^2 \epsilon^2) \ \ \ \ \forall \epsilon > 0 
$$
\end{proof}
\vspace*{3mm}

\begin{lem} 
Under condition (A.\ref{assumptiongoodquantileestimator}) and definitions (\ref{rgamma}), (\ref{meanqf}),  (\ref{meanestimatedqf})
$$ \forall \gamma > 0 : \ \ \ \Pr \left(  \sum_{\ell = 1}^L d_W \left( \overline{\widehat{F}}^{-1}_\ell , \widebar{F}^{-1}_\ell \right)^2 > \gamma \right) \le 
 2 P \exp \left( -2 n c^2 N_\ell \gamma  \right)
$$ 
\label{boundedtermlemma} 
\end{lem}
\begin{proof}[Proof of Lemma \ref{boundedtermlemma}]
\begin{align*}
 & \Pr \left(  \sum_{\ell = 1}^L d_W \left( \overline{\widehat{F}}^{-1}_\ell , \widebar{F}^{-1}_\ell \right)^2 > \gamma \right) \\
= &  \Pr \left( \frac{1}{L N_\ell} \sum_{\ell = 1}^L \sum_{i \in I_\ell} \sum_{k=1}^{P-1} \frac{1}{P}  \left( \widehat{F}^{-1}_i(k/P) - {F}^{-1}_i(k/P) \right)^2 > \frac{\gamma}{L} \right)  \\
\le & \sum_{k=1}^{P-1} \Pr \left( \left| \frac{1}{L N_\ell}\sum_{\ell = 1}^L \sum_{i \in I_\ell} \widehat{F}^{-1}_i(k/P) - {F}^{-1}_i(k/P)  \right| > \sqrt{\frac{\gamma}{L}} \right) \ \ \text{ by the union-bound} \\
\le & 2  \sum_{k=1}^{P-1} \exp \left( - 2 n c^2 L N_\ell \sqrt{\frac{ \gamma}{ L}}^2     \right) =  2 P \exp \left( - 2 n c^2 N_\ell  \gamma     \right)
\end{align*}
where in the last inequality, we have used the fact that (A.\ref{assumptiongoodquantileestimator}) implies the $\widehat{F}^{-1}_i(k/P) - {F}^{-1}_i(k/P)$ are independent sub-Gaussian random variables with parameter $\frac{1}{4nc^2}$, so the inequality follows from  a standard application of the Chernoff bound \citepsi{Vershynin2012, boucheron2013}.  
\end{proof}
\end{singlespace}

\clearpage
\bibliographystylesi{agsm}
\bibliographysi{TRENDS}


\end{document}